\documentclass[sigconf,balance=false]{acmart}
\usepackage{popets}

\setcopyright{popets}
\copyrightyear{YYYY}

\acmYear{YYYY}
\acmVolume{YYYY}
\acmNumber{X}
\acmDOI{XXXXXXX.XXXXXXX}
\acmISBN{}
\acmConference{Proceedings on Privacy Enhancing Technologies}
\usepackage[T1]{fontenc}
\usepackage{microtype} 

\usepackage[table]{xcolor}

\usepackage[symbol]{footmisc} 

\usepackage{amsmath,amsthm}
\usepackage{mathtools} 
\DeclareMathOperator*{\argmax}{arg\,max}

\usepackage{soul,bm} 
\usepackage{fancyvrb} 
\usepackage[most]{tcolorbox} 

\usepackage[font=small]{caption} 
\usepackage{subcaption,wrapfig} 
\usepackage{graphicx}

\usepackage{tabularx}
\newcolumntype{Y}{>{\centering\arraybackslash}X} 
\usepackage{booktabs} 
\usepackage{threeparttable} 
\usepackage{multirow} 
\usepackage{makecell} 
\usepackage{pifont}   

\usepackage{algorithmic}
\usepackage[ruled,vlined,linesnumbered]{algorithm2e} 

\usepackage{tikz}
\usetikzlibrary{fit,calc}

\colorlet{mypink}{magenta}
\colorlet{myblue}{cyan}

\usepackage{listings}
\definecolor{delim}{RGB}{20,105,176}
\definecolor{numb}{RGB}{106, 109, 32}
\definecolor{string}{rgb}{0.64,0.08,0.08}
\lstdefinelanguage{json}{
    numbers=left,
    numberstyle=\scriptsize\color{gray},
    backgroundcolor=\color{gray!5},
    showspaces=false,
    showtabs=false,
    breaklines=true,
    postbreak=\raisebox{0ex}[0ex][0ex]{\ensuremath{\color{gray}\hookrightarrow\space}},
    breakatwhitespace=true,
    basicstyle=\ttfamily\small,
    upquote=true,
    morestring=[b]",
    stringstyle=\color{string},
    literate=
     *{0}{{{\color{numb}0}}}{1}
      {1}{{{\color{numb}1}}}{1}
      {2}{{{\color{numb}2}}}{1}
      {3}{{{\color{numb}3}}}{1}
      {4}{{{\color{numb}4}}}{1}
      {5}{{{\color{numb}5}}}{1}
      {6}{{{\color{numb}6}}}{1}
      {7}{{{\color{numb}7}}}{1}
      {8}{{{\color{numb}8}}}{1}
      {9}{{{\color{numb}9}}}{1}
      {\{}{{{\color{delim}{\{}}}}{1}
      {\}}{{{\color{delim}{\}}}}}{1}
      {[}{{{\color{delim}{[}}}}{1}
      {]}{{{\color{delim}{]}}}}{1},
}

\newtheorem{definition}{\textbf{Definition}}
\newtheorem{theorem}{\textbf{Theorem}}

\newtheorem{proposition}{\textbf{Proposition}}
\newtheorem{example}{\textbf{Example}}

\newcommand{\eldp}{$\varepsilon$-LDP\xspace}

\newcommand{\mechanism}{\mathcal{M}}

\newcommand{\loss}{\mathcal{L}}

\begin{document}

\title{Quantifying Classifier Utility under Local Differential Privacy}


\author{Ye Zheng}
\affiliation{%
  \institution{Rochester Institute of Technology}
  \city{Rochester}
  \country{USA}}

\author{Yidan Hu}
\affiliation{%
  \institution{Rochester Institute of Technology}
  \city{Rochester}
  \country{USA}}



\begin{abstract}
Local differential privacy (LDP) offers rigorous, quantifiable privacy guarantees for personal data by introducing perturbations at the data source. 
Understanding how these perturbations affect classifier utility is crucial for both designers and users. 
However, a general theoretical framework for quantifying this impact is lacking and also challenging, especially for complex or black-box classifiers.

This paper presents a unified framework for theoretically quantifying classifier utility under LDP mechanisms. 
The key insight is that LDP perturbations are concentrated around the original data with a specific probability, 
allowing utility analysis to be reframed as robustness analysis within this concentrated region. 
Our framework thus connects the concentration properties of LDP mechanisms with the robustness of classifiers, 
treating LDP mechanisms as general distributional functions and classifiers as black boxes. 
This generality enables applicability to any LDP mechanism and classifier. 
A direct application of our utility quantification is guiding the selection of LDP mechanisms and privacy parameters for a given classifier. 
Notably, our analysis shows that a piecewise-based mechanism often yields better utility than alternatives in common scenarios.

Beyond the core framework, we introduce two novel refinement techniques that further improve utility quantification. 
We then present case studies illustrating utility quantification for various combinations of LDP mechanisms and classifiers. 
Results demonstrate that our theoretical quantification closely matches empirical observations, particularly when classifiers operate in lower-dimensional input spaces.
\end{abstract}


\maketitle

\section{Introduction}

\begin{figure}[t]
    \centering
    \vspace{0.5em}
    \includegraphics[width=\linewidth]{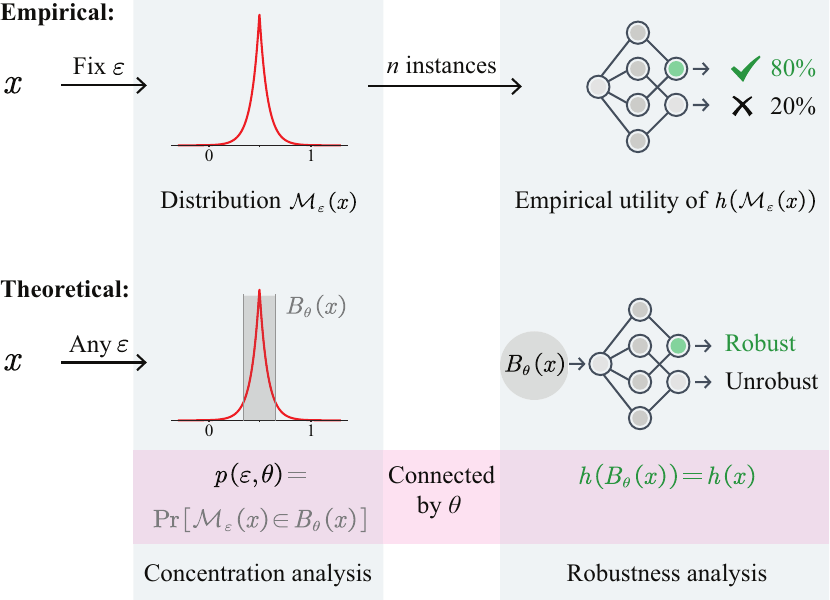}
    \caption{From empirical data utility to theoretical data utility. 
    This paper provides a theoretical quantification of classifier utility under LDP-perturbed inputs by connecting 
    the \emph{concentration analysis} of LDP mechanisms with the \emph{robustness analysis} of classifiers.}
    \label{fig:emperical_theoretical}
\end{figure}

\emph{Classifiers} are functions that map input data to class labels, playing a fundamental role
in industrial applications ranging from predictive modeling and data clustering to image recognition~\cite{doi:10.1080/01431160600746456,DBLP:journals/information/KowsariMHMBB19,DBLP:journals/csur/JainMF99}.
When deployed as services, classifiers require users to provide data inputs, which often contain sensitive information such as medical records or financial data, raising significant privacy concerns.
While users wish to benefit from classification services, they may be reluctant to reveal their actual data.
A common approach to address this concern is data perturbation, e.g. adding noise to numerical data or applying blur filters to images before 
sending them to the classifier.
Data perturbation remains a lightweight and intuitive solution for privacy-preserving classification~\cite{DBLP:conf/kdd/ZhangWZ05,DBLP:conf/icdm/ChenL05,DBLP:journals/widm/ZorarpaciO21}.

\emph{Local differential privacy} (LDP) is a widely accepted mathematical framework for privacy protection 
through controlled data perturbation.
It uses a parameter $\varepsilon$ (the privacy parameter) to quantify the level of privacy protection.
When data is processed through an LDP mechanism before being sent to a classifier, 
it receives provable privacy guarantees, meaning that any observer including the classifier cannot confidently infer the original sensitive data.
Researchers have developed numerous LDP mechanisms for various data domains~\cite{DBLP:journals/fttcs/DworkR14,DBLP:conf/nips/KairouzOV14,DBLP:conf/uss/WangBLJ17,DBLP:conf/icde/WangXYZHSS019,DBLP:conf/sigmod/Li0LLS20,DBLP:conf/uss/Gu0C0020},
each offering different perturbation strategies while maintaining formal privacy guarantees.

\emph{Data utility} is the most crucial metric when implementing LDP mechanisms.
While privacy protection is controlled by the privacy parameter $\varepsilon$,
data utility measures how well the perturbed data serves its intended application purpose.
For simple applications like summation queries, utility can be directly quantified through the theoretical mean squared error (MSE) of LDP mechanisms~\cite{DBLP:conf/uss/WangBLJ17,DBLP:conf/icde/WangXYZHSS019}.
For classification tasks, utility corresponds to classification accuracy, which cannot be analytically expressed using MSE. 
A trivial approach to evaluate classifier utility under an LDP mechanism with a given privacy parameter $\varepsilon$ is to repeatedly perturb the data and 
measure the proportion of correctly classified instances, as illustrated in Figure~\ref{fig:emperical_theoretical}.
While practical, this empirical method has significant limitations: it only applies to specific $\varepsilon$ values and perturbed data instances;
changing $\varepsilon$ requires time-consuming re-evaluation, and different instances lead to different results.
Meanwhile, it fails to establish relationships between utility and privacy parameters, making it inadequate for systematic comparison of LDP mechanisms.
Analytical utility quantification, by contrast, provides valuable guidance for classifier designers and users;
however, unlike the MSE for summation queries, we currently lack such a quantification framework for classifiers.

Quantifying the relationship between privacy and classifier utility presents significant analytical challenges.
These challenges stem from two aspects:
(i) LDP mechanisms inherently introduce perturbations across the entire data domain, potentially causing zero utility for classifiers.
(ii) Classifiers are often complex or even black-box functions, making it difficult to analyze their behavior under perturbed inputs.

This paper develops a framework to theoretically quantify classifier utility under LDP-perturbed inputs.
Our approach addresses the aforementioned challenges through two key insights:
(i) LDP mechanisms generate perturbed data that concentrates within a bounded region around the original data with a specific high probability, making extreme perturbations rare.
(ii) Within this concentration region, the utility analysis of a classifier can be reformulated as a robustness analysis problem.
Robustness refers to how well a classifier maintains correct predictions when inputs are subjected to perturbations.
We leverage established robustness analysis techniques to find the maximum permissible perturbation that preserves a classifier's output,
thereby maintaining utility. 
By combining the concentration analysis of LDP mechanisms with the robustness analysis of classifiers, 
we establish theoretical data utility quantification in relation to the privacy parameter $\varepsilon$, expressed as:

\vspace{0.5em}
\noindent\emph{Given classifier $h$, LDP mechanism $\mechanism_\varepsilon$, and raw data $x$:
with probability at least $p(\varepsilon,\theta)$, $h$ preserves its correct classification under input $\mechanism_\varepsilon(x)$.}
\vspace{0.5em}

\noindent This formulation provides a probabilistic guarantee for preserving the utility of classifier $h$ when using mechanism $\mechanism_\varepsilon$.
The parameter $\theta$ represents a concentration threshold that determines the probability guarantee $p(\varepsilon,\theta)$,
which directly depends on the privacy parameter $\varepsilon$.
Figure~\ref{fig:emperical_theoretical} illustrates this intuition.
By determining the maximum allowable perturbation $B_\theta(x)$ that a classifier $h$ can tolerate,
we can derive utility guarantees for $h$ under $\mechanism_\varepsilon$-perturbed inputs across different privacy parameters $\varepsilon$,
without needing to repeatedly sample from the mechanism.
Furthermore, if the distribution of raw data $x$ is known, our framework also provides average-case and worst-case utility quantifications.

Additionally, we introduce two refinement techniques to refine the utility quantification $p(\varepsilon,\theta)$ for high-dimensional classifiers.
First, we extend the widely used robustness notion of ``robustness radius'' to ``robustness hyperrectangle'', enabling a more precise analysis of classifier robustness.
Second, we adapt a relaxed privacy notion, \emph{Probably Approximately Correct} (PAC) privacy~\cite{DBLP:conf/crypto/XiaoD23}, as an alternative to $(\varepsilon,\delta)$-LDP.
Under this notion, we propose a novel privacy indicator and an extended Gaussian mechanism\footnote{
    It is well known that the Gaussian mechanism cannot satisfy pure (L)DP~\cite{DBLP:journals/fttcs/DworkR14};
    it must rely on a relaxed privacy notion such as $(\varepsilon,\delta)$-LDP or PAC LDP.
}
applicable for any $\varepsilon$.
The original Gaussian mechanism from Dwork~\cite{DBLP:journals/fttcs/DworkR14} is only valid for $\varepsilon \in [0, 1]$.
We improve the proof technique to extend it to $\varepsilon \in [0, \infty)$.    
These extensions allow the framework to provide more accurate utility quantification and incorporate more privacy-preserving mechanisms.

\textbf{Applications.}
The proposed utility quantification framework has direct applications in privacy-preserving classification systems:
(i) It enables systematic comparison of different LDP mechanisms for a given classifier by evaluating their probability guarantees $p(\varepsilon,\theta)$.
For a fixed $\varepsilon$, a mechanism with a larger $p(\varepsilon,\theta)$ provides stronger utility preservation.
(ii) It facilitates the selection of an appropriate privacy parameter $\varepsilon$ to meet specific utility requirements.
For example, given a utility threshold $p^*$, the framework identifies the $\varepsilon$ that satisfies $p(\varepsilon,\theta) \geq p^*$, 
achieving a precise privacy-utility balance when using the classifier.

To summarize, our contributions are as follows:
\begin{itemize}
    \item \emph{Quantification framework.} To the best of our knowledge, 
    this is the first framework that provides analytical utility quantification for classifiers under LDP-perturbed inputs. 
    This framework bridges the concentration analysis of LDP mechanisms with the robustness analysis of classifiers, enabling systematic evaluation of classifier utility.
    \item \emph{Refinement techniques.} We propose two refinement techniques to enhance utility quantification. 
    The first extends the robustness notion from ``robustness radius'' to ``robustness hyperrectangle,'' allowing for more precise robustness analysis. 
    The second adapts the PAC privacy notion, introducing a novel privacy indicator and an extended Gaussian mechanism applicable to any $\varepsilon$.
    \item \emph{Comparison and case studies.} We conduct case studies on typical classifiers, including logistic regression, random forests, and neural networks, 
    under input perturbations by various LDP mechanisms. 
    The results demonstrate that our theoretical utility quantification aligns closely with empirical observations, particularly in low-dimensional input spaces.
\end{itemize}

\textbf{Structure.} The main parts of this paper are organized as follows:
Section~\ref{sec:example} provides an illustrative example of utility quantification for a classifier under the Laplace mechanism.
Section~\ref{sec:framework} introduces the formal framework for utility quantification.
Section~\ref{sec:refined_bound} presents refinement techniques, then 
Section~\ref{sec:discussion} discusses related questions and
Section~\ref{sec:case_studies} conducts case studies.

\section{Preliminaries} \label{sec:preliminaries}

This section formulates the problem and reviews the definition of LDP as well as the concept of classifier robustness.

\subsection{Problem Formulation} 

Given a trained classifier $h$ that provides public services, users must submit their data $x$ as input to utilize the classifier.
However, users' data $x$ may contain sensitive information, and the third-party classifier $h$ may not be fully trusted to protect data privacy. 
To address this concern, users perturb their data $x$ using an LDP mechanism $\mechanism_\varepsilon$ before sending it to the classifier $h$.
The perturbation mechanism $\mechanism_\varepsilon$ acts as an interface between users and the classifier, forming a trust boundary by presenting the perturbed data $\mechanism_\varepsilon(x)$ to the classifier $h$.

While the perturbed data $\mechanism_\varepsilon(x)$ protects user privacy, it may degrade the utility of the classifier $h$, potentially leading to incorrect outputs where $h(\mechanism_\varepsilon(x)) \neq h(x)$.
This degradation is undesirable for both the classifier provider and the users.
Thus, a critical question arises:
\emph{How can classifier designers or users quantify the utility of the classifier $h$ under the perturbation of $\mechanism_\varepsilon$?}
This paper aims to answer this question by introducing a theoretical framework to quantify the utility of $h$ under the perturbation of $\mechanism_\varepsilon$ w.r.t. the privacy parameter $\varepsilon$.

\subsection{Local Differential Privacy}

\begin{definition}[\eldp~\cite{DBLP:journals/corr/DuchiWJ16}] \label{def:eldp}
    A perturbation mechanism $\mathcal{M}: \mathcal{X} \to \mathsf{Range}(\mechanism)$ satisfies $\varepsilon$-LDP
    if, for any two arbitrary $x_1$ and $x_2$, the probability ratio of producing the same $\tilde{x}$ is bounded as follows:
    \begin{equation*}
        \forall x_1, x_2 \in \mathcal{X}, \forall \tilde{x} \in \mathsf{Range}(\mechanism):\; \frac{\Pr[\mathcal{M}(x_1) = \tilde{x}]}{\Pr [\mathcal{M}(x_2) = \tilde{x}]} \leq e^{\varepsilon}.
    \end{equation*}
\end{definition} 

If $\mechanism(x)$ is continuous, the probability $\Pr[\cdot]$ is replaced by the probability density function (pdf).
Definition~\ref{def:eldp} quantifies the difficulty of distinguishing between $x_1$ and $x_2$ based on the observed $\tilde{x}$.
A smaller privacy parameter $\varepsilon \in [0, +\infty)$ indicates stronger privacy protection. 
For convenience, we use $\mechanism_\varepsilon$ to denote an $\varepsilon$-LDP mechanism $\mechanism$ when $\varepsilon$ is necessary.

Some LDP mechanisms, such as the Laplace and Gaussian mechanisms~\cite{DBLP:journals/fttcs/DworkR14}, were originally designed for query functions on raw data.
These mechanisms require the \emph{sensitivity} of the query function to determine the scale of perturbation.
The (global) sensitivity of a function $f: \mathcal{X} \to \mathsf{Range}(f)$ is defined as the maximum change in $f$ for any two data:
\begin{equation*}
    \Delta f = \max_{x_1, x_2 \in \mathcal{X}} \left|f(x_1) - f(x_2)\right|.
\end{equation*}
In the context of LDP, the privacy of the data $x$ itself must be ensured, as the perturbation is applied directly to the data.
Thus, here the query function $f$ is the identity function $f(x) = x$, whose sensitivity is the diameter of input space $\mathcal{X}$.
For classifiers, $\mathcal{X}$ is typically normalized to $\mathcal{X} = [0,1]$, resulting in a sensitivity of $1$.

\subsection{Classifier and Robustness}

\begin{definition}[Classifier]
    A classifier $h: \mathbb{R}^d \to \{1,2,\dots,K\}$ is a function that maps input data $x \in \mathbb{R}^d$ to a specific
    label $h(x) \in \{1,2,\dots,K\}$.
\end{definition}

Based on representation complexity, classifiers can be categorized as follows:
a \emph{closed-form classifier} is represented by a closed-form function, such as a linear classifier;
a \emph{non-closed-form classifier} lacks a closed-form representation, such as a neural network.
In both cases, if we know the classifier's structure and parameters, we call it a \emph{white-box classifier}; otherwise, it is a \emph{black-box classifier}.

\begin{definition}[Local robustness]
    Denote $B_{\theta}(x)$ as the $\ell_{\infty}$-norm ball centered at $x$ with radius $\theta$.\footnote{
        The $\ell_{\infty}$-norm is a common choice for measuring data perturbations and can be converted to other norms using norm inequalities.
        Here, $B_{\theta}(x)$ denotes the set of points $\tilde{x} \in \mathbb{R}^d$ where each dimension of $\tilde{x}$ is within a distance $\theta$ of the corresponding dimension of $x$, i.e. $B_{\theta}(x) \coloneq \{\tilde{x}: |\tilde{x}^{(i)} - x^{(i)}| \leq \theta\}$.
    }
    A classifier $h$ is $\theta$-locally-robust at $x$ if:
    \begin{equation*}
        \forall \tilde{x} \in B_{\theta}(x),\; h(\tilde{x}) = h(x).
    \end{equation*}
\end{definition}

Intuitively, local robustness means that the classifier $h$ is insensitive to small perturbations around the input $x$. 
The largest feasible $\theta$ is referred to as the \emph{robustness radius} of the classifier $h$ at $x$.
Robustness radius serves as a key metric for evaluating the classifier's stability to unseen data.

\section{Illustrating the Framework: An Example} \label{sec:example}

We illustrate our utility quantification framework with an intuitive example.
Specifically, we consider a classifier under the Laplace mechanism and present a utility quantification statement.

\begin{definition}[The Laplace mechanism~\cite{DBLP:journals/fttcs/DworkR14}]
    For any $x\in [0,1]$, 
    the Laplace mechanism $\mechanism_\varepsilon$ is defined as:
    $\mechanism_\varepsilon(x) = x + \xi$,
    where $\xi \sim Lap(1/\varepsilon)$ is a random variable drawn from Laplace distribution with mean $0$ and scale $1/\varepsilon$.
\end{definition}

In the Laplace mechanism, $\mechanism_\varepsilon(x)$ is unbounded due to the Lap\-lace distribution. 
Since the classifier is typically defined on a bounded domain, such as $[0,1]$, a common practice is to truncate the perturbed data to this range,
i.e. $\mechanism_\varepsilon(x)$ is truncated to $[0,1]$.\footnote{
    This truncation can be treated as a post-processing step by the users that does not leak privacy~\cite{DBLP:journals/fttcs/DworkR14,DBLP:conf/ndss/0001LLSL20,DBLP:conf/stoc/GhoshRS09}.  
    We will also examine mechanisms with bounded perturbations in Section~\ref{subsec:boundedness_ldp},
    where no truncation is needed.
    \ There are also truncated Laplace mechanisms~\cite{DBLP:journals/jcst/CroftSS22,DBLP:conf/aistats/GengDGK20}
    that directly sample from the truncated Laplace distribution.
    It is also applicable to our framework; we use the original Laplace mechanism on $[0,1]$ as an example for simplicity.
}

\subsection{Example} \label{subsec:example}

\begin{figure}[t]
    \centering
    \begin{subfigure}[b]{0.50\linewidth}
        \centering
        \includegraphics[width=0.98\textwidth]{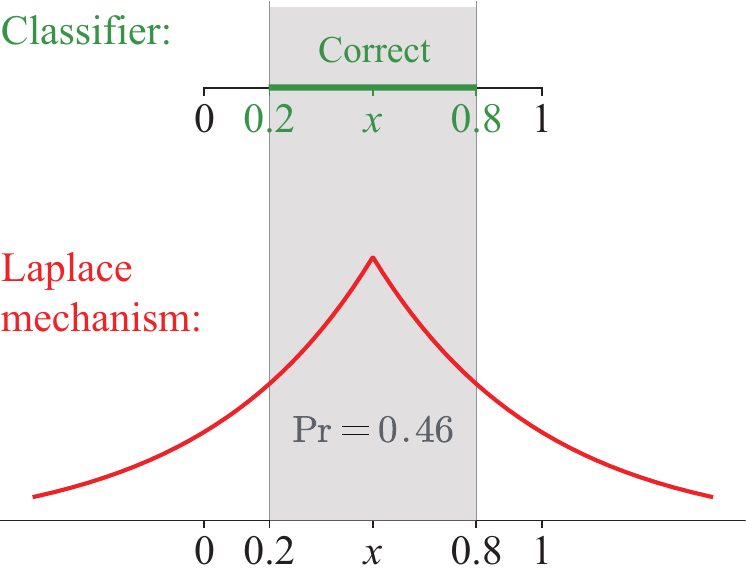}
        \caption{Example illustration.}
        \label{fig:example}
    \end{subfigure}
    \hfill
    \begin{subfigure}[b]{0.40\linewidth}
        \centering
        \includegraphics[width=0.98\textwidth]{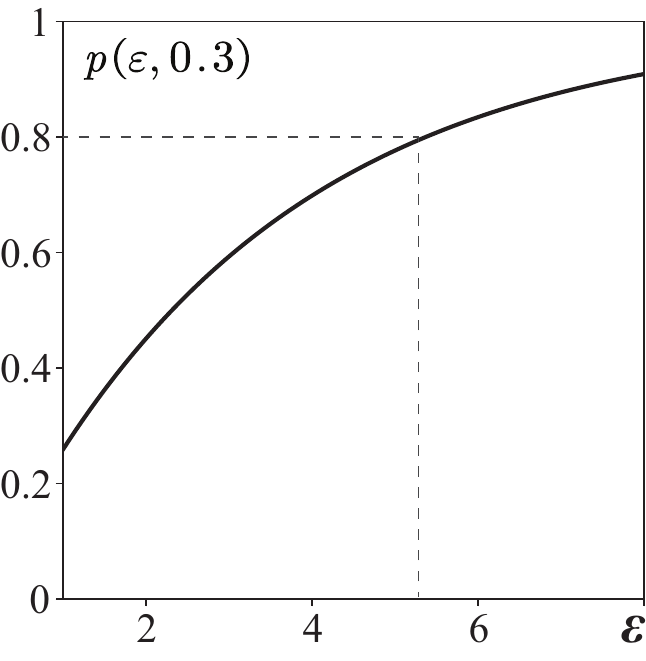}
        \caption{Classifier utility w.r.t. $\varepsilon$.}
        \label{fig:rho_epsilon_laplace}
    \end{subfigure}
    \caption{Illustration of the example and utility quantification of the classifier under the Laplace mechanism for $x=0.5$.}
\end{figure}

Consider a classifier $h:[0,1]\to \{1,2\}$ with one-dimensional input and two classes, defined as:
\begin{equation*}
    h(x) = 
    \begin{cases}
        1, & \text{if } x\in [0.2, 0.8], \\
        2, & \text{otherwise}.
    \end{cases}
\end{equation*}
Consider a specific sensitive data $x=0.5$, 
which will be perturbed using the Laplace mechanism $\mechanism_\varepsilon$ before submitting it to the classifier.
As a result, the classifier's output, $h(\mechanism_\varepsilon(0.5))$, becomes a random variable due to the randomness of $\mechanism_\varepsilon$.

To quantify the utility of $h(\mechanism_\varepsilon(0.5))$ with respect to $\varepsilon$, we analyze two key properties: 
the concentration of LDP mechanism $\mechanism_\varepsilon$ and the robustness of classifier $h$.

\textbf{Concentration analysis.}
Denote region $B_\theta(x)$ as the $\theta$-neighbor\-hood ball of $x$.
For instance, $B_{0.3}(0.5) \coloneq \{\tilde{x}: |\tilde{x}-0.5|\leq 0.3\}$.
Then the perturbed data $\mechanism_\varepsilon(0.5)$ falls within $B_{0.3}(0.5)$ with a probability $p(\varepsilon,\theta=0.3)$.
Specifically, denote $F_{\text{Lap}}(\theta)$ as the cumulative distribution function (CDF) of the Laplace distribution used in $\mechanism_\varepsilon$.
We have:
\begin{equation*}
    \begin{split}
        p(\varepsilon,\theta=0.3) &= \Pr[\mechanism_\varepsilon(x) \in B_{0.3}(x)] = F_{\text{Lap}}(0.3) - F_{\text{Lap}}(-0.3) \\
        &= 2F_{\text{Lap}}(0.3) - 1 = 1 - \exp(-0.3 \varepsilon).
    \end{split}
\end{equation*}
Therefore, we can state that: $\mechanism_\varepsilon(0.5)$ outputs a value within $B_{0.3}(0.5)$ with probability $1 - \exp(-0.3 \varepsilon)$.
For instance, if we set $\varepsilon = 2$, we have $p(2,0.3) = 0.46$.
Figure~\ref{fig:example} shows the illustration of this instance.

\textbf{Robustness analysis.}
Analyzing the utility of $h(\mechanism_\varepsilon(0.5))$ requires analyzing the robustness of $h$ under the input range $B_\theta(0.5)$.
We know that $h$ is robust under input range $B_{0.3}(0.5) = [0.2, 0.8]$ according to its definition.
Then, we can state that: $h$ preserves the correct classification under input range $B_{0.3}(0.5)$.

\textbf{Utility quantification.}
Combining the above, we can quantify the utility of $h(\mechanism_2(0.5))$:
\emph{with probability at least $p(2,0.3) = 0.46$, $h$ preserves the correct classification under input $\mechanism_2(0.5)$.}
Importantly, this utility quantification can be extended to any $\varepsilon$ using the closed-form expression of $p(\varepsilon,\theta)$.

\textbf{Application.}
For this classifier, the theoretical utility quantification $p(\varepsilon,0.3)$ can guide the choice of the privacy parameter $\varepsilon$ in the Laplace mechanism to achieve a desired utility level.
Figure~\ref{fig:rho_epsilon_laplace} shows $p(\varepsilon,0.3)$ as a function of $\varepsilon$.
For instance, to ensure the classifier preserves the correct classification with probability $\geq 0.8$, 
privacy parameter $\varepsilon$ should be set to at least $5.2$.

\subsection{Further Steps}
The above example has demonstrated an analytical approach to quantify the utility of a classifier under the Laplace mechanism.
However, this example is limited in scope as it focuses on a specific noise-adding mechanism and assumes a 
white-box classifier with a known robustness radius.
In practice, several challenges arise: 
(i) Diverse LDP mechanisms. Not all LDP mechanisms behave as noise-adding mechanisms, which means the concentration analysis may differ significantly.
(ii) Variety of classifiers. Classifiers come in various forms, and determining their robustness radius often requires different techniques.
(iii) Conservative utility quantification. The utility provided in this example is a lower bound under strong pure LDP constraints and a conservative robustness radius. 
It can be refined, especially for higher-dimensional classifiers.

The subsequent sections will address these challenges by introducing a general framework for utility quantification and refinement techniques. Specifically, we will cover the following aspects:
\begin{itemize}
    \item \emph{General LDP mechanisms.} We model an LDP mechanism as a general distribution function and derive its concentration analysis $p(\varepsilon,\theta)$. As examples, we explore commonly used continuous and discrete mechanisms.
    \item \emph{Unified method for finding robustness radius.} We treat classifiers as black-box functions and adapt a hypothesis testing framework to determine their robustness radius. This approach provides a unified method applicable to classifiers with different structures and access types (white-box or black-box).
    \item \emph{Refinement techniques.} We refine the utility quantification by incorporating the concepts of robustness hyperrectangles and PAC privacy.
\end{itemize}

\section{Formal Framework} \label{sec:framework}

This section presents the formal framework for quantifying classifier utility under input perturbations from various types of LDP mechanisms.

\subsection{Concentration Analysis of LDP} \label{subsec:boundedness_ldp}

In Definition~\ref{def:eldp} of LDP, the mechanism $\mechanism(x)$ represents a family of
probability distributions determined by the input $x$.
Thus, the perturbed output $\tilde{x} = \mechanism(x)$ is a random variable.
The concentration property characterizes how the probability mass of the perturbed output $\tilde{x}$ concentrates around the original input $x$, and can be stated as follows.

\begin{proposition}[Concentration property of LDP] \label{prop:concentration_ldp}
    Let $\mechanism(x)$ be an $\varepsilon$-LDP mechanism applied to data $x \in \mathbb{R}$,
    and let $F_\mechanism$ denote the CDF of the perturbed data $\mechanism(x)$.
    The concentration property of $\mechanism(x)$ is given by:
    \begin{equation*}
        \forall a < b: \Pr[a \leq \mathcal{M}(x) \leq b] = F_\mechanism(b) - F_\mechanism(a).
    \end{equation*}
\end{proposition}

This proposition follows directly from the definition of the CDF and provides a probabilistic characterization of the ``boundedness'' of the LDP mechanism $\mechanism(x)$.
As a special case, if $a = x - \theta$ and $b = x + \theta$, 
it becomes $\Pr[\mechanism(x) \in B_\theta(x)]$.
For practical mechanisms designed for better utility, $\mechanism(x)$ is typically concentrated around $x$.
This implies that $F_\mechanism(b) - F_\mechanism(a)$ captures most of the probability mass when $x \in [a, b]$.
We analyze the concentration properties of typical continuous and discrete LDP mechanisms in the following subsections.

\subsubsection{Continuous Mechanisms} \label{subsubsec:continuous_mechanism}
When the data domain is continuous, the LDP mechanism corresponds to a continuous distribution.
Based on whether the perturbation depends on the input $x$, continuous mechanisms are categorized into:
(i) Noise-adding mechanisms, which add noise independent of the input $x$.
(ii) Piecewise-based mechanisms, which sample outputs from input-dependent piecewise distributions.

\textbf{Noise-adding Mechanisms.}
These mechanisms achieve LDP by adding carefully designed noise to the input data.

\begin{definition}
    A noise-adding LDP mechanism $\mechanism:\mathcal{X} \to \mathsf{Range}(\mechanism)$ is defined as:
    $\mechanism(x) = x + \xi$, where $\xi$ is a random variable (noise) drawn from a symmetric distribution with mean $0$.
\end{definition}

Noise-adding mechanisms are input-independent, using the same noise distribution for all inputs.
Examples include the Laplace and Gaussian mechanisms~\cite{DBLP:journals/fttcs/DworkR14}.

\textbf{Concentration Analysis.}
For a noise-adding mechanism $\mechanism(x)$, 
the probability of outputting a perturbed value $\mechanism(x) \in B_\theta(x)$ is:
\begin{equation*}
    F_{\mechanism}(x+\theta) - F_{\mechanism}(x-\theta) = 2 \cdot F_{\mechanism}(\theta) - 1,
\end{equation*}
where $F_{\mechanism}$ is the CDF of the noise distribution.
The last equality holds because of the symmetry of the noise distribution, making the result independent of $x$.
For instance, the Laplace mechanism's concentration analysis is given by:
\begin{equation*}
    2 F_{\text{Lap}}(\theta) - 1 = 1 - \exp(-\theta \varepsilon).
\end{equation*}
In the previous section, we have seen that for $\varepsilon=2$ and $\theta=0.3$, the concentration probability is $0.46$.
This indicates that while the Laplace mechanism is defined on the whole real line, it is concentrated in a small region around $x$.

\textbf{Piecewise-based Mechanisms.}
These mechanisms achieve LDP by sampling outputs from carefully designed piecewise distributions that vary based on the input $x$.

\begin{definition}
    A piecewise-based LDP mechanism $\mechanism(x):\mathcal{X} \to \mathsf{Range}(\mechanism)$ is defined as:
    \begin{equation*}
        pdf[\mechanism(x)=\tilde{x}] =
        \begin{dcases}
            p_{\varepsilon} & \text{if} \ \tilde{x} \in [l_{x,\varepsilon}, r_{x,\varepsilon}], \\
            p_{\varepsilon} / \exp{(\varepsilon)} & \text{otherwise},
        \end{dcases}
    \end{equation*}
    where $p_{\varepsilon}$ is the sampling probability, and $[l_{x,\varepsilon}, r_{x,\varepsilon}]$ define the sampling interval based on $x$ and $\varepsilon$.
\end{definition}

Piecewise-based mechanisms are input-dependent, meaning the sampling distribution changes for different inputs $x$.
Examples include PM~\cite{DBLP:conf/icde/WangXYZHSS019} and SW~\cite{DBLP:conf/sigmod/Li0LLS20}, which use different $p_{\varepsilon}$ and intervals $[l_{x,\varepsilon}, r_{x,\varepsilon}]$.
New mechanisms can be designed by modifying these parameters.

\textbf{Concentration Analysis.}
For a piecewise-based mechanism $\mechanism(x)$, the probability of getting a perturbed value $\mechanism(x) \in B_\theta(x)$ depends on $\theta$.
The concentration analysis is given by:
\begin{equation*}
    \begin{cases}
        2\theta \cdot p_{\varepsilon} \quad \text{if } B_\theta(x) \subseteq [l_{x,\varepsilon}, r_{x,\varepsilon}], \\
        (r_{x,\varepsilon} - l_{x,\varepsilon}) p_{\varepsilon} + \left(2\theta - (r_{x,\varepsilon} - l_{x,\varepsilon})\right) \frac{p_{\varepsilon}}{\exp(\varepsilon)} \quad \text{otherwise}.
    \end{cases}
\end{equation*}
This analysis is $x$-dependent due to the terms $l_{x,\varepsilon}$ and $r_{x,\varepsilon}$.
For example, in the PM mechanism with $p_{\varepsilon} = \exp(\varepsilon/2)$, if $B_\theta(x) \subseteq [l_{x,\varepsilon}, r_{x,\varepsilon}],$ the probability is $2\theta \cdot \exp(\varepsilon/2)$. 
Concretely, for $\varepsilon=2$ and $\theta=0.3$, the concentration probability $\Pr[\mechanism(x) \in B_\theta(x)]$ equals $0.86$,
indicating that the PM mechanism is more concentrated than the Laplace mechanism under these parameters.

\subsubsection{Discrete Mechanisms}
When the data domain is discrete and finite, the LDP mechanism corresponds to a discrete distribution.
Many real-world datasets are discrete, such as image pixel values (e.g. $256$ levels) or integer-valued attributes like age.

\begin{definition}
    If $|\mathcal{X}| < \infty$, a discrete LDP mechanism $\mechanism(x):\mathcal{X} \to \mathcal{X}$ is defined as:
    \begin{equation*}
        \Pr[\mechanism(x) = \tilde{x}_i] = p_{\varepsilon}(\tilde{x}_i, x), \quad \tilde{x}_i \in \mathcal{X},
    \end{equation*}
    where $p_{\varepsilon}(\tilde{x}_i, x)$ is the probability of perturbing $x$ to $\tilde{x}_i$,
    determined by $\varepsilon$, $x$, and $\tilde{x}_i$.
\end{definition}

Examples of such mechanisms include the $k$-RR mechanism~\cite{DBLP:conf/nips/KairouzOV14}, where $p_{\varepsilon}(\tilde{x}_i, x)$ is uniform except for $\tilde{x}_i = x$, and the Exponential mechanism~\cite{DBLP:journals/fttcs/DworkR14}, where $p_{\varepsilon}(\tilde{x}_i, x)$ depends on the distance between $\tilde{x}_i$ and $x$.

\textbf{Concentration Analysis.}
For discrete mechanisms, the support is a finite set $\mathcal{X}$, and the CDF is defined as:
$F_{\mechanism}(\tilde{x}) = \sum_{\tilde{x}_i \leq \tilde{x}} p_{\varepsilon}(\tilde{x}_i, x)$.
The probability of $\mechanism(x)$ producing a value in $B_\theta(x)$ is:
\begin{equation*}
    F_{\mechanism}(x+\theta) - F_{\mechanism}(x-\theta) = \sum_{\tilde{x}_i \in [x-\theta, x+\theta]} p_{\varepsilon}(\tilde{x}_i, x).
\end{equation*}
For instance, in the $k$-RR mechanism:
\begin{equation*}
    p_{\varepsilon}(\tilde{x}_i, x) = \frac{\exp(\varepsilon)}{|\mathcal{X}| - 1 + \exp(\varepsilon)} \text{ if } \tilde{x}_i = x
    \text{ otherwise } \frac{1}{|\mathcal{X}| - 1 + \exp(\varepsilon)}.
\end{equation*}
Then, if $\mathcal{X}$ discretizes $[0,1]$ evenly, i.e. $\mathcal{X} = \{0, 1/|\mathcal{X}|, 2/|\mathcal{X}|, \ldots, 1\}$,
the concentration analysis becomes:
\begin{equation*}
    \sum_{\tilde{x}_i \in [x-\theta, x+\theta]} p_{\varepsilon}(\tilde{x}_i, x) = \frac{\exp(\varepsilon)}{|\mathcal{X}| - 1 + \exp(\varepsilon)} + \frac{2\theta |\mathcal{X}| - 1}{|\mathcal{X}| - 1 + \exp(\varepsilon)}.
\end{equation*}
When $|\mathcal{X}|=100$, $\varepsilon=2$, and $\theta=0.3$, the concentration probability of the $k$-RR mechanism is $0.63$.

\subsubsection{Comparison of Concentration Properties} \label{subsubsec:comparison_concentration}

Fixing $\varepsilon=2$ and $\theta=0.3$, the concentration probabilities of the mentioned three mechanisms are:
\begin{itemize}
    \item Laplace mechanism: $p(\varepsilon=2, \theta=0.3) = 0.46$;
    \item PM mechanism: $p(\varepsilon=2, \theta=0.3) = 0.86$;
    \item $k$-RR mechanism: $p(\varepsilon=2, \theta=0.3) = 0.63$.
\end{itemize}
These results highlight that different LDP mechanisms are differently concentrated under the same privacy parameter $\varepsilon$ and region $B_\theta(x)$. 
This variation is critical for analyzing classifier utility and selecting appropriate LDP mechanisms. 
\ For instance, if a classifier's robustness radius $\theta$ is $0.3$, and we set privacy level $\varepsilon=2$, 
the PM mechanism achieves the highest probability of $\mechanism(x) \in B_{0.3}(x)$, making it the best choice in this scenario.

\subsection{Robustness Analysis of Classifier} \label{subsec:robustness_classifier}

We analyze the robustness of classifiers under the concept of \emph{probabilistic robustness},
a probabilistic approximation of deterministic robustness.\footnote{
    If the classifier is white-box, deterministic robustness can be derived analytically or through verification methods.
    See Appendix~\ref{appendix:robustness_radius_white_box} for details.
    In the most ideal case, the deterministic robustness can be expressed as a closed-form function of arbitrary $x$.
    We omit such easier cases for generality.
}
Probabilistic robustness suffices for classifier utility analysis under LDP-perturbed inputs,
as the concentration properties of LDP mechanisms are also probabilistic.
Compared to deterministic robustness, it is computationally efficient and applicable to any classifier, including black-box classifiers,
thus suitable as a general framework.



\subsubsection{Probabilistic Robustness} \label{subsec:black_box}

Hypothesis testing provides an intuitive perspective for probabilistic robustness.
We can state two counter hypotheses: 
\begin{itemize}
    \item $H_0$: $h$ is robust under $B_\theta(x)$.
    \item $H_1$: $h$ is not robust under $B_\theta(x)$.
\end{itemize}
These hypotheses can be statistically tested by querying the classifier on random samples in $B_\theta(x)$.\footnote{
    Robustness analysis is a different scenario from \emph{using} the classifier, and does not leak users' privacy.
    Section~\ref{subsubsec:privacy_discussion} provides further discussions.
}
If robustness holds for almost all samples, $H_0$ is accepted with high confidence.
Formally, probabilistic robustness is defined as follows:

\begin{definition}[Probabilistic robustness~\cite{robust_analysis,DBLP:conf/pkdd/ZhangRF22}] \label{def:probabilistic_robustness}
    Given a classifier $h$, an input $x$, a radius $\theta$, and a tolerance $\tau >0$,
    $h$ is $\theta$-locally-robust at $x$ with probability at least $1-\omega$ if
    \begin{equation*}
        \Pr_{\tilde{x}\sim B_\theta(x)}\big[1 - \Pr[h(\tilde{x}) = h(x)] \leq \tau\big] \geq 1 - \omega,\footnote{
            The inner probability $\Pr[h(\tilde{x}) = h(x)]$ is taken over the uncertainty of robustness, refer to Appendix~\ref{appendix:definition_prob_robustness} for explanation.
        }
    \end{equation*}
    where $\tilde{x}$ is a random variable drawn from $B_\theta(x)$.
\end{definition}

This definition implies that the event $h(\tilde{x}) = h(x)$ is almost always true (within a tolerance $\tau$) with probability at least $1-\omega$.
Here, $1-\omega$ also corresponds to the confidence level (type I error) in the hypothesis testing.
For minimal $\tau$ and $\omega$, e.g. $\tau = 0.01$ and $\omega = 0.05$, $h$ is nearly $\theta$-locally robust at $x$.
Deterministic robustness is a special case with $\tau = 0$ and $\omega = 0$.

Probabilistic robustness is typically computed using the Hoeffding bound~\cite{robust_analysis,DBLP:conf/pkdd/ZhangRF22}, which provides a probabilistic guarantee for the mean of a random variable.
Let $Z$ be the indicator function of $h(\tilde{x}) = h(x)$, i.e. $Z = 1$ if robust and $Z = 0$ otherwise.
Then, the empirical mean $\hat{\mu}_Z$ (robust frequency) is calculated by sampling $n$ points, 
and the true robustness probability $\mu_Z$ can be derived using the Hoeffding bound.

\begin{theorem}[Hoeffding bound~\cite{Hoeffding_inequality}] \label{thm:hoeffding}
    For any $\tau \geq 0$ and $\omega > 0$,
    the probability that the empirical mean ($\hat{\mu}_Z$) is within $\tau$ of the true mean ($\mu_Z$) is bounded by:
    \begin{equation*}
        \Pr_{Z_i \sim P_Z}\big[|\hat{\mu}_Z - \mu_Z| \leq \tau\big] \geq 1 - \omega,
    \end{equation*}
    where $\omega = 2e^{-2n\tau^2}$ and $n$ is the number of samples.
    Equivalently, at least $n(\omega,\tau) = (\ln\frac{2}{\omega}) / (2\tau^2)$ samples are required to ensure the bound.
\end{theorem}

The number of samples $n$ balances utility guarantees and computational cost.
A smaller $\omega$ (more deterministic guarantee) requires more samples $n$.
Given $\tau$ and $\omega$, we check if the empirical mean $\hat{\mu}_Z$ on $n(\omega, \tau/2)$ samples 
satisfies $\hat{\mu}_Z \in [1-\tau/2, 1]$.
If this is true, by the Hoeffding bound, $|\mu_Z - \hat{\mu}_Z| \leq \tau/2$ holds, which implies $\mu_Z \in [1-\tau, 1]$.
Thus, Definition~\ref{def:probabilistic_robustness} holds, and we can claim that $h$ is robust under $B_\theta(x)$ with probability at least $1 - \tau$,
and our confidence level is at least $1 - \omega$.

\textbf{Independent of the classifier.}
The above method treats the classifier as a black-box function, making it applicable to any classifier $h$ 
without requiring knowledge of its internal structure or parameters.

\subsubsection{Procedure for Finding $\theta$} \label{subsubsec:procedure_for_finding_theta_star}

\begin{figure}[t]
    \centering
    \includegraphics[width=0.87\linewidth]{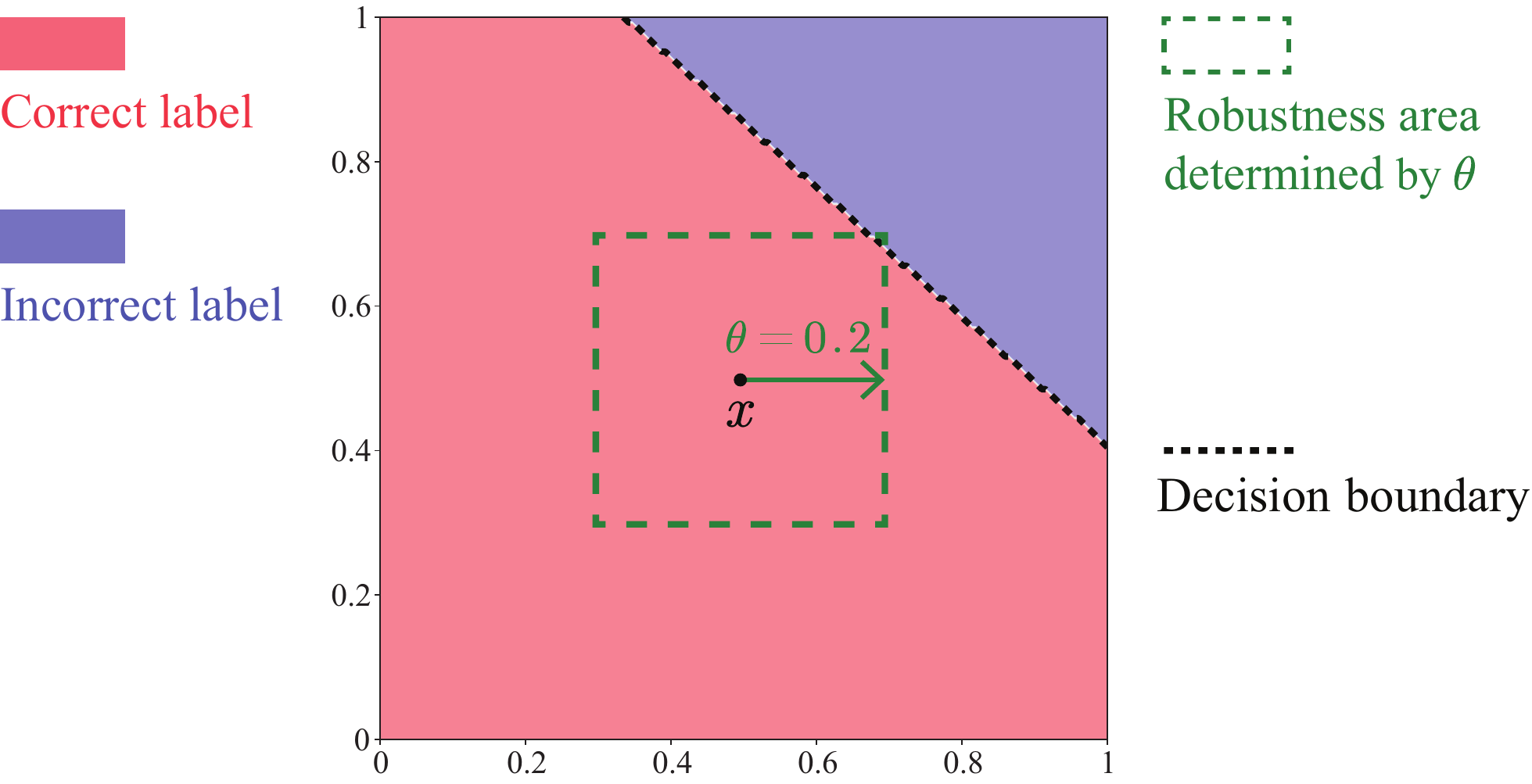}
    \caption{Robustness radius $\theta$ of a $2$D classifier at $x$ and the tested decision boundary by brute force.}
    \label{fig:decision_boundary}
\end{figure}

\begin{figure*}[htbp]
    \centering
    \begin{subfigure}[b]{0.3\textwidth}
        \centering
        \includegraphics[width=0.95\textwidth]{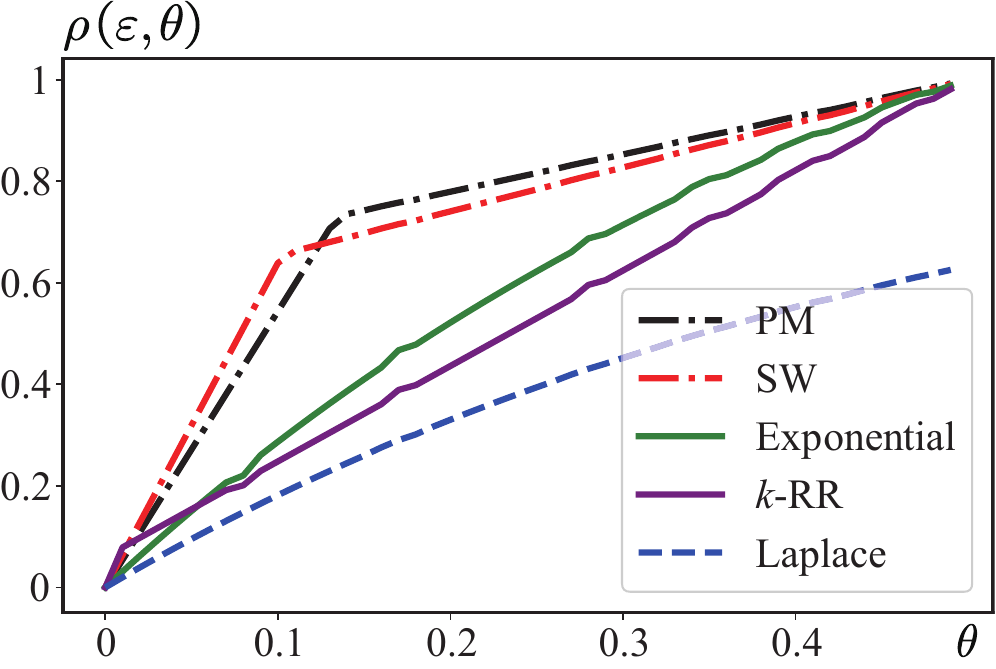}
        \caption{$\varepsilon=2$}
    \end{subfigure}
    \hfill
    \begin{subfigure}[b]{0.3\textwidth}
        \centering
        \includegraphics[width=0.95\textwidth]{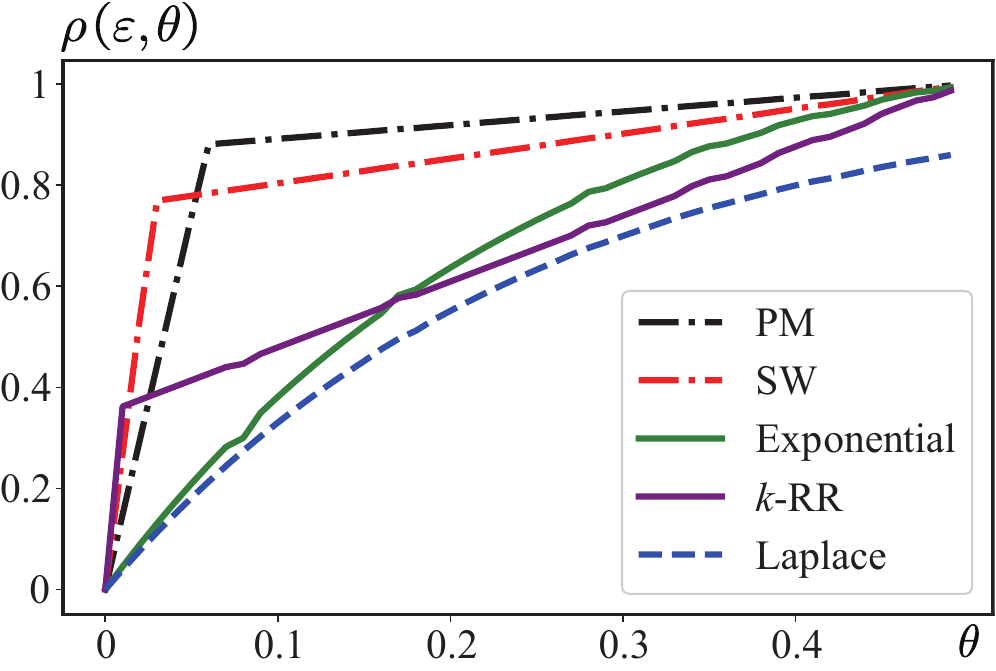}
        \caption{$\varepsilon=4$}
    \end{subfigure}
    \hfill
    \begin{subfigure}[b]{0.3\textwidth}
        \centering
        \includegraphics[width=0.95\textwidth]{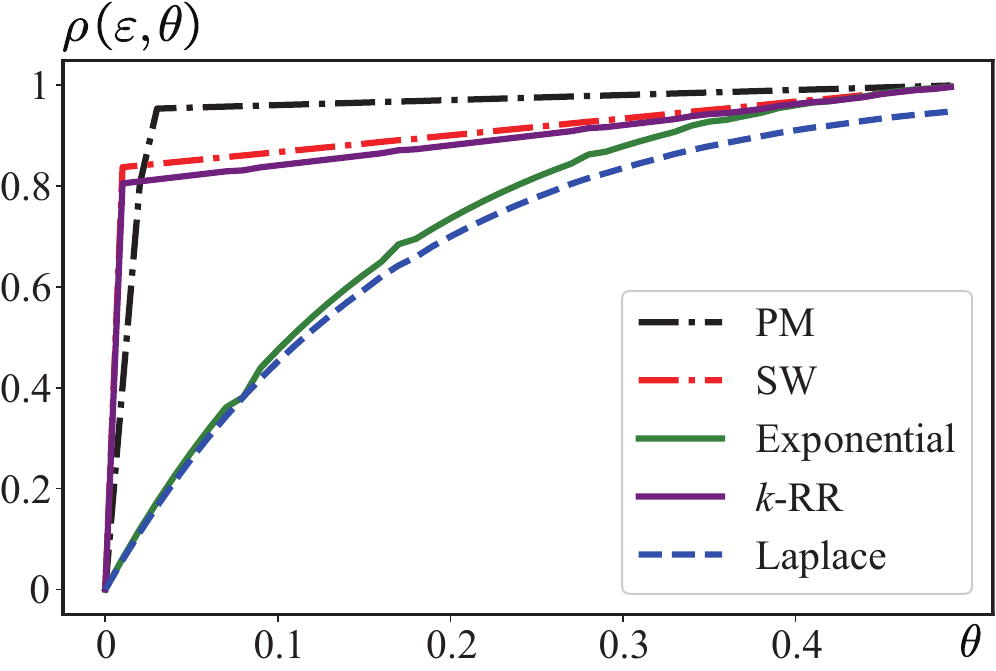}
        \caption{$\varepsilon=6$}
    \end{subfigure}
    \caption{Comparison of the robustness probability $\rho(\varepsilon, \theta)$ for different LDP mechanisms with $\varepsilon=2, 4, 6$. 
    Details of the mechanisms and discussions on the curves of $\rho(\varepsilon, \theta)$ are provided in Appendix~\ref{appendix:mechanisms}.}
    \label{fig:concentration_analysis}
\end{figure*}

After setting tolerance $\tau$ and confidence level $\omega$, we can find a probabilistic guaranteed robustness radius by the following procedure.
The key idea is to increase $\theta$ from $0$ until the misclassification rate exceeds $\tau / 2$.
Specifically,
\begin{enumerate}
    \item Compute the required number of samples $n(\omega, \tau/2)$ using the Hoeffding bound in Theorem~\ref{thm:hoeffding}.
    \item For the current $\theta$, uniformly sample $n$ points from $B_\theta(x)$, then query the classifier and calculate the misclassification rate.
    \item If the misclassification rate is below $\tau / 2$, increase $\theta$ and repeat the second step.
\end{enumerate}
As $\theta$ increases, the misclassification rate will eventually exceed $\tau / 2$. The maximum $\theta$ before this point is the robustness radius.
In practice, binary search is often used in the third step to reduce the complexity to a logarithmic scale.

\begin{example} \label{example:decision_boundary}
    Given a $2$D neural network classifier $h: [0,1]^2 \to \{1,2\}$, treated as a black-box, 
    and parameters $\tau=0.02$ and $\omega=0.05$, 
    we determine the robustness radius $\theta$ at $x = (0.5, 0.5)$ using the described procedure. 
    Specifically: (i) the required number of samples for testing is $n(\omega, \tau/2) \approx 1.8 \times 10^4$; 
    (ii) by testing $\theta \in [0,1]$, we find the robustness radius to be $\theta = 0.20$.

    Figure~\ref{fig:decision_boundary} illustrates this example.
    The true decision boundary of $h$ at $x$ is shown by the black dashed line, computed via brute force.
    The green dashed box represents the robustness area $B_{\theta}(x)$ determined by the found $\theta=0.20$.
    This result closely touches the true decision boundary.
\end{example}

In the above example, any perturbation within the robustness region $B_{\theta}(x)$ ensures that the classifier's output remains unchanged, thereby preserving utility. 
When the perturbation is performed by an LDP mechanism, the perturbed value $\tilde{x}$ will fall within $B_{\theta}(x)$ with a specific probability.
The next subsection establishes the connection between $B_{\theta}(x)$ and this probability to quantify the classifier's utility under LDP-perturbed inputs.


\subsection{Utility Quantification} \label{subsec:utility_statement}

By combining the concentration analysis of LDP mechanisms and the robustness analysis of classifiers,
we derive a utility quantification for a given classifier under inputs perturbed by
an LDP mechanism w.r.t. any privacy parameter $\varepsilon$.
This quantification framework provides a systematic evaluation of classifier performance under LDP-perturbed inputs.
In contrast, the empirical approach requires re-evaluation for different $\varepsilon$ and 
fails to establish analytical relationships between classifier utility and $\varepsilon$.

Formally, by Definition~\ref{def:probabilistic_robustness}, 
the classifier is robust under $B_{\theta}(x)$ with tolerance $\tau$ and confidence level $1-\omega$.
Therefore, under the same confidence level ($1-\omega$), a probabilistic estimate of the event $h(\tilde{x}) = h(x)$,
i.e. exact robustness without tolerance, is at least $1-\tau$.
Meanwhile, the LDP mechanism $\mechanism$ perturbs $x$ into $B_{\theta}(x)$ with probability $F_{\mechanism}(x+\theta) - F_{\mechanism}(x-\theta)$.
Thus, the utility guarantee of the classifier under the LDP mechanism is quantified by the product of these two probabilities.
\ Specifically, given a $d$-dimensional classifier $h$ and raw data $x$, denote $\mechanism^d(x)$ as the perturbed $d$-dimensional input by $d$ independent $\mechanism$.\footnote{
    For simplicity, we assume $\mechanism$ is applied to all dimensions of $h$.
    In practice, it may only be applied to a subset of sensitive features.
}
We claim:
\emph{Under confidence level at least $1 - \omega$, with probability at least $\rho(\varepsilon, \theta)$, 
the classifier $h$ preserves the correct classification result under the input perturbed by $\mechanism^d$, where}
\begin{equation*}
    \begin{split}
        \rho(\varepsilon, \theta) &= \Pr[\mechanism^d(x) \in B_\theta(x)] \cdot \Pr_{\tilde{x}\sim B_\theta(x)}[h(\tilde{x}) = h(x)]\\
        &\geq\left(F_{\mechanism}(x+\theta) - F_{\mechanism}(x-\theta)\right)^d \cdot (1-\tau),
    \end{split}
\end{equation*}
\emph{with $F_{\mechanism}$ the CDF of the LDP mechanism $\mechanism$.}

For a given mechanism $\mechanism$, the term $F_{\mechanism}(x+\theta) - F_{\mechanism}(x-\theta)$ is a closed form w.r.t. $\varepsilon$ and $\theta$,
directly computable from the mechanism's parameters.
$\omega$ and $\tau$ are predetermined based on Definition~\ref{def:probabilistic_robustness} of robustness. 
We fix $\omega = 0.05$ and $\tau = 0.01$ for all subsequent analyses.
Thus, with the known robustness radius $\theta$ of the classifier $h$ at $x$, 
$\rho(\varepsilon, \theta)$ provides an immediate utility quantification for any $\varepsilon$. 

Furthermore, if the distribution of $x$ as $P_x$ is known, we can extend the utility quantification to provide average-case (expected) and worst-case guarantees:\footnote{
    Evaluating average-case and worst-case utility requires knowledge of the input distribution $P_x$.
    In most papers, analyses of worst-case utility for their proposed mechanisms (defined on a specific domain)
    implicitly assume $P_x$ is defined on that domain with non-zero probability.
    In practice, $P_x$ can be estimated from historical data or domain knowledge.
}
\begin{equation*}
    \rho_{\mathrm{avg}}(\varepsilon) = \mathbb{E}_{x \sim P_x}[\rho(\varepsilon, \theta)], \text{and } \rho_{\mathrm{wor}}(\varepsilon) = \min_{x \sim P_x}[\rho(\varepsilon, \theta)].
\end{equation*}
Such guarantees provide input-distribution-aware utility quantification for the classifier.
For example, the average-case utility guarantee can be stated as follows:
\emph{On average, under confidence level at least $1 - \omega$,\footnote{
    For clarity, we omit the confidence level in subsequent statements,
    as it is a predefined constant with $1 - \omega \approx 1$.
}
with probability at least $\rho_{\mathrm{avg}}(\varepsilon)$, the classifier $h$ preserves the correct classification result for an input $x \sim P_x$ perturbed by $\mechanism^d$.}

\textbf{Application 1: Select the best LDP mechanism.}
A direct application of this utility quantification is comparing the 
classifier's utility under different LDP mechanisms to select the best one.
Mechanisms with higher $\rho(\varepsilon, \theta)$ yield higher utility guarantees for classifiers.
However, this varies with $\theta$ for different classifiers and $\varepsilon$ for different mechanisms.
In fact, there is not a single best mechanism for all classifiers under all $\varepsilon$.

\begin{example}
    Figure~\ref{fig:concentration_analysis} shows the utility guarantee $\rho(\varepsilon, \theta)$
    for different LDP mechanisms, with $\theta$ ranging from $0.1$ to $0.5$ and at $x = 0.5$.
    Detailed instantiations of these mechanisms are provided in Appendix~\ref{appendix:mechanisms}.
    A higher $\rho(\varepsilon, \theta)$ indicates a higher utility for classifiers with robustness radius $\theta$.
    It is evident that no single mechanism is universally optimal for all $\varepsilon$ and $\theta$.
    For instance, when $\varepsilon$ and $\theta$ are small (e.g. $\varepsilon=2$ and $\theta\leq 0.1$), the SW mechanism performs best.
    As $\varepsilon$ and $\theta$ increase, the PM mechanism becomes the superior choice.
\end{example}

\textbf{Takeaway results.}
Although no mechanism is universally optimal for all classifiers (with different robustness radius $\theta$) under all $\varepsilon$,
heuristic guidelines can assist in mechanism selection.
From the above example and Figure~\ref{fig:concentration_analysis}, we observe that if the privacy parameter $\varepsilon$ is not too small ($\geq 2$),
and the classifier's robustness radius $\theta$ is not too small ($\geq 0.1$),
then the PM mechanism is generally the best choice.
Experiments in Section~\ref{sec:case_studies} further support this conclusion for higher-dimensional classifiers.

\textbf{Application 2: Select \bm{$\varepsilon$} for a utility requirement.}
With a chosen mechanism, the utility quantification can guide the selection of the privacy parameter $\varepsilon$ to meet a desired utility guarantee.
For example, if the classifier's robustness radius $\theta$ is $0.3$, and we require the classifier to preserve the correct classification
result with probability $\rho(\varepsilon, \theta) \geq 0.8$ under the PM mechanism, we can set $\varepsilon = 1.8$
to meet this utility requirement while ensuring a strong privacy guarantee.

\textbf{Impact of dimensionality.}
For high-dimensional classifiers, the primary challenge lies in the (sensitive) input dimension $d$.
The utility guarantee $\rho(\varepsilon, \theta)$ decreases exponentially with $d$,
and capturing the robustness of high-dimensional classifiers becomes more difficult,
making the utility quantification less meaningful.

The second part of our contribution addresses this dimensionality issue.
By introducing more precise robustness analysis methods and relaxing the privacy constraint,
we refine the utility guarantee $\rho(\varepsilon, \theta)$,
making the utility quantification more practical for high-dimensional classifiers.

\section{Refinement Techniques} \label{sec:refined_bound}

This section introduces two techniques for refining the utility quantification presented in Section~\ref{subsec:utility_statement}.
The utility guarantee $\rho(\varepsilon, \theta)$ is an at-least guarantee based on the robustness radius under pure $\varepsilon$-LDP. 
To refine it, we propose the following techniques.
\begin{itemize}
    \item \emph{Robustness refinement.} We generalize the robustness radius to a robustness hyperrectangle, 
    which assigns distinct radii to different dimensions, 
    enabling a more precise computation of $\rho(\varepsilon, \theta)$.
    \item \emph{Privacy relaxation.} We relax the pure $\varepsilon$-LDP constraint to $(\varepsilon, \delta)$-PAC LDP, 
    which tolerates a small probability of failure. 
    Under this relaxation, we introduce a $\delta$-probabilistic indicator for applying the $\varepsilon$-LDP mechanism and propose an extended Gaussian mechanism tailored for PAC LDP.\footnote{
        The primary motivation for adopting PAC LDP is to accommodate the Gaussian mechanism, 
        which is widely used in privacy-preserving machine learning but does not satisfy pure LDP.
    }
\end{itemize}

\subsection{Robustness Hyperrectangle} \label{subsec:robustness_hyperrectangle}

To refine the utility quantification, we introduce the notion of a robustness hyperrectangle, which provides a more precise description of the robustness area.

The commonly used robustness radius for classifiers~\cite{DBLP:conf/nips/FawziMF16} is a conservative measure of robustness. It defines the robustness area as an $\ell_\infty$-ball,\footnote{
    Other $\ell_p$-balls, such as $\ell_1$-ball and $\ell_2$-ball, are also used, but they are also symmetric across all dimensions.
} where all dimensions share the same radius. 
However, in practice, different dimensions often have varied robustness radii, making the true robustness area larger than that defined by an $\ell_\infty$-ball.
\ In our context, accurately describing the robustness area is crucial for refining utility quantification. 
To this end, we define the robustness hyperrectangle, which accounts for the robustness of each dimension individually. The formal definition is as follows.

\begin{definition}[Robustness hyperrectangle] \label{def:robustness_hyperrectangle}
    Given a $d$-dim\-ensional classifier $h$ and a data point $x$, the robustness hyperrectangle $\theta_\diamond\coloneq [a_1,b_1] \times [a_2,b_2] \times \dots \times [a_d,b_d]$ 
    is a $d$-dimensional hyperrectangle that includes $x$ and satisfies
    \begin{equation*}
        \forall \tilde{x} \in \theta_\diamond, \; h(\tilde{x}) = h(x).
    \end{equation*}
\end{definition}

The robustness radius is a special case of the robustness hyperrectangle, where each dimension has the same size $\theta$ around $x$. More generally, we say the classifier $h$ is $\theta_\diamond$-locally robust at $x$.

\textbf{Finding robustness hyperrectangle.}
The robustness hyperrectangle $\theta_\diamond$ is not unique, unlike the robustness radius.
Different methods for finding $\theta_\diamond$ may yield different results.
Heuristically, we aim to include the robustness radius (i.e. $\theta \subseteq \theta_\diamond$) to ensure an improved utility quantification.

After determining the robustness radius $\theta$ using the procedure in Section~\ref{subsubsec:procedure_for_finding_theta_star}, we initialize the robustness area of each dimension as $[a_i, b_i]$ with $\theta$. 
We then try to expand $[a_i, b_i]$ along the $a_i$ or $b_i$ directions. This process outputs a maximal robustness hyperrectangle $\theta_\diamond$ that satisfies the definition in Definition~\ref{def:robustness_hyperrectangle}. 

\begin{figure}[t]
    \centering
    \begin{subfigure}[b]{0.43\linewidth}
        \centering
        \includegraphics[width=0.9\textwidth]{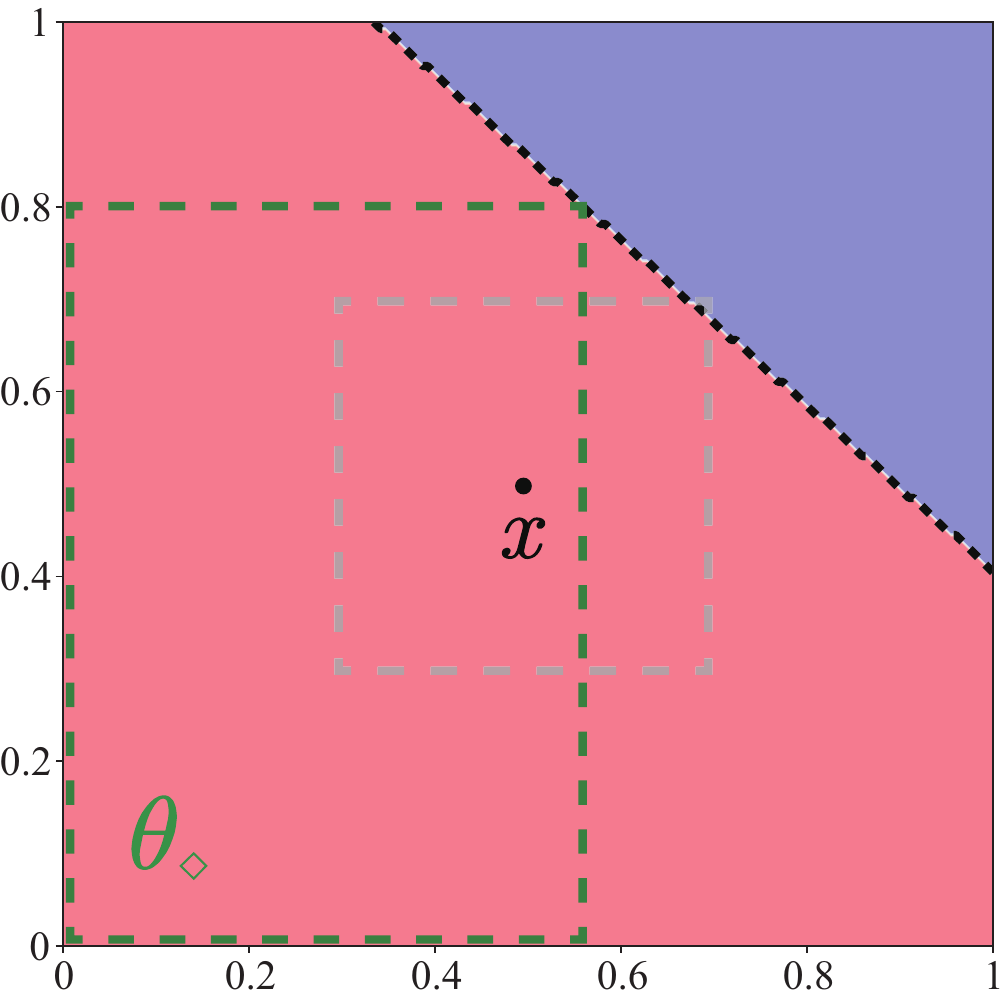}
        \caption{$\theta_\diamond = [0,0.55]\times [0,0.8]$.}
        \label{fig:robustness_hyperrectangle_1}
    \end{subfigure}
    \hfill
    \begin{subfigure}[b]{0.43\linewidth}
        \centering
        \includegraphics[width=0.9\textwidth]{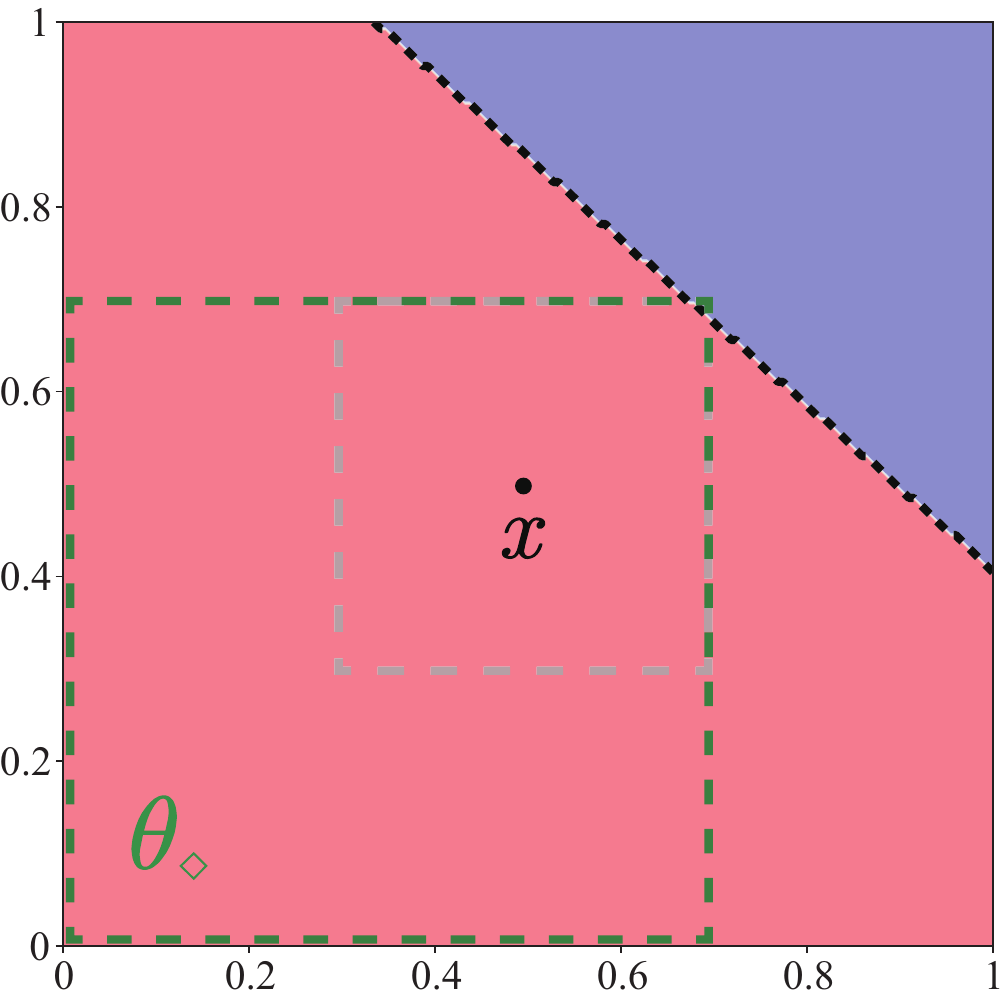}
        \caption{$\theta_\diamond = [0,0.7]\times [0,0.7]$.}
        \label{fig:robustness_hyperrectangle_2}
    \end{subfigure}
    \caption{Two robustness hyperrectangles (green dashed boxes) for the classifier in Figure~\ref{fig:decision_boundary}.
        The original robustness radius $\theta = 0.2$ corresponds to the gray dashed box.}
    \label{fig:robustness_hyperrectangle}
\end{figure}

\begin{example}
    Figure~\ref{fig:robustness_hyperrectangle} illustrates two examples of robustness hyperrectangles for the classifier in Example~\ref{example:decision_boundary}. 
    The original robustness radius is $\theta = 0.2$, corresponding to the hyperrectangle $[0.4,0.6]\times [0.4,0.6]$ (gray dashed boxes). 
    Two different hyperrectangles (green dashed boxes) are shown, both larger than the original robustness radius.
    Figure~\ref{fig:robustness_hyperrectangle_1} partially includes the robustness radius, while Figure~\ref{fig:robustness_hyperrectangle_2} fully encompasses it.
    Our proposed procedure will obtain the hyperrectangle depicted in Figure~\ref{fig:robustness_hyperrectangle_2}.
\end{example}

Each dimension of $\theta_\diamond$, represented as $[a_i,b_i]$, influences the concentration analysis of LDP mechanisms. 
Larger $[a_i,b_i]$ allows for greater perturbation, leading to a stronger utility guarantee. 
Specifically, the concentration level $\theta$ in the utility guarantee $\rho(\varepsilon, \theta)$ is replaced with the refined $\theta_\diamond$, 
resulting in a strictly improved utility guarantee $\rho(\varepsilon, \theta_\diamond)$.

\textbf{Refined utility quantification.}
Given a $d$-dimensional classifier $h$, LDP mechanism $\mechanism_\varepsilon$ for each dimension, and sensitive data $x$.
Denote the utility quantification $\rho(\varepsilon, \theta)$ as defined in Section~\ref{subsec:utility_statement}.
It can be refined to:
\emph{With probability at least $\rho(\varepsilon, \theta_\diamond)$, the classifier $h$ 
preserves the correct classification under the input perturbed by $d$ independent $\mechanism_\varepsilon$ (i.e. pure $d\varepsilon$-LDP), where}\footnote{
    We omit the terms $(1-\omega)$ and $(1-\tau)$ related to $\rho(\varepsilon, \theta_\diamond)$ for simplicity in presentation.
    They are almost $1$ in practice (e.g. $\omega = 0.05, \tau = 0.01$) and thus also negligible.
}
\begin{equation*}
    \rho(\varepsilon, \theta_\diamond) = \prod_{i=1}^d F_{\mechanism}(b_i) - F_{\mechanism}(a_i),
\end{equation*}
\emph{where $[a_i,b_i]$ represents the $i$-th dimension of the hyperrectangle $\theta_\diamond$.}

\subsection{PAC Privacy} \label{subsec:pac_privacy}

Probably approximately correct (PAC) privacy~\cite{DBLP:conf/crypto/XiaoD23} provides a formal framework for relaxing the worst-case privacy guarantee.
While pure LDP requires the privacy constraint to hold in all cases, PAC privacy allows a small probability of failure, i.e. probabilistic LDP.
This relaxation allows the use of a wider range of mechanisms, notably the Gaussian mechanism, which is widely used in privacy-preserving machine learning.

To simplify notation, we denote the privacy loss~\cite{DBLP:journals/fttcs/DworkR14} (of a mechanism $\mechanism$) between $x_1$ and $x_2$ at an observation $\tilde{x}$ as
\begin{equation*}
    \mathcal{L}_{\mechanism,x_1, x_2}(\tilde{x}) \coloneq \ln\left(\frac{\Pr[\mechanism(x_1) = \tilde{x}]}{\Pr[\mechanism(x_2) = \tilde{x}]}\right).
\end{equation*}
Pure $\varepsilon$-LDP can then be expressed as: $\mathcal{L}_{\mechanism,x_1, x_2}(\tilde{x}) \leq \varepsilon$ for all $x_1, x_2, \tilde{x}$.
In contrast, $(\varepsilon, \delta)$-PAC LDP is defined by a relaxed $\delta$-probabilistic condition.\footnote{
    Readers may wonder about the relationship between $(\varepsilon, \delta)$-PAC LDP and $(\varepsilon, \delta)$-LDP.
    In fact, the traditional $(\varepsilon, \delta)$-LDP from Dwork~\cite{DBLP:journals/fttcs/DworkR14}
    is ambiguous and has many interpretations. 
    Concentrated privacy~\cite{DBLP:journals/corr/DuchiWJ16} from Dwork, R\'enyi privacy~\cite{DBLP:conf/csfw/Mironov17}, 
    and PAC privacy~\cite{DBLP:conf/crypto/XiaoD23} can address this ambiguity.
    See Appendix~\ref{appendix:pac_ldp} for a detailed discussion.
}

\begin{definition}[PAC LDP, adopted from~\cite{DBLP:conf/crypto/XiaoD23}] \label{def:pac_ldp}
    A randomized mechanism $\mechanism: \mathcal{X} \to \mathsf{Range}(\mechanism)$ satisfies $(\varepsilon,\delta)$-PAC LDP if
    \begin{equation*}
        \forall x_1, x_2, \tilde{x}: \Pr[\mathcal{L}_{\mechanism,x_1, x_2}(\tilde{x}) \leq \varepsilon] \geq 1- \delta,
    \end{equation*}
    i.e. $\mechanism$ satisfies $\varepsilon$-LDP with probability at least $1-\delta$.
\end{definition}

Pure LDP is a special case of $(\varepsilon,\delta)$-PAC LDP with $\delta=0$.
Although PAC LDP allows a $\delta$-probabilistic failure, 
the adversary can only infer the original data with probability at most $\delta$ when $\varepsilon$-LDP fails.

\textbf{Combination of PAC LDP mechanisms.}
When applying $d$ independent PAC LDP mechanisms to $d$ dimensions, if we follow the combination theorem from Dwork~\cite{DBLP:journals/fttcs/DworkR14},
it gives $(d\varepsilon,d\delta)$-PAC LDP.
However, this result is not tight, as it allows the failure probability $d\delta$ to exceed $1$, which is impossible.
We provide a tighter result for combining $d$ independent PAC LDP mechanisms.

\begin{theorem}[Combination of $(\varepsilon,\delta)$-PAC LDP] \label{thm:pac_composition}
    Given $d$ independent mechanisms $\mechanism_1, \mechanism_2, \dots, \mechanism_d$ that satisfy $(\varepsilon,\delta)$-PAC LDP,
    their combination satisfies $(d\varepsilon, 1-(1-\delta)^d)$-PAC LDP.
\end{theorem}

\begin{proof}
    (Sketch) The result of $\delta$ comes from computing the total failure probability of $d$ independent mechanisms.
    Appendix~\ref{appendix:pac_composition} provides the details.
\end{proof} 

With the notion of PAC LDP, we present two techniques to improve the utility quantification.
The first is a privacy indicator, which extends any pure LDP mechanism to PAC LDP.
The second is an extended Gaussian mechanism, which cannot satisfy pure LDP but satisfy PAC LDP.

\subsubsection{Privacy Indicator} \label{subsubsec:privacy_indicator}

The failure probability in $(\varepsilon,\delta)$-PAC LDP allows us to design a $\delta$-probabilistic 
indicator for the application of the pure LDP mechanism.
Specifically, we define
\begin{equation*}
    \mathcal{I}_\delta(\mechanism(x)) = 
    \begin{cases}
        \mechanism(x) & \text{w.p. } 1-\delta, \\
        x & \text{w.p. } \delta,
    \end{cases}
\end{equation*}
i.e. with probability $1-\delta$, the mechanism $\mechanism$ is applied, and with probability $\delta$, the original data $x$ is returned.
The mechanism $\mathcal{I}_\delta(\mechanism(x))$ satisfies $(\varepsilon,\delta)$-PAC LDP.

\begin{theorem} \label{thm:privacy_indicator}
    Assume $\mechanism$ is a randomized mechanism that satisfies $\varepsilon$-LDP.
        Then the mechanism $\mathcal{I}_\delta(\mechanism(x))$ defined above satisfies $(\varepsilon,\delta)$-PAC LDP.
\end{theorem}

\begin{proof}
    (Sketch) The proof follows directly from the definition of $\mathcal{I}_\delta(\mechanism(x))$. Appendix~\ref{appendix:privacy_indicator} provides the details.
\end{proof}

\textbf{Privacy indicator for multiple mechanisms.}
For a $d$-dimen\-sional data $x$, instead of applying the privacy indicator to each dimension independently, 
we can treat $d$ mechanisms applied to each dimension as a single mechanism designed for $d$-dimensional data, denoted as $\mechanism^d(x)$. 
The privacy indicator $\mathcal{I}_\delta(\mechanism^d(x))$ then controls the application of all mechanisms together.
\ By substituting the pure LDP mechanism $\mechanism^d(x)$ with $\mathcal{I}_\delta(\mechanism^d(x))$, 
we can refine the utility quantification while relaxing the privacy guarantee to $(\varepsilon, \delta)$-PAC LDP.

\textbf{Refined utility quantification.} 
Given a $d$-dimensional classifier $h$, combined LDP mechanism $\mechanism^d_\varepsilon$, privacy indicator $\mathcal{I}_\delta(\cdot)$,
and raw data $x$,
the utility quantification $\rho(\varepsilon, \theta)$ can be refined to:
\emph{With probability at least $\delta + (1-\delta)\rho(\varepsilon, \theta_\diamond)$, the classifier $h$ 
preserves the correct classification under the input perturbed by $\mathcal{I}_\delta(\mechanism^d(x))$, which 
guarantees $(d\varepsilon, \delta)$-PAC LDP.}

The refined utility guarantee $\delta + (1-\delta)\rho(\varepsilon, \theta_\diamond)$ is always larger than $\rho(\varepsilon, \theta_\diamond)$, 
as it equals $\rho(\varepsilon, \theta_\diamond) + \delta(1-\rho(\varepsilon, \theta_\diamond))$, and the latter term is non-negative. 
Therefore, the new utility guarantee is always a refinement over the original one.

\subsubsection{Extended Gaussian Mechanism} \label{subsec:extended_gaussian_mechanism}

It is already known that the Gaussian mechanism~\cite{DBLP:journals/fttcs/DworkR14} cannot satisfy pure $\varepsilon$-LDP,
making the privacy indicator inapplicable.
Nonetheless, Dwork has shown that it can satisfy $(\varepsilon,\delta)$-LDP for $\varepsilon \leq 1$.
\ Under the notion of PAC LDP, we provide an extended Gaussian mechanism to work for any $\varepsilon$.\footnote{
    Analytical Gaussian mechanism~\cite{DBLP:conf/icml/BalleW18} works for any $\varepsilon$,
    but it is based on an alternative privacy definition and lacks analytical form of the noise scale $\sigma$.
    Appendix~\ref{appendix:analytic_gaussian} provides the detailed comparison.
}
The following theorem is the form for data domain $x\in [0,1]$.


\begin{theorem} [Extended Gaussian mechanism] \label{thm:gaussian_mechanism}
    The Gaussian mechanism $\mechanism_\varepsilon(x) = x + \mathcal{N}(0, \sigma^2)$ satisfies $(\varepsilon,\delta)$-PAC LDP
    for $x \in [0,1]$ if $\sigma$ is defined as
    \begin{equation*}
        \sigma \geq \frac{\sqrt{2}}{2} \left(\frac{\sqrt{\ln(2/\delta) + \varepsilon} + \sqrt{\ln(2/\delta)}}{\varepsilon}\right).
    \end{equation*}
\end{theorem}

\begin{proof}
    (Sketch) The proof generally follows the structure of Dwork et al~\cite{DBLP:journals/fttcs/DworkR14}, 
    but uses a rewritten privacy loss and a different tail bound for the Gaussian distribution. 
    Appendix~\ref{appendix:gaussian_mechanism} provides the details and compares this bound with others in the literature.
\end{proof}

The above theorem applies to a single dimension.
For a $d$-dimen\-sional data, we can apply the Gaussian mechanism to each dimension with the same $\sigma$.
Then the combination result from Theorem~\ref{thm:pac_composition} provides a PAC privacy guarantee.

\textbf{Refined utility quantification.}
Given a classifier $h$, combined Gaussian mechanism $\mechanism^d$ with $\delta$, and $d$-dimensional
sensitive data $x$, the utility quantification $\rho(\varepsilon, \theta_\diamond)$ can be refined to:
\emph{With probability at least $\rho(\varepsilon, \theta_\diamond)$, the classifier $h$
preserves the correct classification under the input perturbed by $\mechanism^d$, which guarantees $(d\varepsilon, 1-(1-\delta)^d)$-PAC LDP.}

The extended Gaussian mechanism and the privacy indicator provide two approaches to achieve $(\varepsilon,\delta)$-PAC LDP.
While both operate under the same privacy notion, the Gaussian mechanism actually provides a stronger effective privacy:
when it fails to satisfy the pure $\varepsilon$-LDP (with probability $\delta$),
the privacy loss exceeds $\varepsilon$ but remains bounded by a larger (unknown) value.
In contrast, the privacy indicator exposes the original data when it fails.

\section{Discussions} \label{sec:discussion}

This section discusses detailed privacy questions in the utility quantification framework and 
summarizes other discussions that are deferred to the appendix.

\subsubsection*{Detailed Privacy Discussions} \label{subsubsec:privacy_discussion}

In the LDP paradigm, it is assumed that the classifier is curious about users' sensitive data 
but honest in reporting results, 
i.e. it outputs $h(\mechanism(x))$ rather than fabricating results. 
This is a weak assumption, as any utility evaluation becomes meaningless if the classifier fabricates results.

Utility quantification, as a different scenario from \emph{using} the classifier, focuses on a specific input $x$ but does not involve interaction with users' data. 
It can be conducted either by the classifier designer or by the users,
and in both cases, there is no privacy leakage for the users.
(i) If conducted by the classifier designer, e.g. for improving design, there is no interaction with users,  and thus no privacy concerns arise.
(ii) If conducted by the users, e.g. to achieve a better privacy-utility trade-off for their data,
they can obtain the robustness radii without revealing their original data. 
For white-box classifiers, users know the classifier parameters and can compute the robustness radii.
For black-box classifiers, users can uniformly query the classifier over the entire input domain (i.e. prediction for nearly all input $x\in[0,1]^d$), 
without revealing their specific data.
This one-time procedure yields robustness radii for all inputs and privacy parameters.
Thus, the robustness radius at any input can be considered public knowledge, which is a reasonable assumption.

\subsubsection*{Other Discussions}
For brevity, additional discussions are deferred to Appendix~\ref{appendix:other_discussions}, 
including the following topics:
\begin{enumerate}
    \item Complexity of finding the robustness radius.
    \item Per-instance data utility vs aggregated data utility.
    \item Fixed classifier with noisy input vs retrained classifier with noisy data.
    \item Robustness of LDP mechanisms.
    \item Independent vs correlated LDP mechanisms.
\end{enumerate}
They provide additional insights into potential questions that may arise related to the utility quantification framework.

\section{Case Studies} \label{sec:case_studies}

This section presents case studies of our utility quantification framework.
We use representative LDP mechanisms and classifiers as examples, though the proposed framework is generalizable to any LDP mechanism and classifier type.

\subsection{Setup}
We use the following datasets to train three types of classifiers.

\begin{itemize}
    \item Stroke Prediction~\cite{stroke_dataset}: Predicts stroke likelihood based on features like age, hypertension, and BMI. 
    This dataset contains $6$ numerical input dimensions and a binary output. 
    The sensitive features are ``Age'' and ``BMI''.
    \item Bank Customer Attrition~\cite{bank_attrition}: Predicts whether a bank customer will leave based on features like age, salary, and tenure. 
    This dataset has $9$ numerical input dimensions and a binary output. The sensitive features are ``Age'' and ``Estimated Salary''.
    \item MNIST-7$\times$7 (variant of~\cite{DBLP:journals/spm/Deng12}): Predicts handwritten digits ($0$ to $9$) based on pixel values. 
    This dataset has $49$ input dimensions and a $10$-class output. All dimensions are treated as sensitive features.
\end{itemize}
The first two datasets involve direct sensitive user inputs, while the third is a standard image classification dataset. 
For the first two datasets, we use the \verb|sklearn| library with cross-validation to train Logistic Regression and Random Forest classifiers.
These traditional classifiers are widely used in medical and financial domains due to their interpretability and lower data requirements~\cite{Amini25012023}.
For the third dataset, which consists of high-dimensional images, we use the \verb|torch| library to train a Convolutional Neural Network (CNN) classifier.
All datasets are normalized to $[0,1]$ for training, as normalization aids classifier convergence and is a common practice.

Given a trained classifier, users apply LDP mechanisms to sensitive features (e.g. age, salary, or specific parts of an image) before sending their data to the classifier for prediction.
We consider typical LDP mechanisms in literature, which are summarized in Table~\ref{table:case_study}, along with the classifiers used.

\begin{table}[t]
    \centering
    \caption{Summary of LDP mechanisms, types of classifiers, and datasets used in the case studies.}
    \label{table:case_study}
    \resizebox{\linewidth}{!}{
    \setlength{\tabcolsep}{2pt}
    \begin{tabular}{ll}
    \toprule
    \textbf{Mechanism} & \makecell{Laplace~\cite{DBLP:journals/fttcs/DworkR14}, Gaussian (Theorem~\ref{thm:gaussian_mechanism}), \\ 
    PM~\cite{DBLP:conf/icde/WangXYZHSS019}, SW~\cite{DBLP:conf/sigmod/Li0LLS20}, Exponential~\cite{DBLP:conf/innovations/NissimST12}, $k$-RR~\cite{DBLP:conf/nips/KairouzOV14}} \\
    \midrule
    \textbf{Classifier} & \makecell{Low-dim: Logistic Regression, Random Forest; \\ High-dim: Neural Network} \\
    \midrule
    \textbf{Dataset} &  \makecell{Low-dim: Stroke Prediction~\cite{stroke_dataset}, Bank Customer \\ Attrition~\cite{bank_attrition}; High-dim: MNIST-7$\times$7 (variant of~\cite{DBLP:journals/spm/Deng12})} \\
    \bottomrule
    \end{tabular}}
    \vspace{-0.1cm}
\end{table}

We present theoretical utility quantification for the classifiers under the LDP mechanisms discussed above. 
This quantification is based on the classifier's robustness hyperrectangle $\theta_\diamond$ (at a record), 
determined using the search method described in Section~\ref{subsec:robustness_hyperrectangle}. 
We also compare the theoretical results with empirical utility observations. 
Specifically, for a given classifier and mechanism, we provide and compare:
\begin{itemize}
    \item \emph{Theoretical $\rho(\varepsilon,\theta_\diamond)$}: Derived using our utility quantification framework.
    \item \emph{Empirical $\hat{\rho}(\varepsilon)$}: Estimated by sampling $n = 2000$ instances $\{\tilde{x}\}_n$ from the LDP mechanism and testing the classifier for each $\varepsilon$.
\end{itemize}
The theoretical $\rho(\varepsilon, \theta_\diamond)$ is a closed-form expression that can be directly computed from the mechanism's parameters for any $\varepsilon$. 
It serves as a lower bound of the actual robustness probability and is generally smaller than $\hat{\rho}(\varepsilon)$. 
A smaller gap indicates more accurate theoretical utility quantification.

Without loss of generality, given the input data distribution $P_x$ determined by each dataset, we can also provide average-case and worst-case utility quantification:
$\rho_{\mathrm{avg}}(\varepsilon), \rho_{\mathrm{wor}}(\varepsilon)$ and their empirical counterparts $\hat{\rho}_{\mathrm{avg}}(\varepsilon), \hat{\rho}_{\mathrm{wor}}(\varepsilon)$.

\subsection{Utility Quantification Results}

We present the utility quantification results for the classifiers under input perturbation by
different LDP mechanisms.

\subsubsection{Stroke Prediction} \label{subsubsec:stroke_prediction}

We begin by considering pure LDP for this medical dataset. 
For a randomly selected record ``Age: 79, BMI: 24, Hypertension: 1, \dots'', 
we analyze theoretical and empirical utility for the classifiers under the Laplace, PM, and Exponential mechanisms. 
Additional records and mechanisms can be found in Appendix~\ref{appendix:stroke:other_records}.

\textbf{Logistic Regression.}
The robustness hyperrectangle at the given record is $\theta_\diamond = [0.63, 1]\times[0,1]$ for this classifier.
From this information, the utility guarantee under the PM mechanism can be directly derived as follows:

\vspace{0.5em}
\noindent\emph{For this trained Logistic Regression classifier on the Stroke Prediction dataset, at the record ``Age: 79, Hypertension: 1, \dots, BMI: 24'', 
with probability at least $\rho(\varepsilon, \theta_\diamond)$, 
the classifier preserves the correct prediction under the PM mechanism applied to ``Age'' and ``BMI'' (i.e. $2\varepsilon$-LDP), where}
\begin{equation*}
    \resizebox{\linewidth}{!}{$
    \begin{aligned}
        \rho(\varepsilon,\theta_\diamond) =& [F_{\text{PM}_{\varepsilon,\text{Age}}}(1) - F_{\text{PM}_{\varepsilon,\text{Age}}}(0.63)]\cdot[F_{\text{PM}_{\varepsilon,\text{BMI}}}(1) - F_{\text{PM}_{\varepsilon,\text{BMI}}}(0)] \\
        =& 1- F_{\text{PM}_{\varepsilon,\text{Age}}}(0.63).
    \end{aligned}$}
\end{equation*}
The last equality holds because the CDF of the PM mechanism is $0$ at $0$, and $1$ at $1$.
Term $F_{\text{PM}_{\varepsilon,\text{Age}}}(0.63)$ is the CDF of the PM mechanism (applied to the feature ``Age'') at $0.63$,
which can be computed from the PM mechanism's parameters, as detailed in Appendix~\ref{appendix:stroke:pm_cdf}.

\begin{figure}[t]
    \centering
    \begin{subfigure}[b]{0.45\linewidth}
        \centering
        \includegraphics[width=0.98\textwidth]{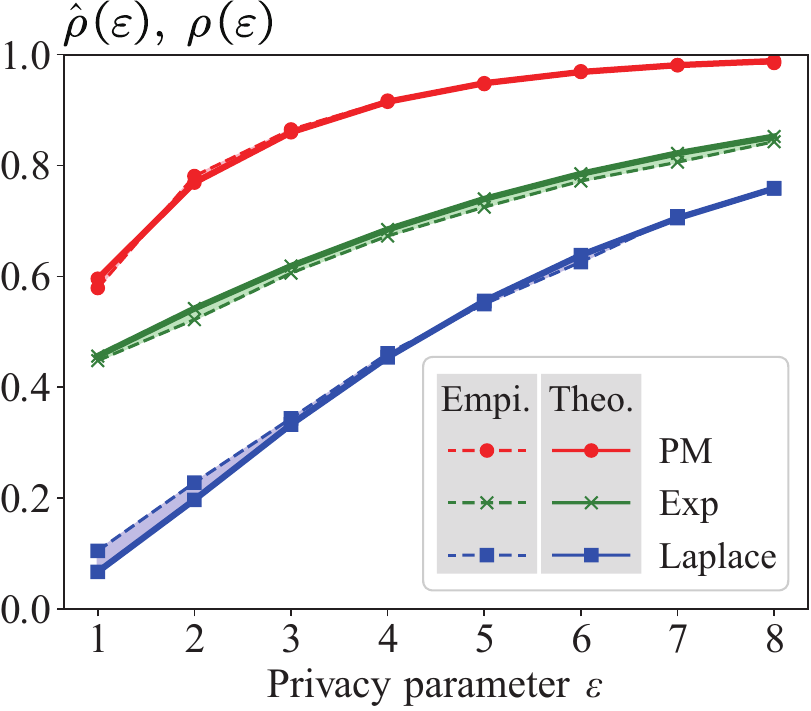}
        \caption{Logistic Regression.}
        \label{fig:exp:stroke_lr}
    \end{subfigure}
    \hfill
    \begin{subfigure}[b]{0.45\linewidth}
        \centering
        \includegraphics[width=0.98\textwidth]{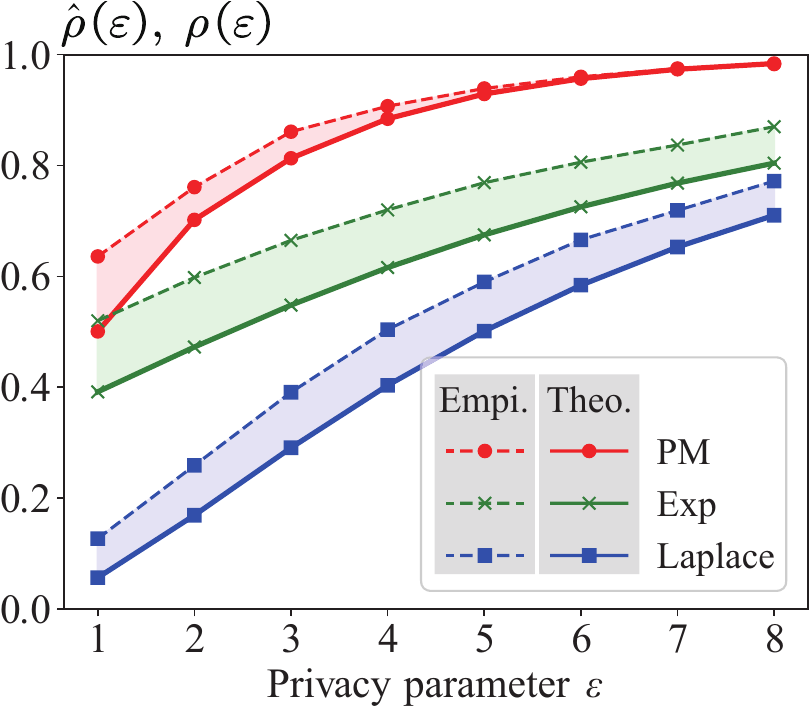}
        \caption{Random Forest.}
        \label{fig:exp:stroke_rf}
    \end{subfigure}
    \caption{Empirical and theoretical utility for two classifiers trained on the Stroke Prediction dataset.}
    \label{fig:exp:stroke}
\end{figure}

\textbf{Random Forest.}
For this classifier, the robustness hyperrectangle at the same record is $\theta_\diamond = [0.50, 1]\times[0,0.62]$.
Based on this information, the utility guarantee under the PM mechanism can be derived as follows:

\vspace{0.5em}
\noindent\emph{For this trained Random Forest classifier on the Stroke Prediction dataset, 
at the record ``Age: 79, Hypertension: 1, \dots, BMI: 24'', with probability at least $\rho(\varepsilon, \theta_\diamond)$, 
the classifier preserves the correct prediction under the PM mechanism applied to ``Age'' and ``BMI'' (i.e. $2\varepsilon$-LDP), where
\begin{equation*}
    \rho(\varepsilon, \theta_\diamond) = (1 - F_{\text{PM}_{\varepsilon,\text{Age}}}(0.5)) \cdot F_{\text{PM}_{\varepsilon,\text{BMI}}}(0.62).
\end{equation*}
}

Figure~\ref{fig:exp:stroke} compares the theoretical utility $\rho(\varepsilon, \theta_\diamond)$ with the empirical utility $\hat{\rho}(\varepsilon)$ for $\varepsilon$ values ranging from $1$ to $8$, applied to each sensitive feature. 
Both utilities increase with $\varepsilon$, with the PM mechanism consistently providing the highest utility. 
This observation aligns with our takeaway that PM is the best choice for moderate $\varepsilon$ and $\theta$ values. 
Similarly, the Exponential mechanism outperforms the Laplace mechanism. 
For both classifiers, the theoretical utility $\rho(\varepsilon, \theta_\diamond)$ closely matches the empirical utility $\hat{\rho}(\varepsilon)$, 
particularly for the Logistic Regression classifier, showing the accuracy of our theoretical utility.

The Random Forest classifier exhibits a larger gap between $\rho(\varepsilon, \theta_\diamond)$ and $\hat{\rho}(\varepsilon)$ compared to the Logistic Regression classifier. 
This discrepancy arises from the robustness hyperrectangle $\theta_\diamond$, which is a conservative estimate of the overall robustness region (difficult to determine precisely). 
For Logistic Regression, the decision boundary is a hyperplane~\cite{sahu_decision_2024}, making it more predictable and easier to approximate a hyperrectangle, 
resulting in a more accurate $\rho(\varepsilon, \theta_\diamond)$. 
Appendix~\ref{appendix:stroke:projected_decision_boundary} provides additional details on decision boundaries.

\textbf{Time cost.}
In addition to a sampling-independent theoretical guarantee, our utility quantification also has an extremely low time cost.
It does not require any sampling from the LDP mechanism, the utility statement can be directly computed for any $\varepsilon$
from the mechanism's parameters and the robustness hyperrectangle $\theta_\diamond$.
Therefore, for mechanisms having high sampling complexity, e.g. the Exponential mechanism ($\mathcal{O}(m)$ with $m$ as the domain size),
the theoretical utility quantification is much more efficient than the empirical one.

Table~\ref{tab:time_cost} shows the average time of computing one $\rho(\varepsilon, \theta_\diamond)$ and $\hat{\rho}(\varepsilon)$ for different mechanisms in this case study.
For each empirical result $\hat{\rho}(\varepsilon)$, we sample $2000$ times from the mechanism and feed the samples to the classifier for inference.
The theoretical approach is significantly faster, especially for the Exponential mechanism.
For this mechanism, computing one empirical $\hat{\rho}(\varepsilon)$ requires $859.83$ milliseconds for $2000$ samples, 
whereas computing one theoretical $\rho(\varepsilon, \theta_\diamond)$ takes only $0.94$ milliseconds, which is less than $0.1\%$.
The efficiency stems from directly substituting the privacy parameter $\varepsilon$ and the robustness hyperrectangle $\theta_\diamond$ into the closed-form expression of $\rho(\varepsilon, \theta_\diamond)$.
Moreover, although computing $\rho(\varepsilon, \theta_\diamond)$ requires the robustness hyperrectangle, $\theta_\diamond$ needs to be computed only once and can then be reused for all $\varepsilon$ values. 
In this case, computing $\theta_\diamond$ takes only $1.20$ milliseconds, which is still more efficient than estimating the empirical utility $\hat{\rho}(\varepsilon)$ separately for each $\varepsilon$.

\begin{table}[t]
    \centering
    \caption{Time cost comparison (in milliseconds).}
    \label{tab:time_cost}
    \setlength{\tabcolsep}{7pt}
    \begin{threeparttable}
    \begin{tabular}{lrrr}
        \toprule
        & PM & Exponential & Laplace \\
        \midrule
        \textbf{Empirical\tnote{a}} & $6.56 + 1.38$ & $859.83 + 1.53$ & $11.29 + 1.53$ \\
        \midrule
        \textbf{Theoretical\tnote{b}} & $0.24$ & $0.94$ & $0.30$ \\
        \bottomrule
    \end{tabular}
    \begin{tablenotes}
        \footnotesize
            \item[a] Time for generating 2000 samples + running inference.
            \item[b] Time to compute $\rho(\varepsilon, \theta_\diamond)$ only; 
            computing $\theta_\diamond$ takes $1.20$ ms but is a one-time cost amortized across all $\varepsilon$ values.
    \end{tablenotes}
    \end{threeparttable}
\end{table}

\textbf{Average-case and worst-case analysis.}
All data points in the dataset collectively define a data distribution $P_x$ over the sensitive features.
Therefore, we evaluate the utility at each $\{\text{Age, BMI}\} \sim P_x$,
then compute the mean for the average-case utility and the minimum for the worst-case utility.
This approach yields average-case and worst-case utility guarantees that account for the input data distribution.
Comprehensive results are provided in Appendix~\ref{appendix:stroke:average_worst_case}.

\subsubsection{Bank Customer Attrition} \label{subsubsec:bank_attrition}

We now evaluate PAC LDP for this financial dataset. For a randomly selected record ``Age: 22, Estimated Salary: 101,348, Credit Score: 619, \dots'', 
we analyze both theoretical and empirical utility for the two classifiers under PAC LDP mechanisms.
Specifically, we apply the privacy indicator $\mathcal{I}_\delta$ to the Laplace, PM, and Exponential mechanisms, 
while also including the Gaussian mechanism. The failure probability is fixed at $\delta = 0.1$ for all mechanisms. 
Detailed utility quantification statements for the privacy indicator and Gaussian mechanism under PAC LDP are provided in Appendix~\ref{appendix:detailed_quantification}.

Figure~\ref{fig:exp:bank_attrition} compares $\rho(\varepsilon, \theta_\diamond)$ and $\hat{\rho}(\varepsilon)$ for $\varepsilon$ values from $1$ to $8$ assigned to each sensitive feature.
The results show a similar trend to the Stroke Prediction dataset, with the Exponential mechanism performing much better in this case,
while the Gaussian mechanism lags behind the others.

Two key observations emerge in this case study:
(i) The Exponential mechanism performs comparably to the PM mechanism for Logistic Regression. 
This is because its robustness hyperrectangle $\theta_\diamond = [0, 0.72]\times [0,1]$ is large enough to cover most of the sensitive feature ranges, 
enabling excellent performance for all mechanisms defined on $[0,1]$. 
Figure~\ref{fig:concentration_analysis} illustrates this phenomenon: larger $\theta$ values lead to $\rho(\varepsilon, \theta)$ approaching $1$ for all such mechanisms. 
However, Laplace and Gaussian mechanisms are less effective here due to their probability mass extending beyond $[0,1]$.
\ (ii) The Gaussian mechanism is not as good as the other mechanisms.
The main reason is from the noise scale $\sigma$ in Theorem~\ref{thm:gaussian_mechanism}.
Although we can extend the Gaussian mechanism to $\varepsilon \geq 1$, this leads to a larger $\sigma$, i.e. less concentrated.

\begin{figure}[t]
    \centering
    \begin{subfigure}[b]{0.45\linewidth}
        \centering
        \includegraphics[width=0.98\textwidth]{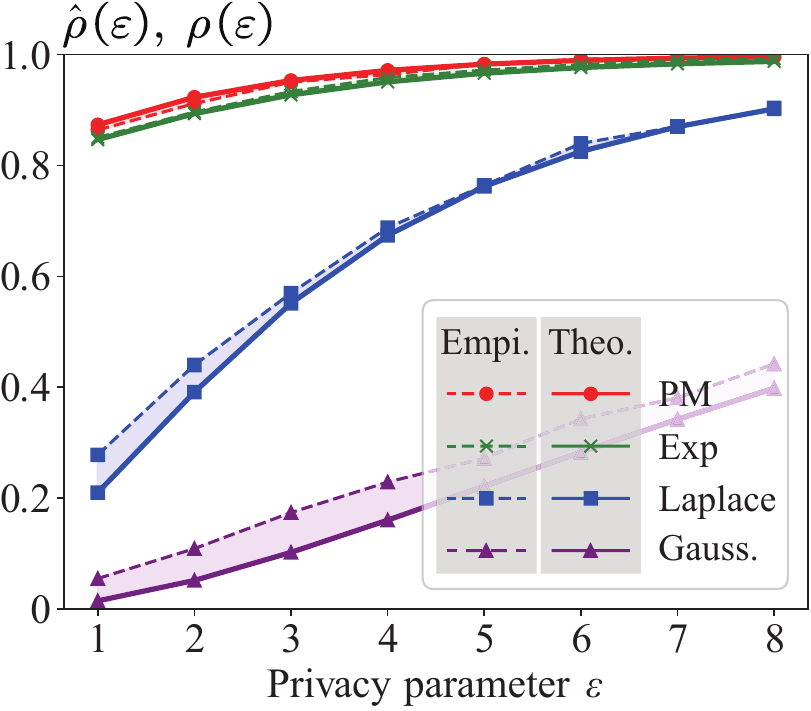}
        \caption{Logistic Regression.}
        \label{fig:exp:bank_attrition_lr}
    \end{subfigure}
    \hfill
    \begin{subfigure}[b]{0.45\linewidth}
        \centering
        \includegraphics[width=0.98\textwidth]{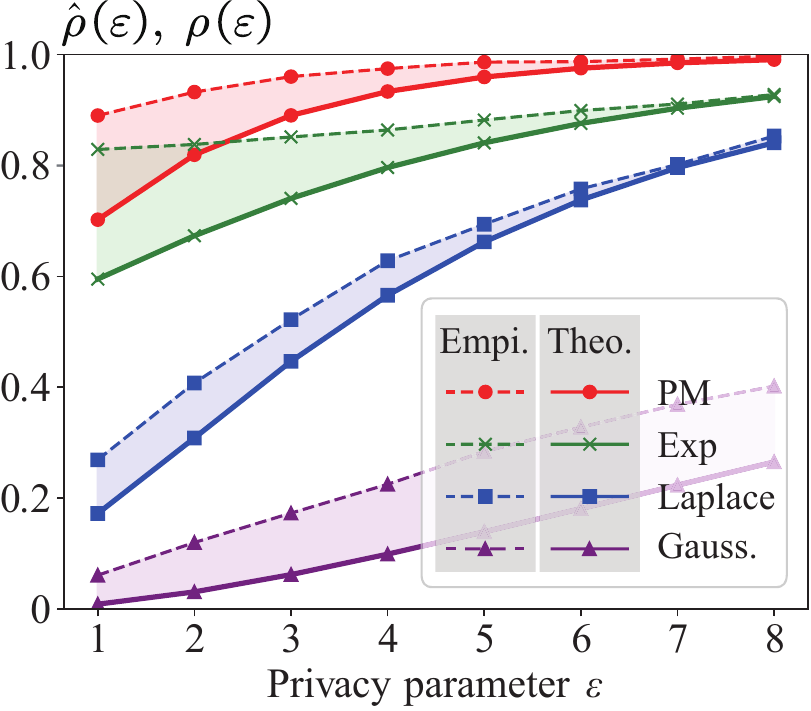}
        \caption{Random Forest.}
        \label{fig:exp:bank_attrition_rf}
    \end{subfigure}
    \caption{Empirical and theoretical utility for
    two classifiers trained on the Bank Customer Attrition dataset.}
    \label{fig:exp:bank_attrition}
\end{figure}

\textbf{Average-case and worst-case analysis.}
We can similarly provide these utility guarantees as done for the Stroke Prediction dataset.
Detailed results are presented in Appendix~\ref{appendix:bank:average_worst_case}.

\subsubsection{MNIST-7$\times$7} \label{subsubsec:mnist}

We evaluate both theoretical and empirical utility under the PM, SW, Exponential, and $k$-RR mechanisms for the first two images in this dataset, labeled as ``digit 5'' and ``digit 0.'' 
Consistent with the setup in previous case studies, we focus on PAC LDP by applying the privacy indicator $\mathcal{I}_\delta$ to these mechanisms, with the failure probability fixed at $\delta = 0.1$. 
The Laplace and Gaussian mechanisms are excluded, as they are prone to perturbing pixel values outside the valid range $[0,1]$ in this high-dimensional dataset, making fair comparisons with other mechanisms difficult.

Figure~\ref{fig:exp:mnist_cnn} compares $\rho(\varepsilon, \theta_\diamond)$ and $\hat{\rho}(\varepsilon)$ 
for $\varepsilon$ values from $1$ to $8$ assigned to each pixel.
Larger gaps between $\rho(\varepsilon, \theta_\diamond)$ and $\hat{\rho}(\varepsilon)$ are observed in this high-dimensional Neural Network classifier, especially for smaller $\varepsilon$ values.
Even for the best mechanism, PM, the theoretical utility only starts to increase significantly when $\varepsilon > 4$, with an exponential growth rate.
Similarly, the theoretical utility of the $k$-RR mechanism shows noticeable improvement when $\varepsilon > 6$.
Despite the larger gap, the trends between theoretical and empirical results remain consistent: mechanisms with higher theoretical utility also exhibit higher empirical utility.

\textbf{Reasons for the results.}
The larger gap between $\rho(\varepsilon, \theta_\diamond)$ and $\hat{\rho}(\varepsilon)$ in the Neural Network classifier compared to traditional classifiers arises from two factors:
(i) The complex decision boundary of the Neural Network, which is difficult to approximate using a single hyperrectangle.
(ii) Error accumulation in high-dimensional data, which is more pronounced than in low-dimensional data.
Specifically, the Neural Network classifier partitions the high-dimensional space into intricate regions (one for each digit class), 
making it challenging to approximate with a single hyperrectangle around a digit.
Additionally, theoretical utility is derived from the product of utility analyses for each pixel, making it more sensitive to error accumulation in high-dimensional data.
In contrast, empirical utility is unaffected by the robustness hyperrectangle $\theta_\diamond$, avoiding such error accumulation.

\textbf{Improvements for high-dimensional classifiers.}
(i) In practice, not all dimensions are sensitive and require privacy protection.
While we treat all dimensions as sensitive in the MNIST-7$\times$7 dataset, 
users can opt to apply LDP mechanisms to only a subset of dimensions, i.e. specific parts of the image.
In such cases, the theoretical utility quantification can be more accurate.
\ (ii) At present, we identify a connected (local) robustness rectangle that contains $x$; 
however, the classifier may admit additional disconnected and non-rectangular robust regions. 
One can use Monte Carlo sampling to discover such regions and then combine this with Monte Carlo integration to obtain a 
more accurate theoretical utility quantification for high-dimensional classifiers.
Appendix~\ref{appendix:mnist:improvement} provides details.

\begin{figure}[t]
    \centering
    \begin{subfigure}[b]{0.45\linewidth}
        \centering
        \includegraphics[width=0.98\textwidth]{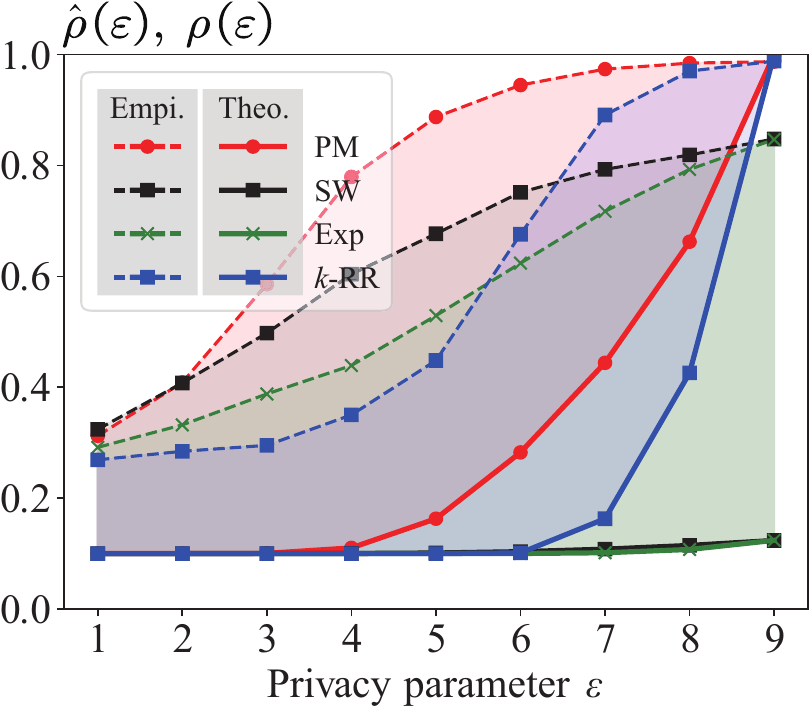}
        \caption{Neural Network on the first image in the dataset.}
        \label{fig:exp:mnist_cnn_0}
    \end{subfigure}
    \hfill
    \begin{subfigure}[b]{0.45\linewidth}
        \centering
        \includegraphics[width=0.98\textwidth]{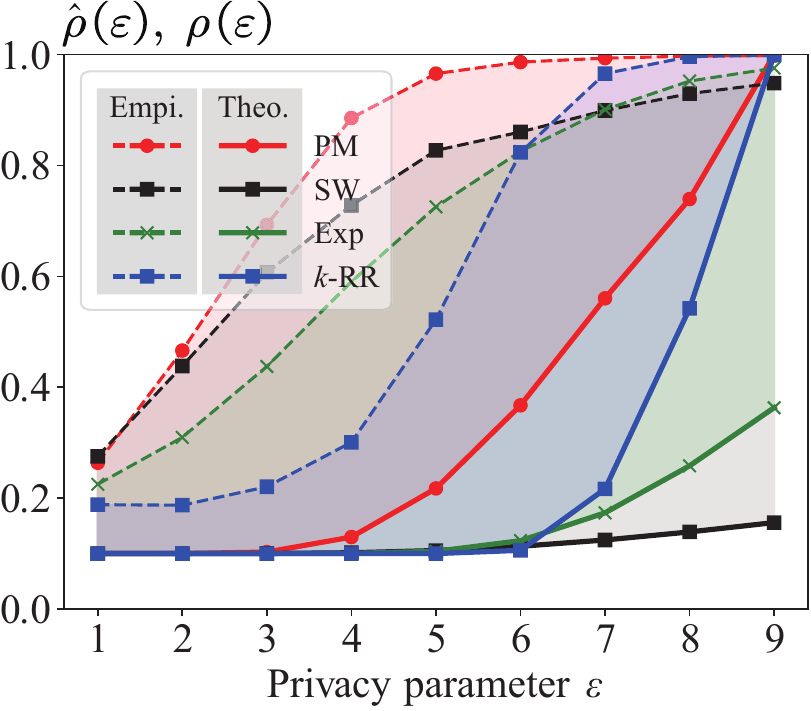}
        \caption{Neural Network on the second image in the dataset.}
        \label{fig:exp:mnist_cnn_1}
    \end{subfigure}
    \caption{Empirical and theoretical utility for a Neural Network classifier trained on the MNIST-7$\times$7 dataset.}
    \label{fig:exp:mnist_cnn}
\end{figure}

\textbf{Average-case and worst-case analysis.}
We use the first 50 images from the MNIST-7$\times$7 dataset as representative samples for analyzing the average-case and worst-case utility under LDP-perturbed inputs.
Appendix~\ref{appendix:mnist:average_worst_case} presents the results of this analysis.

\subsection{Summary of Case Studies}
We presented case studies of theoretical utility statements for different classifiers under inputs perturbed by typical LDP mechanisms and compared them with empirical utility.
The key findings from our case studies are as follows:
\begin{itemize}
    \item Theoretical utility quantification is computationally efficient, 
    requiring negligible time due to its closed-form expression for any $\varepsilon$.
    \item For low-dimensional classifiers, theoretical utility often closely matches empirical utility.
    For high-dimensional classifiers, theoretical utility has a larger gap with empirical utility,
    due to complex decision boundaries and error accumulation. 
    However, mechanisms with higher theoretical utility also exhibit higher empirical utility.
    \item The PM mechanism generally outperforms other mechanisms, providing the best utility across our case studies.
\end{itemize}

\section{Conclusions}
This paper introduces a framework for quantifying classifier utility under LDP-perturbed inputs. 
The framework connects the concentration analysis of LDP mechanisms with the robustness analysis of classifiers. 
It provides guidelines for selecting the best mechanism and privacy parameter for a given classifier. 
Beyond the core framework, we introduce two novel refinement techniques that further improve utility quantification.
Case studies on typical classifiers under LDP-perturbed inputs demonstrate the framework's effectiveness and applicability.

\textbf{Future work:} (i) The robustness is not limited to classifiers; it can be extended to robustness of programs and control systems.
Connecting differential privacy with them is a promising direction.
(ii) This paper analyzes data dimensions independently. 
Extending the utility quantification to multidimensional LDP mechanisms presents another direction for future work.

\section*{Acknowledgments}
We thank the anonymous reviewers and the revision editor for their valuable feedback and guidance, which significantly improved this paper. 
We also acknowledge the use of GPT-5 for language refinement in this paper. 
\vspace*{0.8cm}

\bibliographystyle{ACM-Reference-Format}
\bibliography{reference}
\appendix
\setcounter{footnote}{0}

\section{Related Work}
This paper studies the data utility of classifiers under perturbed inputs.
We first review closely related concepts---semantic perturbations, concentration properties of random variables, and classifier robustness---and clarify how they differ from our setting.
We then discuss broader related work.

\subsection{Closely Related Concepts}


\paragraph{Semantic perturbations.}
A number of the existing literature focuses on empirical evaluations of data utility under semantic perturbations~\cite{DBLP:conf/cvpr/MohapatraWC0D20},
i.e. perturbations that involve semantic meanings.
For image classifiers, such perturbations include brightness adjustments~\cite{bendjillali:hal-03564174}, blur~\cite{DBLP:conf/qomex/DodgeK16}, 
and data augmentation such as rotation and scaling~\cite{DBLP:conf/cvpr/KanbakMF18,DBLP:journals/corr/abs-1712-04621}.
For text classification, examples include synonym replacement and random swap~\cite{DBLP:conf/emnlp/WeiZ19,DBLP:conf/naacl/Kobayashi18}, etc.
For speech recognition, there are acoustic distortions~\cite{DBLP:journals/taslp/LiDGH14}.
These perturbations do not provide formal privacy guarantees, and some of them even preserve the semantic meanings of the original data.

This paper focuses on a specific category of mathematically tractable perturbations, LDP mechanisms.
These perturbations, such as the Laplace or Gaussian noise,
have weaker semantic meanings but provide provable privacy guarantees against adversaries.
Additionally, there are variants of LDP mechanisms tailored for semantic perturbations, such as sentence embedding perturbation~\cite{DBLP:conf/www/DuYCS23}.
The LDP mechanisms examined in this paper are typical building blocks for numerous LDP protocols in practical applications.

\paragraph{Concentration property of random variables.}
The term concentration property mostly relates to the \emph{concentration inequalities} in probability theory, 
e.g. Chernoff bound and Hoeffding inequality (see book~\cite{dubhashi2009concentration}).
These inequalities provide upper bounds on the probability that the sum of independent random variables deviates from its expected value.
From these inequalities, the notion $(\alpha,\beta)$-accuracy~\cite{dubhashi2009concentration,DBLP:conf/eurocrypt/CheuSUZZ19} is defined as:
the error is bounded by $\alpha$ with probability at least $\beta$.
In DP, this notion has been used to analyze the error of the summation query in the shuffle model~\cite{DBLP:conf/eurocrypt/CheuSUZZ19}.
It also appears in accuracy analysis of data exploring languages under DP guarantee, such as Haskell and SQL~\cite{DBLP:conf/sp/VesgaRG20,DBLP:conf/sigmod/0002HIM19}.

In comparison to $(\alpha,\beta)$-accuracy, the theoretical utility quantification $\rho(\varepsilon,\theta)$ includes 
the concentration property of LDP mechanisms and the robustness property of classifiers. 
Among these two properties, the concentration property shares similarities with $(\alpha,\beta)$-accuracy by providing a probabilistic error bound for perturbed data, 
but it specifically focuses on the distribution of an LDP mechanism. 
Meanwhile, this property is used to connect with the robustness of classifiers in this paper.

\paragraph{Robustness of classifiers.}
Robustness refers to a classifier's ability to maintain performance (utility) under input perturbations~\cite{DBLP:conf/nips/FawziMF16}.
This property is crucial as classifiers often encounter data that deviates from their training distribution.
Robustness also serves as a metric for evaluating a classifier's generalization ability~\cite{DBLP:journals/pami/ZhouLQXL23}.
Beyond semantic perturbations, robustness has been extensively studied under 
$\ell_p$-norm perturbations~\cite{DBLP:conf/nips/FawziMF16,DBLP:conf/iclr/MadryMSTV18,DBLP:journals/csur/MachadoSG23}, 
where input data is perturbed within an $\ell_p$-norm ball.
Within this perturbation model, various verification techniques have been developed to analyze classifier robustness.
In addition to approximating the robustness of a classifier as described in Section~\ref{subsec:robustness_classifier},
we can also approximate the classifier using a ``smoothed'' classifier via randomized smoothing~\cite{pmlr-v54-mcmahan17a}
and then analyze the robustness radius of the smoothed classifier.

Unlike these studies, this paper focuses on classifier robustness under LDP perturbations.
These perturbations can be viewed as occurring within \emph{probabilistic} balls defined by the concentration property, 
controlled by the privacy parameter $\varepsilon$.

\subsection{Broader Related Work}

\paragraph{Impact of (L)DP mechanisms on fairness of classifiers.}
Fairness dictates that two individuals who are similar w.r.t. a particular task should be classified similarly.
Under some metrics of similarity, it relates to the concept of DP~\cite{DBLP:conf/innovations/DworkHPRZ12}.
DP mechanisms can impact the fairness of classifiers by changing biases in the input data distribution.
For instance, if certain groups are over- or under-represented in the data provided to a classifier, 
applying DP mechanisms can shift the input distribution and thereby alter the fairness (or bias) of the resulting predictions~\cite{DBLP:conf/csfw/MakhloufSAP24,DBLP:conf/dbsec/ArcoleziMP23}.

\paragraph{DP mechanisms for privacy-preserving classifier training.}
This paper addresses data privacy when (L)DP mechanisms are applied by users before sharing data with classifiers, 
i.e. ensuring privacy during classifier usage. 
An orthogonal line of research examines privacy risks posed by classifiers themselves, 
focusing on data leakage during training. 
Such works aim to counter membership inference attacks~\cite{DBLP:conf/uss/FredriksonLJLPR14,DBLP:conf/sp/ShokriSSS17}, 
which infer training data from model parameters.
DP-SGD~\cite{DBLP:conf/ccs/AbadiCGMMT016} is a well-known defense, applying Gaussian noise to gradient updates during training.

Federated learning~\cite{pmlr-v54-mcmahan17a} is a stricter privacy setting, 
where sensitive training data is distributed across multiple users and cannot be shared with a central server.
In this context, DP-Fed-SGD~\cite{brendan2018learning} perturbs local gradients before aggregation to the central server.
For a comprehensive overview, see a recent survey~\cite{DBLP:journals/corr/abs-2405-08299}.

While these approaches focus on protecting training data privacy, 
this paper examines the utility of trained classifiers under input perturbations caused by LDP mechanisms.
In the former case, the classifier is trusted and used for inference. 
In the latter, the classifier is untrusted, and users perturb their input data before sharing it.

\paragraph{Utility of LDP mechanisms.}

Data utility is the most crucial metric for LDP mechanisms. 
For simple queries such as mean, summation, and histogram~\cite{DBLP:conf/uss/WangBLJ17,DBLP:conf/icde/WangXYZHSS019,DBLP:conf/sp/WangLJ18}, 
data utility is essentially determined by the theoretical mean squared error (MSE) of the mechanism.
While empirical evaluations can be conducted by sampling from the mechanism, 
the empirical utility will converge to the theoretical MSE when the sample size is large enough.
Meanwhile, MSE is also the most common metric in mechanism design and analysis.

However, a mechanism with lower MSE does not always guarantee better utility for all applications, particularly classifiers.
This paper has shown that classifier utility depends on both the concentration property of the mechanism and the classifier's robustness.
We provide a framework to quantify classifier utility under LDP-perturbed data, 
offering an alternative to empirical evaluations.

\section{Proofs} 

\subsection{Proof of Theorem~\ref{thm:privacy_indicator}} \label{appendix:privacy_indicator}

\begin{proof}
    From the definition of $\mathcal{I}_\delta(\mechanism(x))$, we have
    \begin{equation*}
        \Pr[\mathcal{{I}}_\delta(\mechanism(x)) = \mechanism(x)] = 1 - \delta.
    \end{equation*}
    Since $\mechanism(x)$ satisfies $\varepsilon$-LDP, it means the privacy loss
    $\mathcal{L}_{\mechanism,x_1, x_2}(\tilde{x}) \leq \varepsilon$ always holds.
    Thus, with probability at least $1 - \delta$, 
    the $\loss_{\mechanism, x_1, x_2}(\tilde{x})$ is bounded by $\varepsilon$, i.e.
    \begin{equation*}
        \Pr[\mathcal{L}_{\mechanism,x_1, x_2}(\tilde{x}) \leq \varepsilon] \geq 1 - \delta,
    \end{equation*}
    which proves that $\mathcal{I}_\delta(\mechanism(x))$ satisfies $(\varepsilon, \delta)$-PAC LDP.
\end{proof}

\subsection{Proof of Theorem~\ref{thm:gaussian_mechanism} and Discussion} \label{appendix:gaussian_mechanism}

\subsubsection{Our Gaussian Mechanism}

The proof generally follows the structure of Dwork's proof of the Gaussian mechanism on page 261 in~\cite{DBLP:journals/fttcs/DworkR14}. 
Their proof assumes $\varepsilon \leq 1$ for the soundness of the $\ln(\varepsilon, \Delta f)$ function. 
We eliminate this limitation by refining the proof technique and providing a new noise bound.

\begin{proof}
    For the Gaussian distribution and an original data $x_1 \in [0,1]$, the probability of seeing $x_1 + \tilde{x}$ with $\tilde{x} \sim \mathcal{N}(0, \sigma^2)$ 
    is given by the probability density function at $\tilde{x}$, i.e. $pdf[\tilde{x}]$.
    The probabilities of observing \emph{the same perturbed value} when the original input is $x_2 \in [0,1]$
    form the set $\{\,pdf[\tilde{x}+\Delta] : \Delta \in [-1,1]\,\}$.
    Thus, the privacy loss is
    \begin{equation*}
        \begin{split}
            \mathcal{L}_{\mechanism,x_1, x_2}(\tilde{x}) &= \ln\left(\frac{pdf[\tilde{x}]}{pdf[\tilde{x}+\Delta]}\right) = \ln\frac{\exp(-\frac{1}{2\sigma^2}\cdot \tilde{x}^2)}{\exp(-\frac{1}{2\sigma^2}\cdot (\tilde{x}+\Delta)^2)} \\
            &= \frac{1}{2\sigma^2}\cdot \left(2\tilde{x}\Delta + \Delta^2\right) = \frac{\Delta}{\sigma^2}\cdot \tilde{x} + \frac{\Delta^2}{2\sigma^2}.
        \end{split}
    \end{equation*}
    Since $\tilde{x} \sim \mathcal{N}(0, \sigma^2)$, the privacy loss is distributed as a Gaussian random variable
    with mean $\Delta^2 / (2\sigma^2)$ and variance $(\Delta/\sigma)^2 \cdot \sigma^2 = \Delta^2 / \sigma^2$.
    We can see that the Gaussian mechanism can not satisfy pure $\varepsilon$-LDP for any $\varepsilon$ because $\tilde{x} \in (-\infty, \infty)$,
    meaning the privacy loss is unbounded.

    Nonetheless, the privacy loss random variable can be bounded by $\varepsilon$ with a certain probability $1 - \delta$.
    Specifically, letting $Z \sim \mathcal{N}(0,1)$, the privacy loss can be rewritten as a random variable:
    \begin{equation*}
        \loss_{\mechanism,x_1, x_2}(\tilde{x}) \sim \frac{\Delta}{\sigma} \cdot Z + \frac{\Delta^2}{2\sigma^2}.
    \end{equation*}
    In the remainder of the proof, we will find a value of $\sigma$ that ensures $\loss_{\mechanism,x_1, x_2}(\tilde{x}) \leq \varepsilon$ with probability at least $1 - \delta$.

    \textbf{Step 1: Gaussian tail bound.}
    When $\loss_{\mechanism,x_1, x_2}(\tilde{x}) \geq \varepsilon$, we have
    \begin{equation*}
        \frac{\Delta}{\sigma} \cdot Z + \frac{\Delta^2}{2\sigma^2} \geq \varepsilon, \text{which means } |Z| \geq \frac{\varepsilon \sigma - \frac{\Delta^2}{2\sigma}}{\Delta}.
    \end{equation*}
    Since $\Delta \in [-1,1]$, we can set $\Delta = 1$ to minimize the right-hand side of the above inequality,
    yielding the worst-case $|Z|$, which gives $|Z| \geq \varepsilon \sigma - 1 / (2\sigma)$.
    We need to bound the probability of this event by $\delta$:
    \begin{equation*}
        \Pr\left[|Z| \geq \varepsilon \sigma - \frac{1}{2\sigma}\right] \leq \delta, \text{ implying } \Pr\left[Z \geq \varepsilon \sigma - \frac{1}{2\sigma}\right] \leq \frac{\delta}{2}.
    \end{equation*}
    By the Gaussian tail bound (Chernoff bound):
    \begin{equation*}
        \Pr[Z \geq t] \leq \exp\left(\frac{-t^2}{2}\right),
    \end{equation*}
    we have
    \begin{equation*}
        \exp\left(\frac{-\left(\varepsilon \sigma - \frac{1}{2\sigma}\right)^2}{2}\right) \leq \frac{\delta}{2}.
    \end{equation*}

    \textbf{Step 2: Solving \bm{$\sigma$}.}
    Define $C = \sqrt{-2\ln(\delta / 2)}$. We can rewrite the above inequality as
    \begin{equation*}
        \left(\varepsilon \sigma - \frac{1}{2\sigma}\right)^2 \geq C^2.
    \end{equation*}
    Taking the square root of both sides and rearranging gives us two cases:
    \begin{equation*}
        \varepsilon \sigma - \frac{1}{2\sigma} \geq C \quad \text{or} \quad \varepsilon \sigma - \frac{1}{2\sigma} \leq -C.
    \end{equation*}
    Therefore, the best choice of $\sigma$ is the one that satisfies:
    \begin{equation*}
        \varepsilon \sigma - \frac{1}{2\sigma} = \pm C.
    \end{equation*}
    Solving this quadratic equation gives us two solutions for $\sigma$:
    \begin{equation*}
        \sigma = \frac{\pm C \pm \sqrt{C^2 + 2\varepsilon}}{2\varepsilon}.
    \end{equation*}
    Since $\sigma$ must be positive, we take the positive root, which gives us the final solution:
    \begin{equation*}
        \sigma = \frac{\sqrt{-2\ln(\delta / 2)} +  \sqrt{-2\ln(\delta / 2) + 2\varepsilon}}{2\varepsilon},
    \end{equation*}
    which is equivalent to the result in Theorem~\ref{thm:gaussian_mechanism}.
\end{proof}

\subsubsection{Comparison with Dwork's Proof.}

There are two main concerns with Dwork's proof of the Gaussian mechanism:
\begin{itemize}
    \item The proof is hard to interpret within the $(\varepsilon, \delta)$-DP notion;
    \item The proof is based on a problematic Gaussian tail bound.
\end{itemize}
Specifically, 
(i) the proof is based on the assumption that the privacy loss can be probabilistically bounded by $1 - \delta$,
which is consistent with the PAC privacy notion but hard to interpret in the original $(\varepsilon, \delta)$ privacy notion.
The $(\varepsilon, \delta)$ privacy notion intuitively
says that the distance between two distributions is \emph{always} bounded by the given $(\varepsilon, \delta)$ pair, 
i.e. strictly cannot be unbounded.
From this notion, it is hard to interpret the probabilistic bound of $1 - \delta$ in Dwork's proof.
(ii) The proof is based on a tail bound (p262 in~\cite{DBLP:journals/fttcs/DworkR14}):
\begin{equation*}
    \Pr[Z \geq t] \leq \frac{\sigma}{\sqrt{2\pi}} \cdot \exp\left(\frac{-t^2}{2\sigma^2}\right),
\end{equation*}
where $Z \sim \mathcal{N}(0,\sigma^2)$. This bound is problematic.
As a counterexample, if $t = 0$ and $\sigma = 1$, the bound gives $\Pr[Z \geq 0] \leq 1/\sqrt{2\pi} \approx 0.399$,
which is incorrect since $\Pr[Z \geq 0] = 0.5$ for arbitrary Gaussian distributions with mean $0$.

\subsubsection{Comparison with the Analytic Gaussian.} \label{appendix:analytic_gaussian}
Another Gaussian mechanism is the analytic Gaussian mechanism~\cite{DBLP:conf/icml/BalleW18}.
The name originates from the analytic form of the Gaussian distribution in its result.
We don't adopt this mechanism for two reasons:
(i) It lacks analytical form of noise scale $\sigma$ due to the complexity of the given formula to satisfy their privacy definition.
i.e. Theorem 8 in~\cite{DBLP:conf/icml/BalleW18}.
As stated in their paper: ``we propose to find $\sigma$ using a numerical algorithm... (p4, line 7)'', 
which complicates analytical utility quantification.
(ii) It is based on an alternative privacy definition: 
\begin{equation*}
    \Pr\left[\loss_{\mechanism,x_1, x_2} \geq \varepsilon\right] - \exp(\varepsilon) \cdot \Pr\left[\loss_{\mechanism,x_1, x_2} \leq -\varepsilon\right] \leq \delta.
\end{equation*}
While they proved this definition satisfies $(\epsilon,\delta)$-DP, it is hard to interpret,
because it bounds the \emph{difference} between the probability of the privacy loss being too large ($\loss\geq\epsilon$) and too small ($\loss\leq\epsilon$), 
different as Dwork's bounding $\Pr[\loss\geq\epsilon]\leq\delta$, which is more direct and interpretable.

\subsection{Proof of Theorem~\ref{thm:pac_composition}} \label{appendix:pac_composition}

\begin{proof}
    We calculate the total failure probability of $d$ mechanisms to show that it is $1-(1-\delta)^d$.
    
    The combination of $d$ independent $(\varepsilon, \delta)$-PAC LDP mechanisms satisfies pure $\varepsilon$-LDP 
    only if all $d$ mechanisms satisfy pure $\varepsilon$-LDP. 
    Therefore, the probability of no failure is at least $(1-\delta)^d$ by independence.
    Thus, the total failure probability of $d$ independent $(\varepsilon, \delta)$-PAC LDP mechanisms is
    at most $1 - (1-\delta)^d$.

    The combination result for $\varepsilon$ is straightforward, and it is the same as the combination of pure $\varepsilon$-LDP mechanisms.
\end{proof}

\emph{Remark.} Compare with the original combination theorem for $d$ mechanisms, i.e. $(d\varepsilon, d\delta)$-LDP~\cite{DBLP:journals/fttcs/DworkR14},
our result is guaranteed to be more precise,
as $1-(1-\delta)^d \leq d\delta $ for any $\delta \in (0,1)$ and $d \geq 1$.

\section{Complementary Materials} 

Notations used in this paper are summarized in Table~\ref{tab:qcu_notations}.

\begin{table}[t]
    \begin{center}
        \caption{Notations used in this paper.}\label{tab:qcu_notations}
            \begin{tabularx}{\linewidth}{Y p{0.63\linewidth}}
                \toprule
                \textbf{Notation} & \textbf{Meaning} \\
                \midrule
                $\mechanism(x)$ & LDP mechanism applied to input $x$ \\
                $\mathcal{X}$ & Input domain of an LDP mechanism \\
                $\tilde{x} = \mechanism(x)$ & Perturbed version of $x$ (sent to $h$) \\
                $h$ & Classifier function \\
                $B_\theta(x)$ & $\theta$-ball centered at $x$ \\
                $\rho(\varepsilon, \theta)$ & Utility quantification of a classifier \\
                $F_\mechanism(\cdot)$ & CDF of LDP mechanism $\mechanism$ \\
                $\omega$ & $1 - \omega$ indicates the confidence level \\
                $\tau$ & Robustness tolerance \\
                \bottomrule
            \end{tabularx}
    \end{center}
\end{table}

\subsection{Deterministic Robustness of White-box Classifiers (Section~\ref{subsec:robustness_classifier})} \label{appendix:robustness_radius_white_box}

When the classifier is a (public) white-box model, we can directly analyze its exact robustness radius from the model parameters.
According to their representativeness, classifiers can be categorized into two types: closed-form classifiers and non-closed-form classifiers.

\subsubsection{Closed-form Classifiers} \label{appendix:robustness_closed_form}
A closed-form classifier has a mathematical expression that can be directly analyzed.
Examples include the naive Bayes classifier~\cite{naive_bayes}, QDA~\cite{DBLP:journals/corr/abs-1906-02590,lda_qda}, 
certain variants of SVM~\cite{DBLP:books/lib/HastieTF09}, and clustering models like $k$-means~\cite{wiki_k-means}
and Gaussian mixture models~\cite{GMM}.
Closed-form expression means excellent interpretability, including robustness analysis.
For these classifiers, decision boundaries can be analytically derived from the model parameters,
allowing for analytical expression of robustness radius $\theta$ for any input variable $x_0$.

Formally, if the decision boundary of $h$ is a closed-form curve (or surface) $C(x) = 0$,
then $\theta$ at an input variable $x_0$ is the distance from $x_0$ to $C(x)$.
Its analytical expression w.r.t variable $x_0$ can be calculated by minimizing the distance between $x_0$ and $C(x)$.

\begin{example}
    Consider a 2D QDA classifier $h_{1,2}(x):(x_u, x_v) \to \{1,2\}$,
    where the discriminant function for the first class is $h_1(x) = -0.5(x_u^2 + x_v^2) + \ln(0.5)$,
    and for the second class is $h_2(x) = -0.5(0.5(x_u -1)^2 + 0.5(x_v-1)^2) - 0.5\ln(4) + \ln(0.5)$.
    The decision boundary $h_1(x)-h_2(x)=0$ is a circle with radius $2\sqrt{\ln2+1}$ centered at $(-1,-1)$.
    Thus, the $\ell_2$ robustness radius at $x_0 = (x_u, x_v)$ is $\theta = |2\sqrt{\ln2 + 1} - \sqrt{(x_u+1)^2 + (x_v+1)^2}|$.\footnote{
        A trivial lower bound of the $\ell_{\infty}$ robustness radius can be obtained by using 
        $\theta / \sqrt{2}$, based on the norm inequality.
    }
\end{example}

\subsubsection{Non-closed-form Classifiers} \label{appendix:robustness_non_closed_form}

Non-closed-form classifiers often lack explicit mathematical expressions for their decision boundaries.
Examples include neural network~\cite{DBLP:conf/iclr/WengZCYSGHD18} and decision tree~\cite{DBLP:conf/nips/ChenZS0BH19}.
Their decision boundaries are analytically intractable, making it challenging to derive the robustness radius $\theta$ directly.
Finding $\theta$ for these classifiers is known as the robustness verification problem~\cite{DBLP:conf/cav/KatzBDJK17,DBLP:journals/corr/abs-2206-12227},
i.e. verify whether $h(B_\theta(x)) = h(x)$ for a given $\theta$ at $x$.
The main insight is to encode classifiers as a set of constraints,
transforming the problem of verifying $h(B_\theta(x)) = h(x)$ into a satisfiability problem
that can be solved by constraints solvers.

Formally, given a concrete sample $x_0$ and $\theta$, we have $B_\theta(x_0) = \{x: ||x-x_0||_\infty \leq \theta\}$
as a linear constraint on $x$.
Denote $h_{i}(x)$ as the score function of the $i$-th class of $h(x)$, and let $c$ be the correct class of $x_0$.
By encoding $B_\theta(x_0)$ and the classifier $h(x)$ into a satisfiability problem
\begin{equation*}
    x \in B_\theta(x_0) \ \land \ h_c(x) < h_i(x) \text{ for } i\neq c,
\end{equation*}
we can verify the problem using a constraint solver.
If it is unsatisfiable, then $h(x)$ outputs the correct class within $B_\theta(x_0)$.
The maximum $\theta$ that makes the problem unsatisfiable is the robustness radius $\theta$.

\begin{example}
    (To replicate the QDA classifier from the previous example, which is also the same classifier depicted in Figure~\ref{fig:decision_boundary}.)   
    Consider a $2$D neural network classifier $h_{1,2}(x)$ $:[0,1]^2 \to \{1,2\}$ with structure $a_2(\mathrm{ReLU}(a_1(x)))$,
    where $\mathrm{ReLU} = \max(0, x)$ and $a_1, a_2$ are affine functions of the hidden layer and output layer that
    \begin{equation*}
        \begin{split}
            a_1(x) &= \begin{bmatrix}
                \begin{array}{rr}
                    -1.86 & -2.09 \\ 0.12 & -0.46
                \end{array}
            \end{bmatrix}x + \begin{bmatrix}
                \begin{array}{r}
                    3.71 \\ -0.08
                \end{array}
            \end{bmatrix}, \\
            a_2(x) &= \begin{bmatrix}
                \begin{array}{rr}
                    -3.05 & 0.40 \\ 4.02 & -0.22
                \end{array}
            \end{bmatrix}x + \begin{bmatrix}
                \begin{array}{r}
                    0.94 \\ -0.58
                \end{array}
            \end{bmatrix}.
        \end{split}
    \end{equation*}
    Denote $h_1(x)$ and $h_2(x)$ as the 1st and 2nd dimension of $h_{1,2}(x)$.
    Given $x_0 = (0.5, 0.5)$ and $\theta = 0.1$, we have $B_{0.1}(x_0) = ||x-x_0||_\infty \leq 0.1$ as a linear constraint on $x$.    
    If the correct class is $h_{1,2}(x_0) = 1$, the robustness verification problem is
    \begin{equation*}
        x \in B_{0.1}(x_0) \ \land \ h_1(x) < h_2(x).
    \end{equation*}
    If this problem is verified as unsatisfiable, it means $h_{1,2}(x)$ always outputs class $1$ within $B_{0.1}(x_0)$. 
    Finally, finding $\argmax_\theta h(B_\theta(x)) = h(x)$ is a sequence of decision problems on $\theta$.
\end{example}

\subsection{Explanation of Definition~\ref{def:probabilistic_robustness}} \label{appendix:definition_prob_robustness}

There are two randomness sources in Definition~\ref{def:probabilistic_robustness}:
(i) Denote the deterministic robustness $h(\tilde{x}) = h(x)$ an indicator $\phi(\tilde{x})\in\{0,1\}$. 
We can say a system is robust under random $\tilde{x}$ with $\Pr[\phi(\tilde{x})=1]>1-\omega$. 
In this setting, the robustness boundary is assumed to be known, and the randomness comes from $\tilde{x}$. 
(ii) When the robustness boundary is unknown, $\phi(\tilde{x})$ must be estimated and thus takes values in $[0,1]$ instead of $\{0,1\}$.
In this case, the condition $\phi(\tilde{x})>1-\tau$ becomes a random event, where the randomness comes from the estimation of $\phi$. 
This leads to two-level probabilistic guarantee $\Pr[1 - \Pr[\phi(\tilde{x}) \le \tau]] \ge 1-\omega$. 
(Such interpretation also is lacked in existing definitions~\cite{robust_analysis,DBLP:conf/pkdd/ZhangRF22}.)

\subsection{LDP Mechanisms in Figure~\ref{fig:concentration_analysis}} \label{appendix:mechanisms}

We provide concrete instantiations of the LDP mechanisms in Figure~\ref{fig:concentration_analysis}
and discuss their $\rho(\varepsilon, \theta)$ curves.

\subsubsection{Piecewise-based Mechanism}

The original PM mechanism~\cite{DBLP:conf/icde/WangXYZHSS019} is defined on $[-1,1] \to [-C_\varepsilon, C_\varepsilon]$,
where $C_\varepsilon$ is a variable dependent on $\varepsilon$.
Such design ensures unbiasedness, as it was originally designed for mean estimation.
The original SW mechanism~\cite{DBLP:conf/sigmod/Li0LLS20} is defined on $[0,1] \to [-b_\varepsilon, 1+b_\varepsilon]$.
Unlike PM, SW is biased and designed for distribution estimation.

Although defined on different domains, both PM and SW mechanisms can be ``normalized'' to the same domain $[0,1] \to [0,1]$
with the same privacy parameter $\varepsilon$.
We give their formulation as the following two definitions.

\begin{definition}[The PM mechanism] \label{def:pm_mechanism}
    The PM mechanism takes an input $x \in [0,1]$ and outputs a random variable $\tilde{x} \in [0,1]$ as follows:
    \begin{equation*}
        pdf[\mathcal{M}(x) = \tilde{x}] = 
        \begin{dcases}
            p_{\varepsilon} & \text{if} \ \tilde{x} \in [l_{x,\varepsilon}, r_{x,\varepsilon}], \cr
            p_{\varepsilon} / \exp{(\varepsilon)} & \tilde{x} \in [0,1] \, \backslash \, [l_{x,\varepsilon}, r_{x,\varepsilon}],
        \end{dcases} 
    \end{equation*}
    where $p_{\varepsilon} = \exp(\varepsilon / 2)$,
    \begin{equation*}
        [l_{x,\varepsilon},r_{x,\varepsilon}] = 
        \begin{dcases}
            [0, 2C) &\text{if} \ x \in [0, C), \cr
            x + [-C, C] &\text{if } x \in [C, 1-C], \cr
            (1-2C, 1] &\text{otherwise}, \cr  
        \end{dcases}
    \end{equation*}
    with $C=(\exp(\varepsilon / 2) - 1) / (2\exp(\varepsilon) - 2)$.
\end{definition}

The SW mechanism is similar but uses different instantiations for $p_{\varepsilon}$ and $[l_{x,\varepsilon}, r_{x,\varepsilon})$.

\begin{definition}[The SW mechanism]
    The SW mechanism takes an input $x \in [0,1]$ and outputs a random variable $\tilde{x} \in [0,1]$ as the same formulation as the PM mechanism,
    but with
    \begin{equation*}
        p_{\varepsilon} = \frac{\exp(\varepsilon) - 1}{\varepsilon}, \quad C = \frac{\exp(\varepsilon)(\varepsilon - 1) + 1}{2(\exp(\varepsilon) - 1)^2}.
    \end{equation*}
\end{definition}

Both mechanisms guarantee $\varepsilon$-LDP through their piecewise design,
as the probability of seeing the same value $\tilde{x}$ when the original data is $x_1$ and $x_2$ is bounded by $\exp(\varepsilon)$.

\textbf{\bm{$\rho(\varepsilon,\theta)$} curves.}
The concentration analysis $F(x + \theta) - F(x - \theta)$ of the PM and SW mechanisms results in
a linear curve with two segments.
When $\theta$ is small, the concentration is determined by the $[l_{x,\varepsilon}, r_{x,\varepsilon})$ segment,
which gives a slope proportional to $p_{\varepsilon}$.
When $\theta$ is large, the concentration is also determined by the $[0,1) \, \backslash \, [l_{x,\varepsilon}, r_{x,\varepsilon})$ segment,
which gives a slope proportional to $p_{\varepsilon} / \exp{(\varepsilon)}$.
Thus, the $\rho(\varepsilon, \theta)$ curves for the PM and SW mechanisms exhibit a linear pattern with two distinct slopes.

\subsubsection{Discrete Mechanism}
By discretizing the $[0,1]$ domain into bins, LDP mechanisms defined on discrete domains can also be applied.
In Figure~\ref{fig:concentration_analysis}, we fixed the number of bins to $100$, i.e. the domain is $\mathcal{X}\coloneq \{0, 0.01, 0.02, \ldots, 1\}$.
The according Exponential mechanism and $k$-RR mechanism is defined as follows.

\begin{definition}[The Exponential mechanism]
    Given score function $d(x,\tilde{x})$,
    the Exponential mechanism defined on $\mathcal{X}$ takes an input $x\in \mathcal{X}$ and outputs a random variable $\tilde{x}\in \mathcal{X}$ as follows:
    \begin{equation*}
        \Pr[\mechanism_{\exp}(x)=\tilde{x}] = \frac{\exp \left (\frac{\varepsilon d(x,\tilde{x})}{2\Delta d}\right )}{\sum_{y' \in \mathcal{Y}} \exp \left (\frac{\varepsilon d(x,\tilde{x}')}{2\Delta d}\right )},
    \end{equation*}
     where $\Delta d = \max_{x,\tilde{x},\tilde{x}' \in \mathcal{X}} |d(x,\tilde{x}) - d(x,\tilde{x}')|$ is the sensitivity of the score function $d$.
\end{definition}

We choose the score function as the negative absolute distance, i.e. $d(x,\tilde{x}) \coloneq -|x-\tilde{x}|$, so that outputs closer to $x$ are assigned higher probabilities.
Accordingly, the sensitivity is $\Delta d = 1$.

\begin{definition}[The $k$-RR mechanism]
    $k$-RR is a sampling mechanism $\mechanism:\mathcal{X}\to \mathcal{X}$ that, 
    given input $x\in\mathcal{X}$, outputs $\tilde{x}\in\mathcal{X}$ according to
    \begin{equation*}
        \Pr[\mechanism(x)=\tilde{x}] = 
        \begin{dcases}
            \frac{\exp(\varepsilon)}{|\mathcal{X}|-1+\exp(\varepsilon)} & \text{if} \ \tilde{x} = x, \\ 
            \frac{1}{|\mathcal{X}|-1+\exp(\varepsilon)} & \text{otherwise}.
        \end{dcases}
    \end{equation*}
\end{definition}

\textbf{\bm{$\rho(\varepsilon,\theta)$} curves.}
The concentration analysis $F(x + \theta) - F(x - \theta)$ 
of the Exponential mechanism results in a smooth curve,
as the probability of seeing value $\tilde{x}$ further away from $x$ decreases smoothly w.r.t $\theta$.
In contrast, the concentration analysis of the $k$-RR mechanism exhibits a pattern similar to the PM and SW mechanisms.
When $\theta$ is small, the probability is determined by the bin containing $x$, 
results in a slope $\propto \exp(\varepsilon) / (|\mathcal{X} - 1 + \exp(\varepsilon)|)$.
When $\theta$ becomes larger, the probability is also determined by the other bins, which gives a slope $\propto 1 / (|\mathcal{X} - 1 + \exp(\varepsilon)|)$.

\subsection{PAC LDP vs $(\varepsilon,\delta)$-LDP (Section~\ref{subsec:pac_privacy})} \label{appendix:pac_ldp}

This section discusses the difference between PAC LDP, $(\varepsilon,\delta)$-LDP,
and other related privacy notions.

\begin{definition}[$(\varepsilon,\delta)$-LDP~\cite{DBLP:journals/fttcs/DworkR14}]
    A randomized mechanism $\mechanism$ satisfies $(\varepsilon,\delta)$-LDP if for all $x_1, x_2 \in \mathcal{X}$ and all $S \subseteq \mathcal{Y}$,
    \begin{equation*}
        \Pr[\mechanism(x_1) \in S] \leq \exp(\varepsilon) \cdot \Pr[\mechanism(x_2) \in S] + \delta.
    \end{equation*}    
\end{definition}

This definition says that the distance between the distributions of $\mechanism(x_1)$ and $\mechanism(x_2)$ is bounded by a $(\varepsilon,\delta)$ pair.
More specifically, it can be rewritten as:
\begin{equation*}
    \frac{\Pr[\mechanism(x_1) \in S]}{\Pr[\mechanism(x_2) \in S] + \delta / \exp(\varepsilon)} \leq \exp(\varepsilon).
\end{equation*}
Given concrete values of pair $(\varepsilon, \delta)$,\footnote{
    In fact, $\delta$ can be defined in a more intricate manner, such as incorporating the noise scale $\sigma$ as in~\cite{DBLP:conf/nips/Canonne0S20}. 
    However, this approach introduces a recursive dependency thus uncontrollable in practical applications, as $\sigma$ itself should be defined in terms of $\delta$.
} the term $\delta / \exp(\varepsilon)$ is a constant.
Thus, it means that the distance between the distributions of $\mechanism(x_1)$ and $\mechanism(x_2)$ must be strictly bounded,
i.e. similar to pure $\varepsilon$-LDP.
\ However, this definition is often interpreted as: it is $\varepsilon$-LDP with a failure probability of $\delta$.
Meanwhile, the Gaussian mechanism relies on this interpretation.

Other relaxed privacy notions, including the Concentrated privacy from Dwork~\cite{DBLP:journals/corr/DuchiWJ16}, R\'enyi privacy~\cite{DBLP:conf/csfw/Mironov17}, 
and PAC privacy~\cite{DBLP:conf/crypto/XiaoD23}, provide patches to $(\varepsilon,\delta)$-LDP.
The Gaussian mechanism has natural results under the Concentrated privacy and R\'enyi privacy notions,
while we provide a PAC privacy result in Appendix~\ref{appendix:gaussian_mechanism}.
\ Among these privacy notions, $(\varepsilon,\delta)$-PAC LDP provides the same meaning as the common interpretation of $(\varepsilon,\delta)$-LDP:
it is $\varepsilon$-LDP with a failure probability of $\delta$.

\subsection{Other Discussions (Section~\ref{sec:discussion})} \label{appendix:other_discussions}

\subsubsection{Complexity of Finding the Robustness Radius}

The procedure for finding a probabilistic robustness radius is independent of the classifier's structure and parameters. 
Its complexity arises from the iterations required to determine the robustness radius and the testing of $n$ perturbed data points in each iteration.

Using binary search to find the radius results in a complexity of $\Theta(\log(1/\kappa))$, 
where $\kappa$ is the precision of the radius. 
For instance, if $\kappa = 0.01$ (two decimal places), the complexity is $\Theta(\log 100)\approx \Theta(7)$.
\ Meanwhile, classifiers are often implemented with parallel inference on GPUs, 
which allows for $\Theta(1) \times T$ complexity for $n$ data points, 
where $T$ is the inference time for a single data point, provided $n$ is not excessively large. 
Consequently, the total complexity for finding a two-decimal precision robustness radius is $\Theta(7T)$. 
In practice, $T$ is typically in the order of milliseconds, making the complexity of finding the robustness radius sufficiently low.
Moreover, the robustness radius is a one-time computation for a given classifier,
which can be used for all future utility quantification.

\subsubsection{Per-instance Data Utility vs Aggregated Data Utility}
Our utility quantification framework focuses on classifiers and per-instance data utility, 
providing utility guarantees at the level of individual data points.   
This stands in contrast to aggregated data utility analyses such as mean estimation, 
which evaluate the accuracy of aggregated statistics computed from the perturbed data of many users.
Per-instance data utility is crucial for several reasons:
(i) Classifiers and other personalized services like recommendation systems
operate on individual data points.
Their performance depends directly on per-instance utility and has no direct connection to aggregated utility.
(ii) Per-instance data utility focuses on mechanism-level utility analysis,
characterizing the privacy-utility tradeoff curves without relying on sample size or aggregation effects.
These two reasons relate to our focus on per-instance data utility in this paper.

\subsubsection{Fixed Classifier with Noisy Input vs Retrained Classifier with Noisy Data}
These are two different scenarios. 
(i) In the former scenario, users locally apply an LDP mechanism to their data for privacy protection, 
agnostic to the service provider (i.e. the classifier), and are affected by the resulting utility of the existing classifier. 
(ii) In the latter scenario, the service provider applies LDP to users' data and retrains a new classifier tailored to the perturbed data, 
which relates to ``robust training''.

We focus on the former user-privacy-centered scenario; in the latter scenario, users expose their raw data to the service provider. 
Moreover, from users' perspective, the classifier is fixed at the time of use (i.e. when they query the service). 
If a retrained classifier is deployed, then the classifier utility is for the retrained classifier (from users' perspective). 

\subsubsection{Robustness of LDP Mechanisms}
Intuitively, LDP mechanisms with higher concentration around raw data are more prone to reconstruction (with sufficient reports) 
and poisoning attacks, as malicious data are less averaged. 
In our concentration analysis, PM, SW, and $k$-RR (especially with small $k$) are more concentrated, 
thus expected to be more vulnerable. This aligns with prior findings: 
PM and $k$-RR show more vulnerability to poisoning attacks~\cite{DBLP:conf/uss/LiLSG023,DBLP:conf/ndss/LiLLS25}. 

\subsubsection{Independent LDP Mechanisms vs Correlated LDP Mechanisms}
We adopt independent perturbations for cleanness of presentation. 
In practice, DP libraries (e.g. Google DP~\cite{googledifferential-privacy_2026}) typically provide independent (L)DP primitives and do not explicitly support correlated perturbations. 
The essential reason is that the underlying data correlations are often unknown. 

When the correlation coefficients are known or estimable, 
independent LDP can be extended via (i) post-processing of perturbed data to enforce a given correlation structure, 
or (ii) adopting correlated-perturbation mechanisms (e.g.~\cite{DBLP:journals/corr/abs-2507-17516,DBLP:conf/satml/VithanaCCJ25}; although they are primarily tailored to mean estimation). 
Our utility quantification framework also applies to such correlated mechanisms
by replacing the CDF of independent mechanisms with the conditional CDF of correlated mechanisms.

\begin{figure}[t]
    \centering
    \begin{subfigure}[b]{0.45\linewidth}
        \centering
        \includegraphics[width=0.98\textwidth]{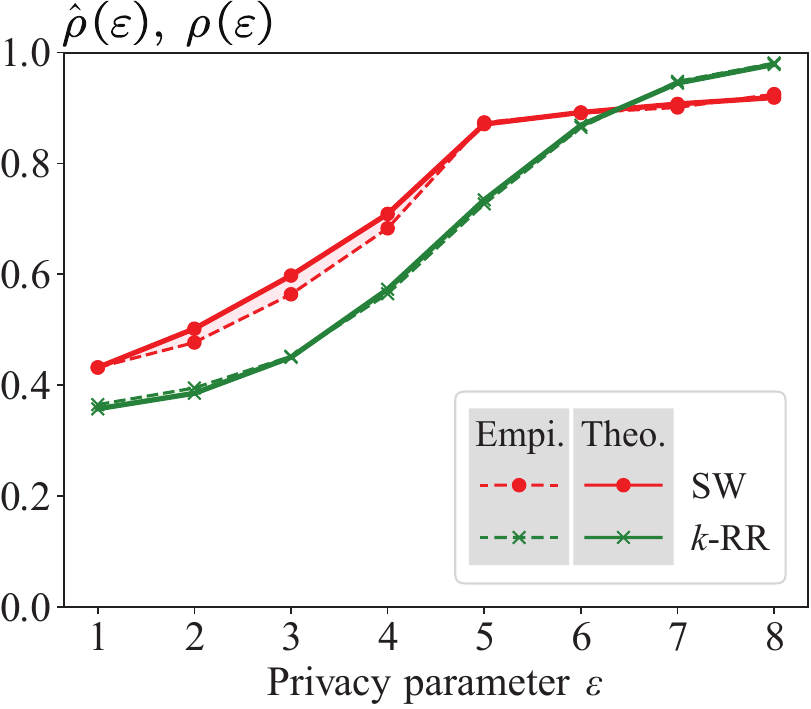}
        \caption{Logistic Regression.}
    \end{subfigure}
    \hfill
    \begin{subfigure}[b]{0.45\linewidth}
        \centering
        \includegraphics[width=0.98\textwidth]{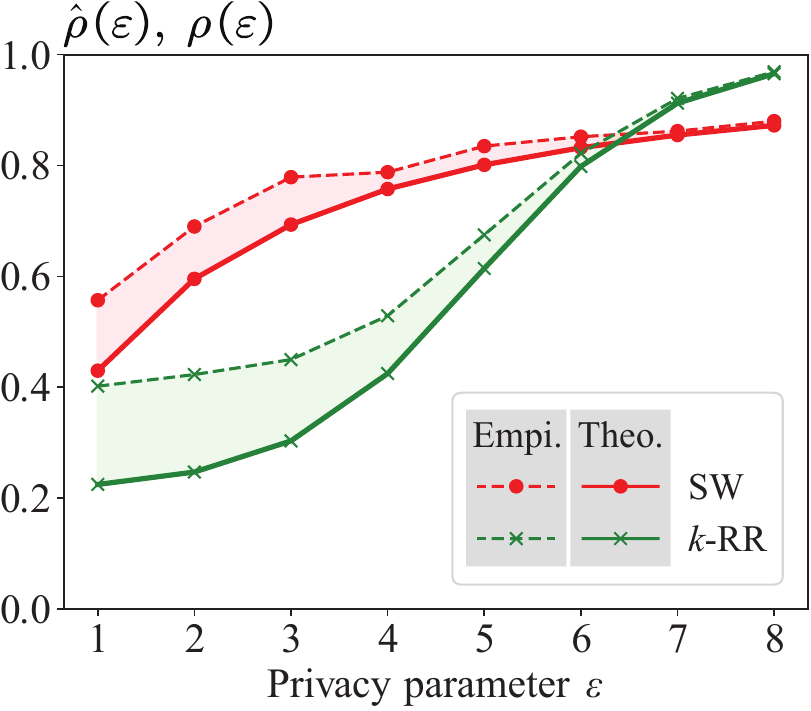}
        \caption{Random Forest.}
    \end{subfigure}
    \caption{Empirical and theoretical utility for two classifiers trained on the Stroke Prediction dataset.}
    \label{fig:appendix:other_records}
\end{figure}

\begin{figure}[t]
    \centering
    \begin{subfigure}[b]{0.45\linewidth}
        \centering
        \includegraphics[width=0.98\textwidth]{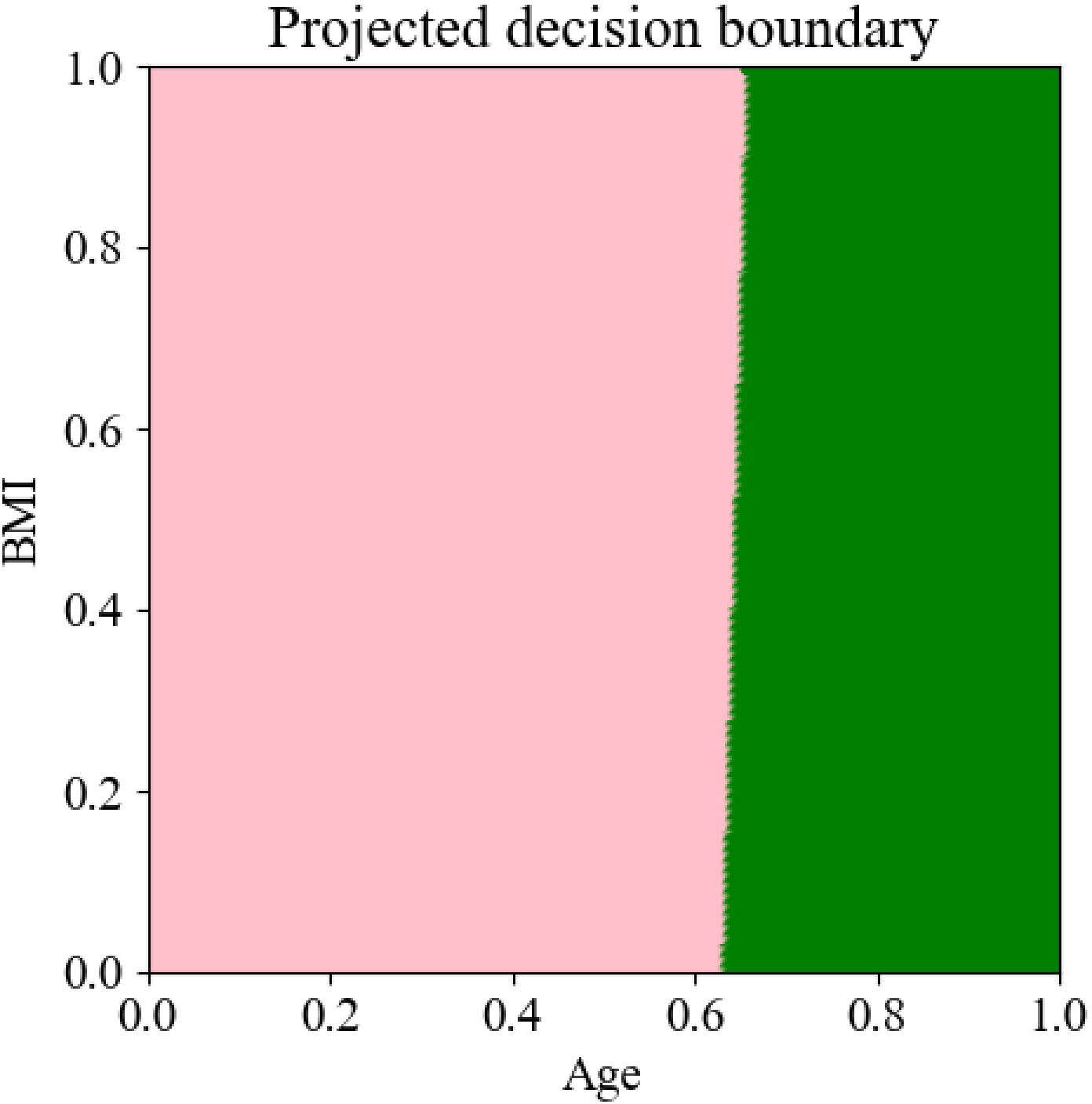}
        \caption{Logistic Regression (LR).}
        \label{fig:appendix:projected_decision_boundary_lr}
    \end{subfigure}
    \hfill
    \begin{subfigure}[b]{0.45\linewidth}
        \centering
        \includegraphics[width=0.98\textwidth]{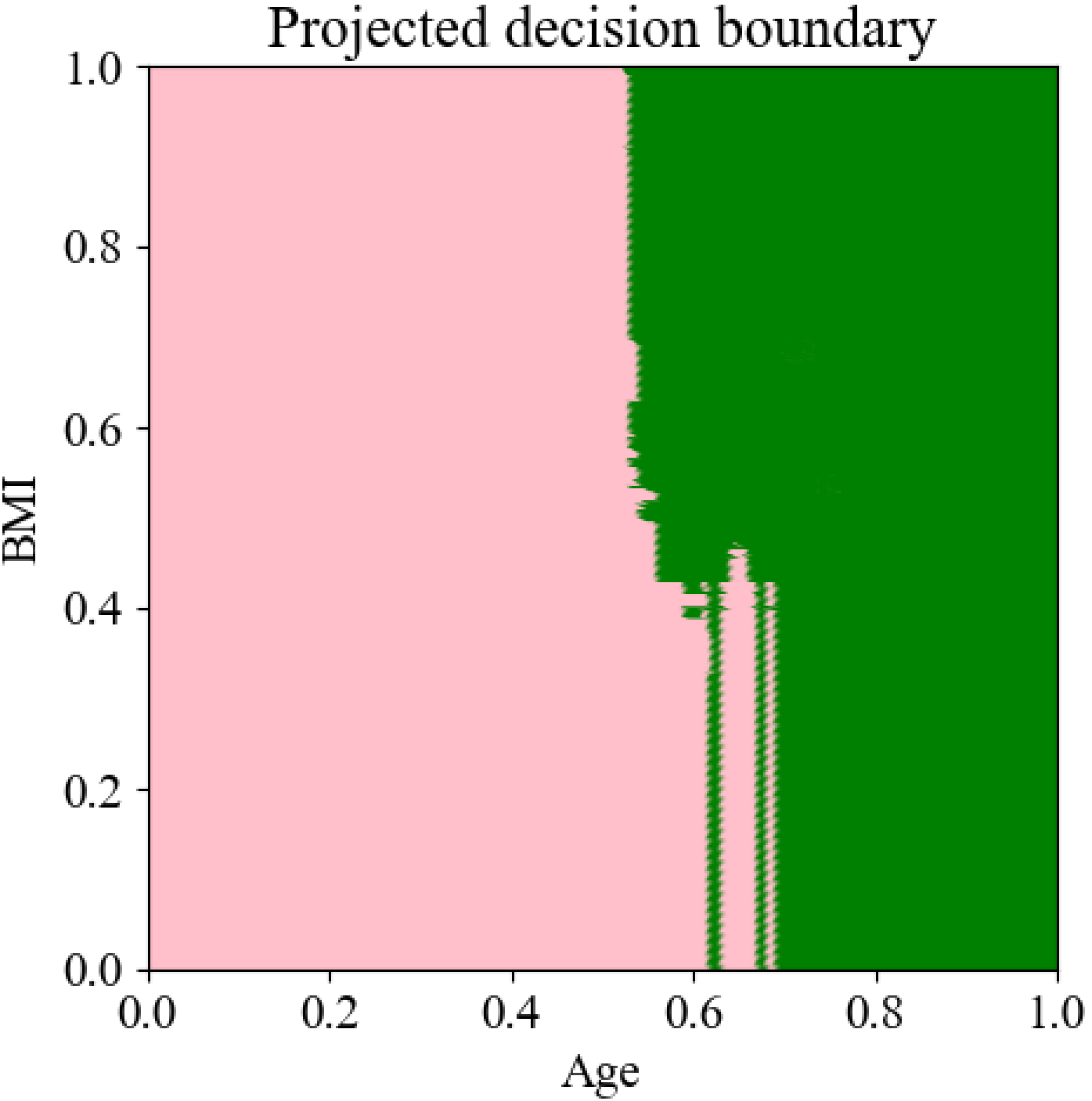}
        \caption{Random Forest (RF).}
        \label{fig:appendix:projected_decision_boundary_rf}
    \end{subfigure}
    \caption{Projected decision boundary of the two classifiers on the stroke dataset.}
    \label{fig:appendix:projected_decision_boundary}
\end{figure}

\subsection{More Case Studies for the Stroke Prediction Dataset (Section~\ref{subsubsec:stroke_prediction})} \label{appendix:more_examples}

This section provides a detailed utility analysis of the two classifiers
trained on the Stroke Prediction dataset.

\subsubsection{Other Records and Mechanisms} \label{appendix:stroke:other_records}
We also focus on the first record in the dataset, i.e. the record ``Age: 67, BMI: 36.6, Hypertension: 0, \dots'',
to evaluate the performance of other two LDP mechanisms: the SW mechanism and the $k$-RR mechanism.

Figure~\ref{fig:appendix:other_records} shows the theoretical and empirical utility of the two classifiers
under the SW and $k$-RR mechanisms.
The theoretical utility of both mechanisms is almost the same as the empirical utility
for the Logistic Regression classifier.
For the Random Forest classifier, the theoretical utility is a lower bound for the empirical utility,
especially when $\varepsilon$ is small.
\ Meanwhile, we can see that the SW mechanism does not always outperform the $k$-RR mechanism.
When $\varepsilon$ exceeds $6$, $k$-RR outperforms SW in both classifiers.

\subsubsection{Closed Form CDF of the PM Mechanism} \label{appendix:stroke:pm_cdf}

This section provides the closed form expression for $\rho(\varepsilon,\theta)$ when using the PM mechanism.

The input ``Age: 79'' corresponds to the normalized value $x = 0.79$ in domain $[0,1]$.
Then according to Definition~\ref{def:pm_mechanism} of the PM mechanism, we have
\begin{equation*}
    F_{\text{PM}_{\varepsilon,\text{Age}}}(0.63) = \int_{0}^{0.63} pdf[\mechanism(0.79) = \tilde{x}] \mathrm{d}\tilde{x}.
\end{equation*}
For piecewise-based mechanisms, we need to consider the range of $[l_{0.79,\varepsilon}, r_{0.79,\varepsilon}]$ relative to $0.63$.
For $\varepsilon$ that makes $0.63 \not\in [l_{0.79,\varepsilon}, r_{0.79,\varepsilon}]$,
the above integral is 
\begin{equation*}
    F_{\text{PM}_{\varepsilon,\text{Age}}}(0.63) = 0.63 \cdot \exp(-\frac{\varepsilon}{2}).
\end{equation*}
For $\varepsilon$ that makes $0.63 \in [l_{0.79,\varepsilon}, r_{0.79,\varepsilon}]$,
the above integral is
\begin{equation*}
    F_{\text{PM}_{\varepsilon,\text{Age}}}(0.63) = (0.79 - C) \cdot \exp(-\frac{\varepsilon}{2}) + (C - 0.16) \cdot \exp(\frac{\varepsilon}{2}),
\end{equation*}
where $C=(\exp(\varepsilon / 2) - 1) / (2\exp(\varepsilon) - 2)$, as define in Definition~\ref{def:pm_mechanism}.

\subsubsection{Projected Decision Boundary} \label{appendix:stroke:projected_decision_boundary}

The gap between the theoretical and empirical utility for the Random Forest classifier
comes from the complexity of its decision boundary.

Figure~\ref{fig:appendix:projected_decision_boundary} illustrates the projected decision boundaries of the two classifiers
on the ``Age'' and ``BMI'' features.
We can see that the decision boundary of the Random Forest classifier is significantly more complex than that of the Logistic Regression classifier,
i.e. a combination of many hyperrectangles, making it hard to be approximated by one hyperrectangle.
Consequently, the theoretical utility for the Random Forest classifier is conservative compared to its actual utility.
\ In contrast, the simpler decision boundary of the Logistic Regression classifier 
can be well approximated by hyperrectangles, leading to a precise theoretical utility.

\begin{figure}[t]
    \centering
    \begin{subfigure}[b]{0.45\linewidth}
        \centering
        \includegraphics[width=0.98\textwidth]{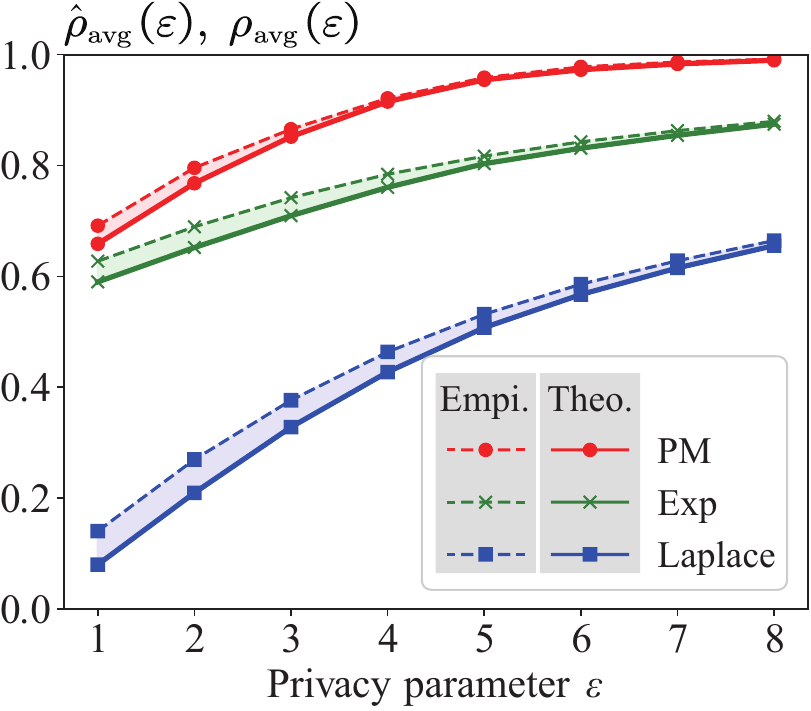}
        \caption{LR: Average-case utility.}
        \label{fig:exp:stroke_lr_avg}
    \end{subfigure}
    \hfill
    \begin{subfigure}[b]{0.45\linewidth}
        \centering
        \includegraphics[width=0.98\textwidth]{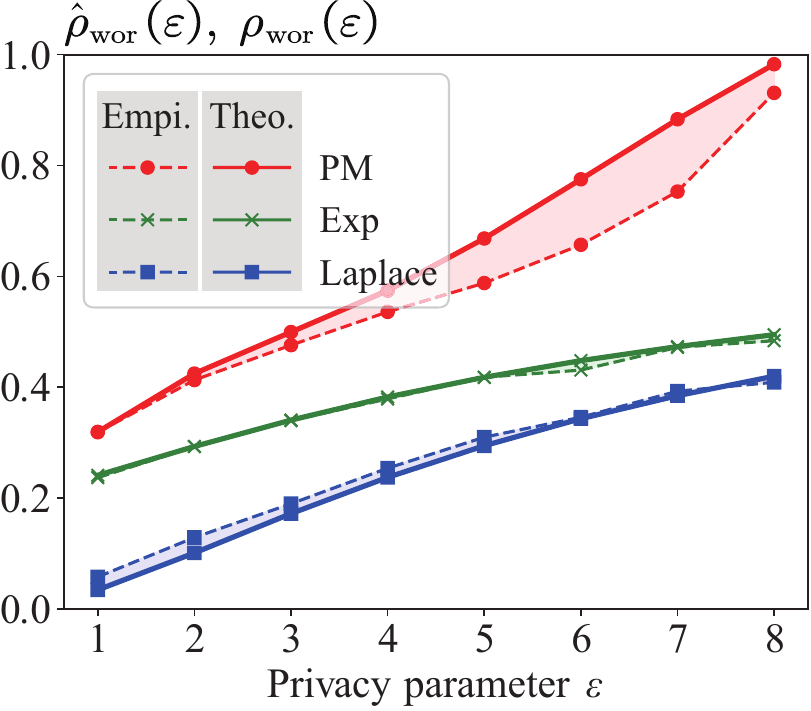}
        \caption{LR: Worst-case utility.}
        \label{fig:exp:stroke_lr_wor}
    \end{subfigure}
    \caption{Average-case and worst-case utility for the Logistic Regression classifier trained on the Stroke Prediction dataset.}
    \label{fig:exp:stroke_lr_avg_wor}
\end{figure}

\begin{figure}[t]
    \centering
    \begin{subfigure}[b]{0.45\linewidth}
        \centering
        \includegraphics[width=0.98\textwidth]{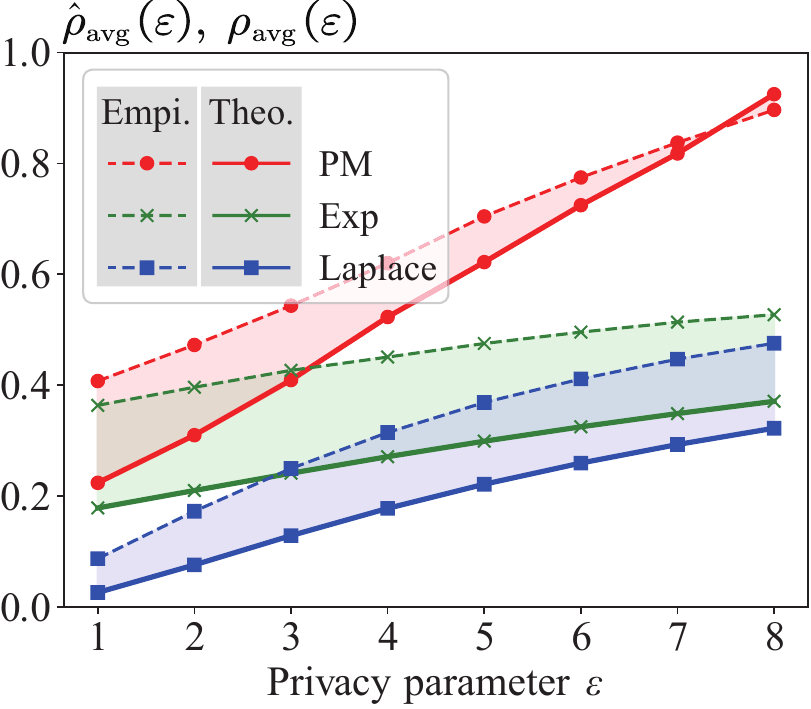}
        \caption{RF: Average-case utility.}
        \label{fig:exp:stroke_rf_avg}
    \end{subfigure}
    \hfill
    \begin{subfigure}[b]{0.45\linewidth}
        \centering
        \includegraphics[width=0.98\textwidth]{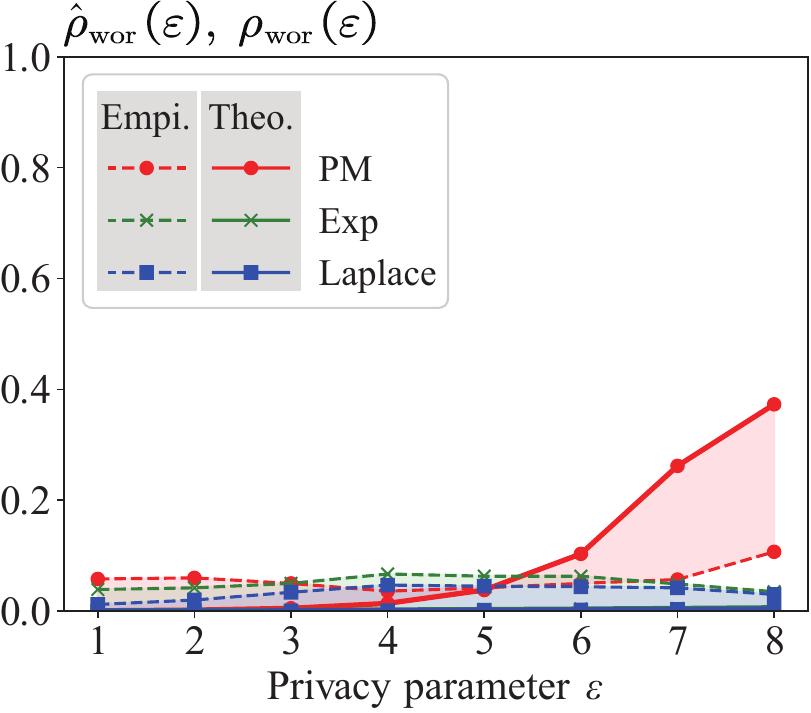}
        \caption{RF: Worst-case utility.}
        \label{fig:exp:stroke_rf_wor}
    \end{subfigure}
    \caption{Average-case and worst-case utility for the Random Forest classifier trained on the Stroke Prediction dataset.}
    \label{fig:exp:stroke_rf_avg_wor}
\end{figure}

\subsubsection{Average-case and Worst-case Utility} \label{appendix:stroke:average_worst_case}

Figure~\ref{fig:exp:stroke_lr_avg_wor} shows the average-case and worst-case utility of the Logistic Regression classifier under three LDP mechanisms. 
For the average-case utility, the trends closely mirror those observed in the point-wise utility quantification presented in the main text.
For the worst-case utility, the curves are less regular but still exhibit the same overall ordering among mechanisms: 
the PM mechanism consistently achieves the highest utility, followed by the Exponential mechanism, and then the Laplace mechanism.
In the worst-case scenario, we observe that the theoretical utility is not always lower than the empirical utility (though they are close), 
especially for the PM mechanism when $\varepsilon$ is large. 
This occurs because the worst-case robustness hyperrectangle is typically two-dimensional and small,
which amplifies the impact of sampling errors in empirical evaluation.
In such cases, the actual utility is often higher than the empirical utility.

Figure~\ref{fig:exp:stroke_rf_avg_wor} shows the average-case and worst-case utility of the Random Forest classifier.
The worst-case utility for the Random Forest classifier is significantly smaller than that of the Logistic Regression classifier.
This is due to the complex decision boundary of the Random Forest classifier, as shown in Figure~\ref{fig:appendix:projected_decision_boundary_rf}.
Around $\text{Age, BMI} \approx \{0.65, 0.3\}$, the robustness hyperrectangle is significantly smaller than that of the Logistic Regression classifier,
leading to both small theoretical and empirical utilities.

\textbf{Conclusion on Robustness.}
From the results of average-case and worst-case utility analysis of the Logistic Regression and Random Forest classifiers,
we can conclude that the Logistic Regression classifier is more robust under LDP-perturbed inputs than the Random Forest classifier.

\subsection{More Case Studies for the Bank Customer Attrition Dataset (Section~\ref{subsubsec:bank_attrition})} \label{appendix:more_bank_attrition}

\subsubsection{Detailed Quantification} \label{appendix:detailed_quantification}

We take the Logistic Regression classifier as an example to show the detailed utility quantification under the PAC LDP.

The robustness hyperrectangle at this record is $\theta_\diamond = [0, 0.72]\times [0,1]$. 
Using this information, we can provide the quantification of the utility at the record under the PM mechanism.

\vspace{0.5em}
\noindent\emph{For this trained Logistic Regression classifier on the Bank Customer Attrition dataset,
at the record ``Age: 22, Estimated Salary: 101\,348, Credit Score: 619, \dots'', with probability at least $\rho(\varepsilon, \theta_\diamond)$, 
the classifier preserves the correct prediction under the privacy indicator $\mathcal{I}_{0.1}$ and the PM mechanism applied to 
``Age'' and ``Estimated Salary'' (i.e. $(2\varepsilon, 0.1)$- PAC LDP), where
\begin{equation*}
    \resizebox{\linewidth}{!}{$
    \begin{aligned}
        \rho(\varepsilon,\theta) =& [F_{\text{PM}_{\varepsilon,\text{Age}}}(0.72) - F_{\text{PM}_{\varepsilon,\text{Age}}}(0)][F_{\text{PM}_{\varepsilon,\text{Sal}}}(1) - F_{\text{PM}_{\varepsilon,\text{Sal}}}(0)] \\
        =& F_{\text{PM}_{\varepsilon,\text{Age}}}(0.72).
    \end{aligned}$}
\end{equation*}
}

For the Gaussian mechanism, the utility quantification is similar but uses the Gaussian CDF.

\vspace{0.5em}
\noindent\emph{For this trained Logistic Regression classifier on the Bank Customer Attrition dataset,
at the record ``Age: 22, Estimated Salary: 101\,348, Credit Score: 619, \dots'', with probability at least $\rho(\varepsilon, \theta_\diamond)$, 
the classifier preserves the correct prediction under the Gaussian mechanism (Theorem~\ref{thm:gaussian_mechanism}) applied to 
``Age'' and ``Estimated Salary'' (i.e. $(2\varepsilon, 0.1)$- PAC LDP), where
\begin{equation*}
    \resizebox{\linewidth}{!}{$
    \begin{aligned}
        \rho(\varepsilon,\theta) =& [F_{\text{Gau}_{\varepsilon,\text{Age}}}(0.72) - F_{\text{Gau}_{\varepsilon,\text{Age}}}(0)][F_{\text{Gau}_{\varepsilon,\text{Sal}}}(1) - F_{\text{Gau}_{\varepsilon,\text{Sal}}}(0)] \\
        =& F_{\text{Gau}_{\varepsilon,\text{Age}}}(0.72) - F_{\text{Gau}_{\varepsilon,\text{Age}}}(0).
    \end{aligned}$}
\end{equation*}
}

\begin{figure}[t]
    \centering
    \begin{subfigure}[b]{0.45\linewidth}
        \centering
        \includegraphics[width=0.98\textwidth]{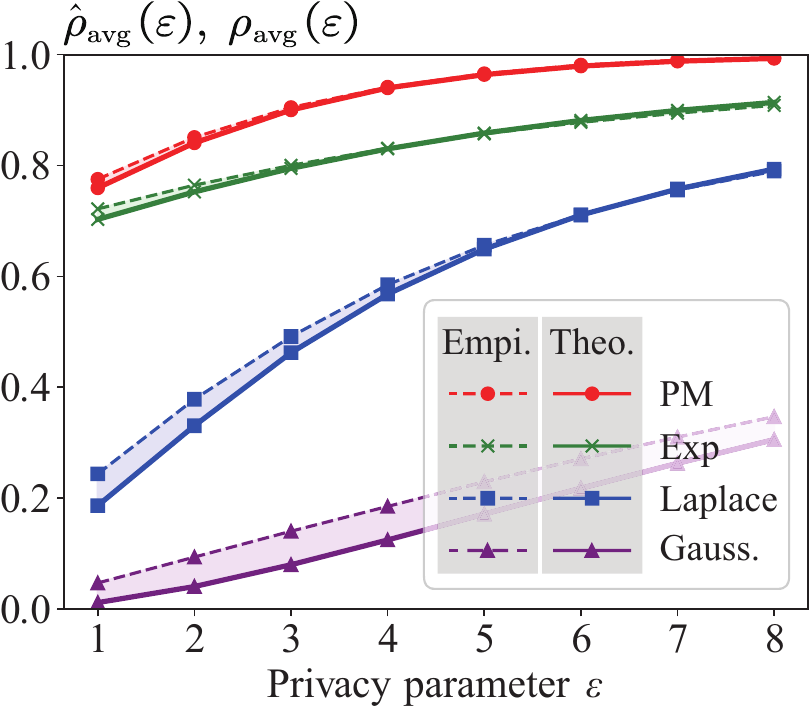}
        \caption{LR: Average-case utility.}
        \label{fig:exp:bank_lr_avg}
    \end{subfigure}
    \hfill
    \begin{subfigure}[b]{0.45\linewidth}
        \centering
        \includegraphics[width=0.98\textwidth]{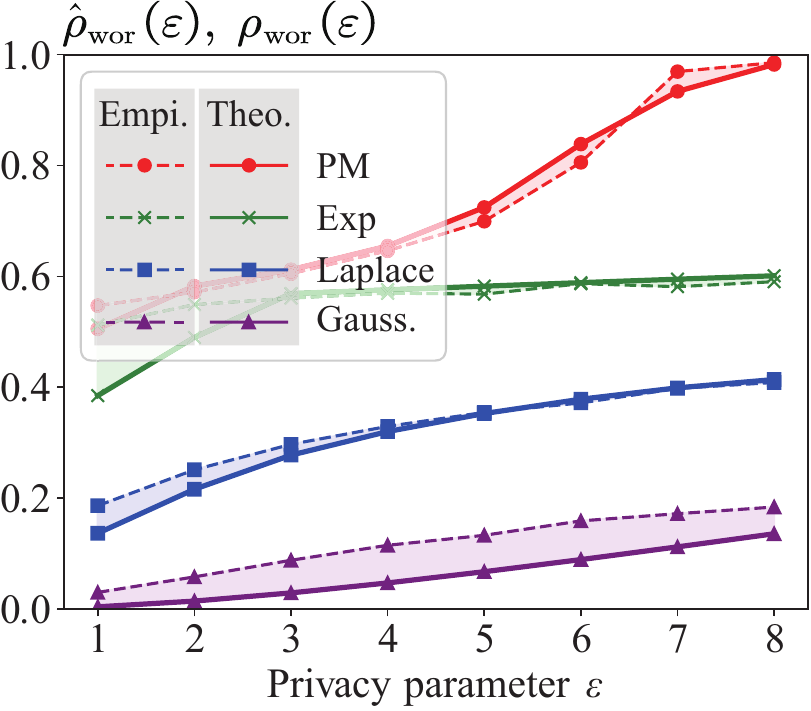}
        \caption{LR: Worst-case utility.}
        \label{fig:exp:bank_lr_wor}
    \end{subfigure}
    \caption{Average-case and worst-case utility for the Logistic Regression classifier trained on the Bank Customer Attrition dataset.}
    \label{fig:exp:bank_lr_avg_wor}
\end{figure}

\begin{figure}[t]
    \centering
    \begin{subfigure}[b]{0.45\linewidth}
        \centering
        \includegraphics[width=0.98\textwidth]{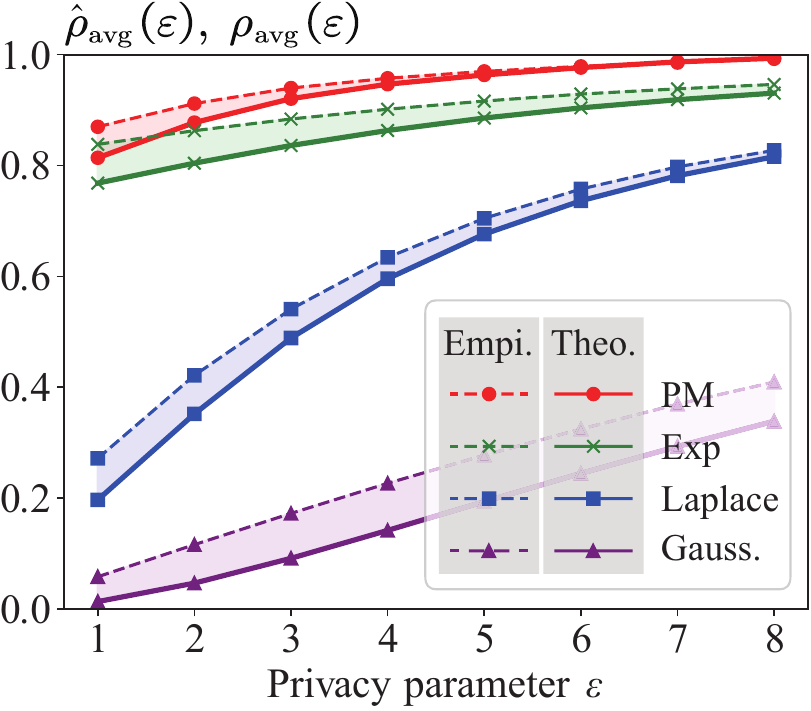}
        \caption{RF: Average-case utility.}
        \label{fig:exp:bank_rf_avg}
    \end{subfigure}
    \hfill
    \begin{subfigure}[b]{0.45\linewidth}
        \centering
        \includegraphics[width=0.98\textwidth]{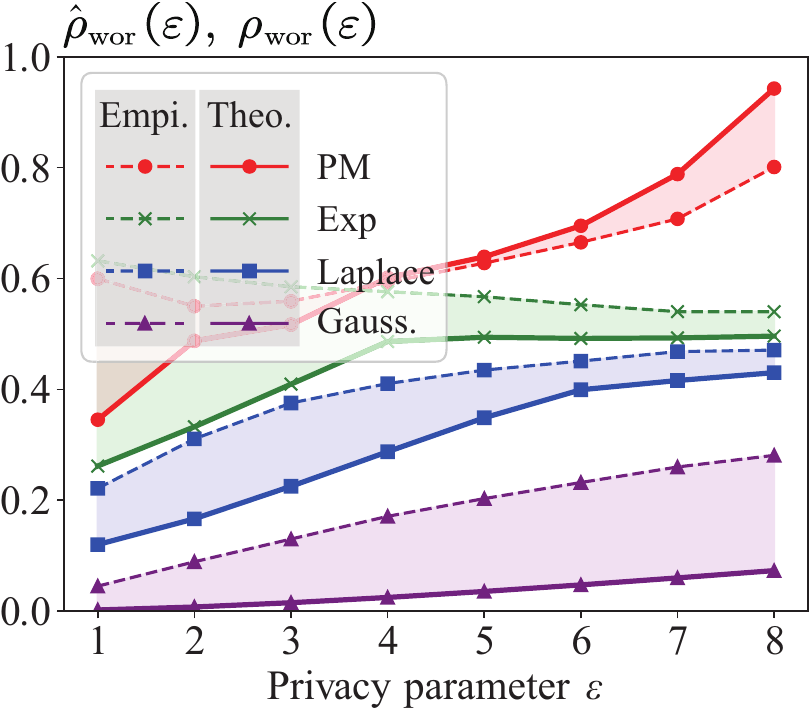}
        \caption{RF: Worst-case utility.}
        \label{fig:exp:bank_rf_wor}
    \end{subfigure}
    \caption{Average-case and worst-case utility for the Random Forest classifier trained on the Bank Customer Attrition dataset.}
    \label{fig:exp:bank_rf_avg_wor}
\end{figure}

\subsubsection{Average-case and Worst-case Utility} \label{appendix:bank:average_worst_case}

Figure~\ref{fig:exp:bank_lr_avg_wor} and Figure~\ref{fig:exp:bank_rf_avg_wor} summarize the average-case and worst-case utility for classifiers trained on the Bank Customer Attrition dataset.
The observed trends are consistent with those from the Stroke Prediction dataset:
(i) The theoretical utility closely matches the empirical utility, especially in the average-case scenario.
(ii) The PM mechanism consistently achieves the highest utility, followed by the Exponential mechanism, then the Laplace mechanism, and finally the Gaussian mechanism.

\subsection{More Details for the MNIST Classifier (Section~\ref{subsubsec:mnist})} \label{appendix:mnist}

\subsubsection{Monte Carlo Estimation of the Utility} \label{appendix:mnist:improvement}
We emphasize the conceptual connection between LDP and robustness while adopting robustness notions that largely align with the standard \emph{local robustness} used in prior work. 
Although these notions are clean and well established, they can yield conservative utility estimates in high-dimensional settings.

To obtain tighter utility quantification in high dimensions, one can consider more flexible robustness notions, e.g. robustness regions given by irregular sets rather than axis-aligned hyperrectangles (as induced by complex decision boundaries). 
Such regions can be estimated via Monte Carlo sampling in high-dimensional spaces. 
This estimation is performed once; afterward, for any given $\varepsilon$, the utility can be theoretically quantified by approximately integrating the $d$-dimensional LDP mechanism's probability distribution over the estimated robustness region.

Formally, let $h:[0,1]^d \to \{1,\ldots,K\}$ be a $d$-dimensional classifier and let $x\in[0,1]^d$ be an input with predicted (and correct) label $h(x)$.
Draw $\tilde{x}_1,\ldots,\tilde{x}_N$ i.i.d. uniformly from $[0,1]^d$, and for each sample check whether $h(\tilde{x}_i)=h(x)$ holds.
Assume that $m$ of these samples satisfy the check, and let $\mathcal{S}\subseteq[0,1]^d$ denote an estimated robustness region that contains the accepted samples.
Under a $d$-dimensional (continuous) LDP mechanism,\footnote{
    For discrete LDP mechanisms, the utility can be computed directly by summing the probabilities that the mechanism outputs samples in $\mathcal{S}$, eliminating the need for volume estimation.
    We omit this discussion for simplicity.
}
the utility can then be quantified as
\begin{equation*}
    \begin{split}
        \rho(\varepsilon, \mathcal{S}) =& \int_{\mathcal{S}} pdf[\mechanism_\varepsilon(x) = \tilde{x}] \mathrm{d}\tilde{x} \\
        \approx & \mathrm{Vol}(\mathcal{S}) \cdot \frac{1}{m} \sum_{i=1}^m pdf[\mechanism_\varepsilon(x) = \tilde{x}_i],
    \end{split}
\end{equation*}
where $\mathrm{Vol}(\mathcal{S})$ is the volume of the region $\mathcal{S}$, which can be estimated by $\mathrm{Vol}([0,1]^d) \cdot m / N$.
It is still a theoretical measure, since samples are fixed and $pdf[\mechanism_\varepsilon(x) = \tilde{x}_i]$
is determined analytically by the LDP mechanism and can be evaluated for any $\varepsilon$.

This approach follows the connection between concentration analysis and robustness analysis established in this paper, 
with the concentration analysis now performed over the estimated robustness region $\mathcal{S}$ rather than the hyperrectangle $\theta_\diamond$.
The error is now governed by the Monte Carlo estimation error (i.e. $\mathcal{O}(1/\sqrt{N})$)
and volume estimation error, instead of Hoeffding bound in Theorem~\ref{thm:hoeffding}.

We experimented with this approach on the MNIST-7$\times$7 dataset. 
While it is theoretically sound, in such a high-dimensional space the samples that pass the uniform-sampling check are sparse. 
As a result, for most accepted samples $\tilde{x}_i$, the density $pdf[\mechanism_\varepsilon(x)=\tilde{x}_i]$ is close to zero, 
which drives the estimated utility to an unrealistically small value.
Moreover, we cannot directly apply importance sampling using $pdf[\mechanism_\varepsilon(x)]$ as the proposal distribution, 
since this proposal distribution depends on $\varepsilon$ and would essentially reduce the procedure to empirical utility estimation.

\subsubsection{Average-case and Worst-case Utility} \label{appendix:mnist:average_worst_case}

\begin{figure}[t]
    \centering
    \begin{subfigure}[b]{0.45\linewidth}
        \centering
        \includegraphics[width=0.98\textwidth]{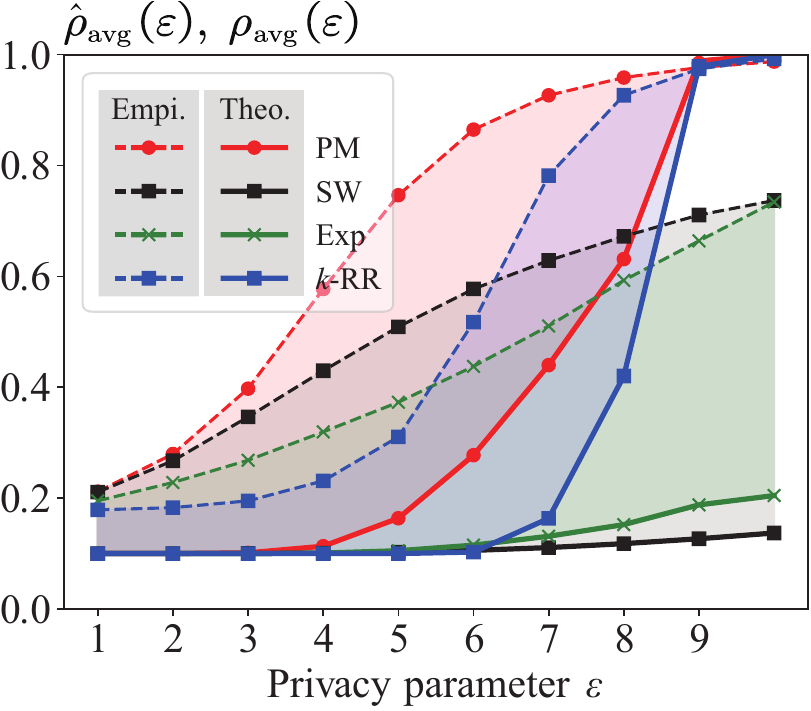}
        \caption{NN: Average-case utility.}
        \label{fig:exp:mnist_cnn_avg}
    \end{subfigure}
    \hfill
    \begin{subfigure}[b]{0.45\linewidth}
        \centering
        \includegraphics[width=0.98\textwidth]{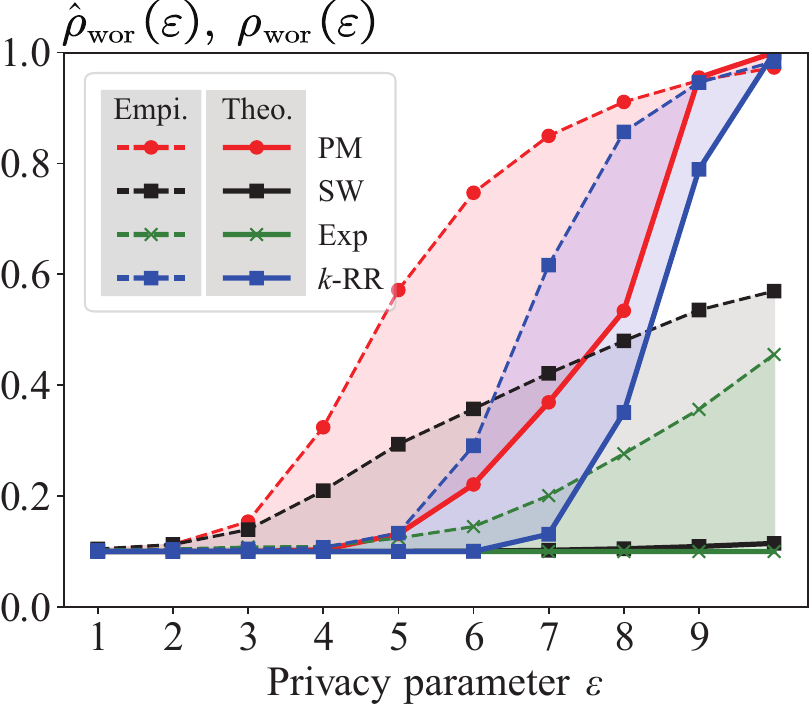}
        \caption{NN: Worst-case utility.}
        \label{fig:exp:mnist_cnn_wor}
    \end{subfigure}
    \caption{Average-case and worst-case utility for the Neural Network classifier trained on the MNIST-7$\times$7 dataset.}
    \label{fig:appendix:mnist_cnn_avg_wor}
\end{figure}

We use the first 50 images from the MNIST-7$\times$7 dataset as representative samples for analyzing the average-case and worst-case utility under LDP-perturbed inputs.
Figure~\ref{fig:appendix:mnist_cnn_avg_wor} presents the results of this analysis.
Both the average-case and worst-case utility show a similar trend to the point-wise utility quantification, with the PM mechanism consistently achieving the highest utility.

\end{document}